\documentclass{ieeeaccess}
\usepackage{cite}
\usepackage{amsmath,amssymb,amsfonts,amsthm}
\usepackage{algorithmic}
\usepackage{graphicx}
\usepackage{textcomp}
\usepackage{comment}
\usepackage{bm}

\usepackage{braket}
\usepackage{textalpha}
\usepackage{url}
\usepackage{multirow} %
\usepackage{array}
\usepackage{multirow}
\usepackage{graphicx} %

\usepackage{xcolor}
\usepackage{soul}

\sethlcolor{white}

\soulregister\cite7
\soulregister\ref7
\soulregister\eqref7

\definecolor{greenhl}{RGB}{255,255,255}
\definecolor{bluehl}{RGB}{255,255,255}
\definecolor{redhl}{RGB}{255,255,255}

\newcommand{\hlgreen}[1]{%
    {\sethlcolor{greenhl}\hl{#1}}%
}

\newcommand{\hlred}[1]{%
    {\sethlcolor{redhl}\hl{#1}}%
}

\newif\ifhighlightmath
\highlightmathtrue

\newcommand{\highlight}[1]{\ifhighlightmath\colorbox{white}{$#1$}\else#1\fi}

\theoremstyle{definition}
\newtheorem{lemma}{Lemma}

\makeatletter
\AtBeginDocument{\DeclareMathVersion{bold}
\SetSymbolFont{operators}{bold}{T1}{times}{b}{n}
\SetSymbolFont{NewLetters}{bold}{T1}{times}{b}{it}
\SetMathAlphabet{\mathrm}{bold}{T1}{times}{b}{n}
\SetMathAlphabet{\mathit}{bold}{T1}{times}{b}{it}
\SetMathAlphabet{\mathbf}{bold}{T1}{times}{b}{n}
\SetMathAlphabet{\mathtt}{bold}{OT1}{pcr}{b}{n}
\SetSymbolFont{symbols}{bold}{OMS}{cmsy}{b}{n}
\renewcommand\boldmath{\@nomath\boldmath\mathversion{bold}}}
\makeatother

\def\BibTeX{{\rm B\kern-.05em{\sc i\kern-.025em b}\kern-.08em
    T\kern-.1667em\lower.7ex\hbox{E}\kern-.125emX}}

\begin{document}
\history{This work has been accepted for publication in IEEE Access. This is the authors' accepted manuscript; the final published version is available on IEEE Xplore.}
\doi{to be assigned}

\title{Error Mitigation by Zero-Noise Extrapolation (ZNE) in Indirect-Control System}
\author{\uppercase{Arijit Das}\authorrefmark{1}, %
\uppercase{Masaki Owari}\authorrefmark{2}
}

\address[1]{Department of Information Science and Technology, Graduate School of Science and Technology (GSST), Shizuoka University, Hamamatsu, 432-8011, Shizuoka, Japan (e-mail: a.das.23@shizuoka.ac.jp)}
\address[2]{Department of Information Science and Technology, Graduate School of Science and Technology (GSST), Shizuoka University, Hamamatsu, 432-8011, Shizuoka, Japan (e-mail: masakiowari@inf.shizuoka.ac.jp)}
\tfootnote{\hl{This work was supported by JSPS KAKENHI Grant Number JP24H00952.}}

\markboth
{A. Das and M. Owari: Error Mitigation by ZNE in Indirect-Control System}
{A. Das and M. Owari: Error Mitigation by ZNE in Indirect-Control System}

\corresp{Corresponding author: Masaki Owari (e-mail: masakiowari@inf.shizuoka.ac.jp).}

\begin{abstract}
The indirect-control paradigm offers a scalable approach to quantum computation by restricting external operations to a small subset of qubits, whereas the \hl{entire} system evolves under its intrinsic many-body Hamiltonian. This framework reduces hardware complexity and limits the noise introduced through control channels. Recently, Anan \emph{et al.} proposed an implementation of the \hl{variational quantum eigensolver} (VQE) within this indirect-control architecture. However, the impact and integration of quantum error mitigation techniques in such a setting remain largely unexplored. In particular, zero-noise extrapolation (ZNE), a widely used error mitigation strategy, has not yet been systematically adapted to indirect-control VQE. A key requirement of ZNE is the controlled amplification of noise, which is commonly achieved through circuit-level techniques such as unitary folding. In the indirect-control framework, however, quantum operations are largely governed by continuous time-evolution under the system Hamiltonian, making the construction of physically realizable inverse operations non-trivial. In this study, we investigate the feasibility of implementing ZNE within the indirect-control VQE framework. \hl{We address the challenges associated with unitary folding in the presence of time-evolution operators and, with a minimal relaxation of the indirect-control constraint, we propose a ZNE-based error-mitigation strategy that exploits the symmetry of the one-dimensional XY spin chain.}
\end{abstract}


\titlepgskip=-21pt

\maketitle

\section{Introduction}
\label{sec:introduction}

\PARstart{T}{he} early conceptualization of quantum computing dates back to the 1980s \cite{Benioff1980, Feynman1982,Deutsch85}. In the 1990s, the development of quantum algorithms, such as Shor's algorithm for integer factorization \cite{Shor_1997} and Grover’s algorithm for unstructured search problems \cite{grover1996fastquantummechanicalalgorithm}, demonstrated that quantum computers could have a profound impact on practical applications and industry. Initially, it was widely believed that the practical realization of quantum computers would be infeasible due to noise arising from interactions between qubits and their environment \cite{Landauer95}. However, the introduction of quantum error-correcting codes (QEC) \cite{Shor95,Steane96,Calderbank96,Steane96b,Gottesman96,Kitaev97} established that, in principle, fault-tolerant quantum computation (FTQC)\hl{,} which enables arbitrarily long quantum computations by suppressing the effects of noise\hl{,} is achievable \cite{Shor96b,Gottesman98,Preskill98,Steane1999,Aharanov08,Nielsen_Chuang_2010,Brun_2020}. Consequently, experimental efforts towards realizing practical quantum computers have accelerated since the late 1990s. In particular, technological advances in superconducting quantum circuits since the 2000s \cite{Blais04,Schuster2007,Blais07,Koch07,Schreier08,Lucero2012,Barends2014} have enabled the development of quantum processors with hundreds of qubits by research groups such as IBM and Google \cite{Kim2023,CarreraVazquez2024,Acharya2024}. Furthermore, some of these devices have demonstrated quantum supremacy \cite{quantumsup,PhysRevLett.127.180501}. Notably, a recently developed 105-qubit superconducting processor by Google has achieved quantum error correction below the fault-tolerance threshold \cite{Acharya2024}. Despite these significant advances, the realization of practical FTQC remains challenging, as encoding a single logical qubit requires a large number of physical qubits. Consequently, current noisy intermediate-scale quantum (NISQ) devices are not yet capable of implementing FTQC \cite{Preskill2018,nisq}.

Conventional quantum algorithms such as Shor's algorithm and Grover's algorithm are not executable on NISQ devices. For this reason, significant \hl{efforts have} been devoted to developing quantum algorithms that can be implemented on NISQ devices. Many of these algorithms are quantum-classical hybrid algorithms that utilize parameterized quantum circuits \cite{Endo_2021,vqa,hegde2024buzzstrategicpathsenabling}, such as the \hl{variational quantum eigensolver} (VQE) \cite{Peruzzo_2014,Tilly_2022}, which is used to find the minimum eigenvalue of an observable in quantum systems. Many of these NISQ algorithms use the expectation value of an observable as their output. This expectation value is obtained by performing statistical processing on the measurement outcomes using a classical computer. Therefore, the bias introduced by noise directly contributes to errors in the output of the algorithm. Quantum error mitigation (QEM) \cite{PhysRevX.7.021050,PhysRevLett.119.180509,Endo_2021,Cai_2023}, which has been explored in recent years as an alternative to QEC, aims to reduce this bias by performing classical post-processing on the measurement data.

While QEC remains the ultimate goal for realizing FTQC, QEM serves as an interim strategy for mitigating errors in current devices. Various QEM techniques have been developed, including zero-noise extrapolation (ZNE) \cite{PhysRevX.7.021050,PhysRevLett.119.180509}, probabilistic error cancellation \cite{PhysRevLett.119.180509}, and Clifford regression \cite{Czarnik2021}, among others. Despite current hardware limitations, QEM has demonstrated significant improvements in the performance of NISQ algorithms \cite{Endo_2021}.

Another approach to mitigating noise is to reduce the number of pathways through which noise can enter the quantum device from the external environment. In current NISQ devices, the entire quantum system is {\em controlled} via direct access to individual qubits. However, such external access also introduces additional noise into the system. Therefore, reducing the number of externally accessed qubits can, in principle, lead to a reduction in noise. It is known that, in a typical quantum many-body system, universal quantum computation on the entire system is possible by accessing only a small number of qubits (control units) \cite{Lloyd04,Burgarth07,Burgarth09,Kay10,Burgarth10,Heule10,Maruyama2016,Gokler17,KATO2020168046}. 

For example, it has been shown that any unitary transformation on a spin chain of arbitrary length can be implemented by accessing only the first qubit in the case of the 1D Heisenberg model \cite{Burgarth10}, and by accessing only the first two qubits in the case of the 1D XY model \cite{Burgarth09}. This method of controlling the entire quantum system through a combination of the free time-evolution of the many-body Hamiltonian and quantum operations on a small control unit is referred to as {\em indirect-control} \cite{Maruyama2016} (Figure \ref{indirect-spinchain}). Indirect-control also simplifies the overall architecture of quantum computing systems. In conventional quantum devices, the number of control mechanisms typically scales with the number of qubits \cite{indirect}.

\Figure[t!](topskip=0pt, botskip=0pt, midskip=0pt){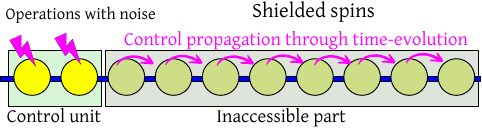}
{Indirect-control architecture in which the first two qubits of the spin chain form the control subsystem, while the remaining qubits are not directly accessible. By applying quantum operations on the control subsystem and allowing the entire system to evolve under its natural (drift) Hamiltonian, effective control over the full system can be achieved.}\label{indirect-spinchain}

When indirect-control is used, we need to repeatedly apply various quantum operations on the control unit at specific timings to compute the desired unitary to the whole system. 
In other words, for a given unitary transformation on the whole system, it is necessary to calculate the pulse sequence that needs to be applied to the control unit to realize it. For the 1D XY model, Burgarth \textit{et al.} showed that this calculation can be performed efficiently \cite{Burgarth10}. However, for general quantum many-body systems, it remains unclear whether an efficient method exists for calculating such pulse sequences. Furthermore, another problem is that the control pulse depends on the details of the many-body Hamiltonian. In other words, it is generally not possible to determine the control pulse without knowing all the coupling constants of the many-body Hamiltonian. In response to this problem,  it is possible to efficiently estimate the coupling constants using indirect-control alone under several assumptions on prior knowledge of a spin model Hamiltonian \cite{Burgarth09b,Franco09,Burgarth09c,Kato14,Zhang14,Sone17}. 
Even in the absence of prior knowledge of the Hamiltonian, its identification is, in principle, possible as long as the system is completely controllable \cite{Burgarth12}. 
Further, even if the system is not completely controllable,
in principle, any many-body Hamiltonian can be identified using indirect-control if FTQC can be used as an ancilla system \cite{Owari15}. However, the analysis on \cite{Burgarth12, Owari15} does not show the efficiency of the proposed methods, and no practically efficient method of identification of general Hamiltonian  has been proposed.

Recently, Anan \textit{et al.} showed that the VQE can be implemented using indirect-control \cite{anan_indirect_vqe}. In VQE, the expectation value of a target observable for the output state of a parameterized quantum circuit is optimized using a classical computer. For efficient optimization, it is important to obtain the partial derivatives of the expectation value of the observable with respect to the circuit parameters. To implement VQE using indirect-control, it is therefore necessary to compute these partial derivatives within the indirect-control framework. Anan \textit{et al.} demonstrated that the partial derivatives with respect to circuit parameters can be evaluated using indirect-control by treating the times at which the control unit performs gate operations as circuit parameters. In addition, they showed that VQE implemented using this method for a typical one-dimensional spin chain exhibits performance comparable to that of VQE implemented using conventional methods. 

In this indirect-control-based VQE approach, finding the optimal control pulse is equivalent to determining the optimal circuit parameters, which is achieved automatically through the parameter optimization in VQE. Furthermore, this optimization can be carried out even when the coupling coefficients of the many-body Hamiltonian are unknown. Therefore, in the VQE implementation method of Anan \textit{et al.}, there is no need to explicitly determine the control pulses or the coupling coefficients, unlike in conventional indirect-control methods.

In this way, the work of Anan \textit{et al.} suggests that indirect-control systems provide a promising alternative architecture for NISQ devices. However, several challenges remain in the indirect-control approach. In particular, it is necessary to develop methods for dealing with noise under indirect-control. In quantum error correction, the entropy that increases due to decoherence must be removed from the system through repeated syndrome measurements. However, because indirect-control methods have limited access to the system, it is not straightforward to efficiently remove this entropy. For this reason, it appears fundamentally difficult to implement quantum error correction using indirect-control. 

On the other hand, in the case of error mitigation, there is no requirement to remove such entropy, and thus it may still be feasible to implement it within the indirect-control framework. Therefore, in this paper, we investigate the feasibility of implementing error mitigation using indirect-control.

We show that, in the case of VQE using indirect-control on a one-dimensional XY spin chain, error mitigation can be achieved using the zero-noise extrapolation method, provided that, in addition to unitary operations on the control unit, $Y$-gate operations are allowed on the odd-numbered qubits, including those in the inaccessible region.

\section{Variational Eigensolver and Indirect-control}\label{sec2}
This section provides an overview of the indirect-control VQE.

\subsection{Variational quantum eigensolver (VQE)}
Due to the inherent limitations of NISQ devices, hybrid algorithms that integrate quantum and classical processing have emerged as a practical approach \cite{hegde2024buzzstrategicpathsenabling}. Among these, the Variational Quantum Eigensolver (VQE) \cite{Tilly_2022}, first proposed by Peruzzo {\em et al.} \cite{Peruzzo_2014}, is recognized as one of the most promising hybrid quantum-classical algorithms tailored for the NISQ era. The primary aim of VQE is to determine the ground-state energy of a quantum system, a fundamental step in exploring the energetic properties and the electronic configuration on the ground state of molecular systems, which has application in quantum chemistry \cite{Cao2019,RevModPhys.92.015003} and condensed matter physics \cite{PhysRevA.92.042303, Kokail2019}. 

The application of VQE to the Hamiltonian of a many-electron system is carried out as follows \cite{Tilly_2022}: The Hamiltonian of a many-electron system can be expanded in terms of a polynomial number of Pauli operators by employing techniques such as the Hartree-Fock approximation or the Jordan-Wigner transformation. 
Similarly, the most of the Hamiltonian of a tight \hlgreen{binding} model and a spin network can be decomposed by a polynomial number of Pauli operators. 
Thus, the target Hamiltonian $H_{T}$ on \hl{an $N$-}qubit system can be written as $H_T = \sum _i r_i P_i$. Here $P_i$ is \hl{an} $N$-qubit Pauli operator generated by $i I$ and $\{X_k, Y_k, Z_k \}_{k=1}^N$, where $X_k$, $Y_k$, $Z_k$ are Pauli $X$, $Y$, $Z$ operators on $k$-th qubit, respectively.

The minimum eigenvalue $E_0$ of $H_{T}$ is a solution of the following variational problem $E_0 = \min_{\ket{\Psi}} \bra{\Psi}H_T\ket{\Psi}$. 
VQE approximates $E_0$ using variational method with trial state generated by a parametric ansatz circuit.
Suppose a parametric ansatz circuit is described by a unitary $\mathbf{U}(\boldsymbol{\xi})$, where $\boldsymbol{\xi}$ is a vector-parameter.
The trial state $\ket{\Psi(\boldsymbol{\xi})}$ with the parameter $\boldsymbol{\xi}$ can be written as 
$\mathbf{U}(\boldsymbol{\xi}) \ket{\mathbf{0}}$, where $\mathbf{\ket{0}}$ is a typical initial state.
Suppose $E_\text{VQE}$ is the minimum of 
the expectation value of $H_T$ of the trial state:
\begin{align}\label{E VQE}
    E_\text{VQE} &= \min_{\boldsymbol{\xi}}
E(\boldsymbol{\xi}),
\end{align}
where  $E(\boldsymbol{\xi}) = 
\bra{\Psi(\boldsymbol{\xi})}H_T \ket{\Psi(\boldsymbol{\xi})}$.
By means of the decomposition $H_T = \sum _i r_i P_i$, 
$E(\boldsymbol{\xi})$ can be written as 
\begin{equation}
    E(\boldsymbol{\xi})= \bra{\mathbf{0}} \mathbf{U}^{\dagger}(\boldsymbol{\xi}) H_T \mathbf{U}(\boldsymbol{\xi}) \ket{\mathbf{0}}
    =\sum _i \highlight{r_i} \bra{\mathbf{0}} \mathbf{U}^{\dagger}(\boldsymbol{\xi}) P_i \mathbf{U}(\boldsymbol{\xi}) \ket{\mathbf{0}}
\end{equation}
The expectation values $\bra{\mathbf{0}} \mathbf{U}^{\dagger}(\boldsymbol{\xi}) P_i \mathbf{U}(\boldsymbol{\xi}) \ket{\mathbf{0}}$ can be calculated by the parametric circuit $\mathbf{U}(\boldsymbol{\xi})$.
Then, $E(\boldsymbol{\xi})$ is derived by summing these expectation values on a classical computer. 
Furthermore, the partial derivatives of $E(\boldsymbol{\xi})$ with respect to $\boldsymbol{\xi}$ can also be computed using a similar quantum circuit. Consequently, by utilizing the output of the parameterized quantum circuit, the minimization of $E(\boldsymbol{\xi})$ can be performed sequentially on a classical computer using gradient methods. 

Here, we should note that this minimization may not be easy.  
That is, unless the ansatz circuit is chosen carefully, the so-called plateau problem, where the gradient vanishes, arises \cite{McClean2018,Cerezo2021,Wang2021}, and furthermore, it has been demonstrated that the minimization problem for this parameterized quantum circuit is NP-hard \cite{PhysRevLett.127.120502}. \hlgreen{To address the barren plateau problem, several strategies have recently been proposed, including exploiting many-body localization \cite{Cao_2025} and strategically initializing the circuit parameters \cite{Grant_2019}.}
Consequently, the above optimization is not performed rigorously, and the output of VQE becomes an approximation $E_\text{VQE}'$ of $E_\text{VQE}$. 
In experiments using actual quantum computers, furthermore, errors due to decoherence are added to the output of VQE. Hence, the accuracy of VQE estimates must be further enhanced through quantum error mitigation (QEM) techniques like ZNE or tailored probabilistic machine learning \cite{Kurita_2023, PhysRevResearch.6.033069}. 

\subsection{Indirect-control variational quantum eigensolver (VQE)}\label{vqe}

In quantum many-body systems, external access is often possible only to a limited number of degrees of freedom. Nevertheless, manipulating these limited degrees of freedom affects the entire system through time-evolution governed by the many-body Hamiltonian. Indeed, in typical spin chains like the XY model, it is known that accessing only a small number of spins allows complete control over the entire spin chain\hl{, meaning that,} in principle, any unitary transformation can be applied to the entire chain \cite{Burgarth09, Burgarth10,Maruyama2016}. 
Here, we refer to the directly accessible subsystem, where we perform unitary operators, as the control unit and the inaccessible subsystem,  where we cannot perform any active operations, as the inaccessible part.

To indirectly control a quantum many-body system by accessing only the control unit and achieve the desired unitary transformation, it is necessary to explicitly determine the operations to be performed on the control unit over time. However, at present, no general method has been established to efficiently determine the specific control method required to realize a given unitary transformation.
On the other hand, when implementing variational quantum algorithms like VQE on quantum many-body systems via indirect-control, the above problem does not arise. In this case, the desirable control method is automatically determined by optimizing the parameters of the ansatz circuit to minimize the cost function \cite{anan_indirect_vqe}.

Let us consider indirect-control of a quantum many-body system undergoing free time-evolution with the system Hamiltonian $H_S$. To implement an ansatz circuit in such system, we instantaneously operate a parameterized unitary operator $V_n(\vec{\theta}_n)$ on the control unit at time $t_n$ for $n=1,2,3, \cdots, L$ with $t_n <t_{n+1}$.  Then, an ansatz circuit which is given by the time-evolution of the whole system can be described as 
\begin{align}\label{eq total evolution}
\mathbf{U(\boldsymbol{\xi})}:=    \Pi_{n=1}^L U(t_n,t_{n+1}) V_n(\Vec{\theta}_n)
\end{align}
Here, the parameter of the circuit $\boldsymbol{\xi}$ consists of the time parameters $\{t_n \}_
{n=1}^L$ and the angular parameters $\{\vec{\theta}_n\}_{n=1}^L$. In Equation \eqref{eq total evolution}, $U(t_n,t_{n+1}):= \exp (-iH_S (t_{n+1}-t_n))$ is a time-evolution operator acting on the whole system from time $t_n$ to time $t_{n+1}$.

In \cite{anan_indirect_vqe}, the first two qubits of 1-D XY-spin chain are chosen as a control unit, and two-qubits unitary $V_n(\vec{\theta}_n)$ consist of $X$ and $Y$ rotation gates $R_X, R_Y$ and control-$Z$ gates $CZ$ as 
\begin{equation}
    V_n(\vec{\theta}_n) 
    = R_Y(\theta_{4n-2})R_X(\theta_{4n-3})\otimes R_Y(\theta_{4n})R_X(\theta_{4n-1})\cdot CZ, \label{eq def local unitary}
\end{equation}
where $\vec{\theta}_n$ is defined by
$\vec{\theta}_n=(\theta_{4n-3},\theta_{4n-2},\theta_{4n-1},\theta_{4n})$.
Further, the system Hamiltonian $H_S$ is chosen as the 1-D XY model Hamiltonian $H_\text{XY}$:

\begin{equation}
H_{\text{XY}} 
= \sum_{k=1}^{N-1} c_k \left[ (1 + \highlight{\nu}) X_k X_{k+1} + (1 - \highlight{\nu}) Z_k Z_{k+1} \right]
+ \sum_{k=1}^{N} b_k Y_k,
\label{eq-xy-hamiltonian}
\end{equation}
where $N$ is the number of qubits,  \(c_k\) are the coupling constants, \(b_k\) are the strengths of the local magnetic fields, and \hl{$\nu$} is an anisotropy parameter.
Figure \ref{fig-indirect-vqe} depicts the $n$-th layer of the indirect-control VQE ansatz.  In \cite{anan_indirect_vqe}, Anan {\em et al.} showed that the performance of VQE using the above indirect-control parametric circuit is not so different from that of using standard parametric circuits. 

Here, we add a note regarding the measurements necessary to implement VQE. In previous research on indirect-control, it was often assumed that measurements could only be performed on the control unit. On the other hand, implementing VQE generally requires obtaining the expectation value of the Pauli operator, including the inaccessible part. For the one-dimensional XY model without an external magnetic field, Anan  {\em \hl{et al.}} provided a method using indirect-control to efficiently determine the expectation value of such a Pauli operator using only measurements of the control unit \cite{anan_indirect_vqe}. However, for general many-body Hamiltonians, no known method efficiently performs such measurements. Therefore, this study assumes that it is possible to perform a measurement of the Pauli operator, including the inaccessible part, just once at the end of the circuit.

\Figure[t!](topskip=0pt, botskip=0pt, midskip=0pt){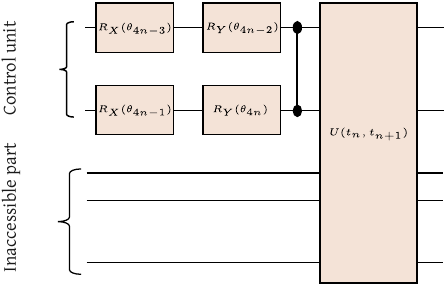}
{ An illustration of the $n$-th layer of an indirect-control VQE ansatz circuit for a five-qubit system. The first two qubits constitute the control unit, where arbitrary quantum gate operations can be applied, while the remaining three qubits are not directly accessible and evolve solely under the system’s intrinsic dynamics.\label{fig-indirect-vqe}}

\section{Error mitigation by zero noise extrapolation (ZNE)}\label{zne}
\subsection{ZNE by the Richardson extrapolation}
When a VQE is implemented on a real quantum device, the output of the VQE, which is an estimator of minimum eigenvalue of a target Hamiltonian $H_T$, must be affected by environmental noise. 
This is the case even though we reduce the noise by implementing the VQE via indirect-control.
In general, an output of a variational quantum algorithm like a VQE is an estimate of the mean value of 
an observable. Thus, the variance and bias of the estimate may increase by the effect of noise. 
In such a situation, error mitigation reduces the bias at the expense of the variance by post-processing of the measurement results. 
Zero-noise extrapolation (ZNE), which is one of the most common error mitigation techniques, mitigates errors by extrapolating data to the zero-noise limit via classical post-processing \cite{PhysRevX.7.021050, PhysRevLett.119.180509}. 

Let us consider the simplest problem setting, where the energy estimation by a VQE, denoted by $E(\lambda)$, is a function of the noise level represented by a single parameter $ \lambda $. 
By definition, $E(0)$ is an output $E_{VQE}'$ of the VQE on a noiseless circuit,
and, thus, $E(0)$ is a good estimator of the minimum eigenvalue $E_0$ of $H_T$. However, in actual experiments, due to hardware limitations, there exists a certain $\lambda_0>0$ such that the noise level $\lambda$ of the ansatz circuits  cannot be set to a value smaller than $\lambda_0$.
Zero-noise extrapolation (ZNE) mitigates errors by extrapolating $E(\lambda)$ to the zero-noise limit $\lambda \rightarrow 0$ through classical post-processing. 
 Although it is difficult to decrease the noise level $\lambda$ less than $\lambda_0$ in a real quantum device, we may increase the noise level  $\lambda$ by various methods,  such as unitary folding or parameter noise scaling \cite{PhysRevLett.120.210501,digitalZne,PhysRevA.102.012426}. 
 Thus, by calculating $E(\lambda)$ for various values of $\lambda$ satisfying $\lambda \ge \lambda_0$ and extrapolating the function $E(\lambda)$ to $\lambda = 0$, we can achieve a good estimated value of $E(0)$.

 Suppose we have $n+1$ noisy samples $\big \{E\left(\alpha_0 \lambda_0\right),$ $ E\left(\alpha_1 \lambda_0\right),$ $ \dots, E\left(\alpha_n \lambda_0 \right)\big \}$ of the noisy expectation value of an observable for a given $\lambda_0 >0$, where $\alpha_n \ge 1$ for all $n$.
In the Richardson extrapolation method with a single variable \cite{doi:10.1098/rsta.1911.0009},  we assume that $E\left(\lambda\right)$ can be expressed as a polynomial of $\lambda$:
\begin{equation}
    E\left(\lambda\right) = E\left(0\right) + \sum_{k=1}^n C_k \lambda^k + \mathcal{O}\left(\lambda^{n+1}\right),
    \label{eq-noisy-expec-val}
\end{equation}
where $E\left(0\right)$ is the zero-noise expectation value and $C_k$ are coefficients.
When $\lambda$ is small, the last term of Equation \eqref{eq-noisy-expec-val} can be ignored.
Then, $E\left(\lambda\right)$ can be approximated by an $n$-th order polynomial,
which is written as $E_\text{Ric}(\lambda)$. 
The polynomial $E_\text{Ric}$ satisfying $E\left(\alpha_i \lambda_0 \right)=E_\text{Ric}(\alpha_i \lambda_0 )$
for $i=0,\cdots, n$ is unique, and given by the following Lagrange interpolating polynomials:
\begin{align}\label{eq lagrange interpolating poly}
  \highlight{  E_\text{Ric}(\lambda)=\sum_{i=0}^n \frac{\omega_{n,i}(\lambda )}{\omega_{n,i}(\alpha_i\lambda_0)}E\left(\alpha_i \lambda_0 \right),}
\end{align}
where \hl{$\omega_{n,i}$} is defined by 
\begin{equation}
   \highlight{ \omega_{n,i}(\lambda)=\frac{(\lambda-\alpha_0\lambda_0)(\lambda-\alpha_1\lambda_0)\cdots (\lambda-\alpha_n\lambda_0)}{(\lambda - \alpha_i\lambda_0)}.}
\end{equation}

From Equation \eqref{eq lagrange interpolating poly}, $E_\text{Ric}\left(0\right)$, which is an approximation  of
$ E\left(0\right)$, can be written as 
\begin{equation}\label{eq E Ric 0 equal sum}
    E_\text{Ric}\left(0\right) = \sum_{k=0}^{n}\beta_k E\left(\alpha_k \highlight{\lambda_0}\right),
\end{equation}
where the coefficients $\beta_k$ are given by   
\begin{equation}
    \beta_k = \prod_{i \neq k} \frac{\alpha_i}{\alpha_i - \alpha_k}.
    \label{zne-beta}
\end{equation}

Although $E_\text{Ric}\left(0\right)$ can be calculated from $\{ E\left(\alpha_k \highlight{\lambda_0}\right) \}_{k=0}^n$ by means of Equation \eqref{eq E Ric 0 equal sum}, 
what we can access in a real experiment is not $E\left(\alpha_k \highlight{\lambda_0}\right)$ itself, but an estimator $\mathbf{E}\left(\alpha_k \highlight{\lambda_0}\right)$ of $E\left(\alpha_k \lambda_0\right)$, where $\mathbf{E}\left(\alpha_k \highlight{\lambda_0}\right)$ is calculated from measurement results of corresponding circuits. 
Suppose $\mathbf{E}_\text{Ric}\left(0\right)$ is an estimator of $E_\text{Ric}\left(0\right)$. 
The variance of $\mathbf{E}_\text{Ric}\left(0\right)$ can be written down in terms of the variances of $\mathbf{E}\left(\alpha_k \highlight{\lambda_0}\right)$ as 
\begin{align}
\highlight{\operatorname{Var}}\big(\mathbf{E}_{\text{Ric}}(0)\big)
&= \sum_{k=0}^{n} \beta_k^2 \, \highlight{\operatorname{Var}}\big(\mathbf{E}(\alpha_k \highlight{\lambda_0})\big) \nonumber \\
&= \sum_{k=0}^{n}
\left( \prod_{i \neq k} \frac{\alpha_i^2}{(\alpha_i - \alpha_k)^2} \right)
\, \highlight{\operatorname{Var}}\big(\mathbf{E}(\alpha_k \highlight{\lambda_0})\big).
\label{eq-vari-ric-uni}
\end{align}
where $\highlight{\operatorname{Var}}(\mathbf{X})$ is the variance of random variable $\mathbf{X}$.
In this paper, we will choose $\{ \alpha_i\}_{i=0}^n$ as 
$\alpha_i=2i+1$. In this case, by a straightforward calculation, we can show  $\prod_{i \neq k} \frac{\alpha_i^2}{(\alpha_i - \alpha_k)^2}\ge 2^n/(2n+1)^2$. Thus, $\highlight{\operatorname{Var}}(\mathbf{E}_\text{Ric}\left(0\right))$ increases exponentially with respect to the number $n$ of the noisy samples.
This means that the choice of $n$ should be moderately small.

\subsection{Unitary folding of $n$-qubit gates under local depolarizing channel}

Let us consider a realistic situation where the ansatz has a stochastic noise process
as the $N$-qubit quantum channel $\mathcal{E}$.
Taking into account the possibility that time-evolution gates could be the primary source of noise
in our ansatz circuit (Figure \ref{fig-indirect-vqe}),
here, for simplicity, we focus on circuits consisting of a sequence of $N$-qubit gates.  
Suppose \hl{an} $N$-qubit unitary gate $G$ acting on \hl{an} $N$-qubit density matrix $\rho$ as $\rho\mapsto [G](\rho)=G \rho G^\dagger$ in an ideal ansatz circuit, where $[G]$ is a unitary channel corresponding to $G$. 
Then, it will act as $\rho \mapsto \mathcal{E} \circ [G](\rho)=\mathcal{E} (G \rho G^\dagger)$ in the corresponding noisy ansatz circuit.

In this paper, we opted for unitary folding, where, for an arbitrary $N$-qubit unitary gate $G$, we insert  $G^{\dagger}G$, which is equal to an identity in an ideal case, just after $G$ in the circuit. In other words, $G$ is replaced  by $G^\dagger GG$.
Then, $G^\dagger GG$ acts as 
\begin{equation}
\rho\mapsto \mathcal{E} \circ [G]^\dagger \circ \mathcal{E} \circ [G]\circ \mathcal{E} \circ [G](\rho)    
\end{equation}
in a noisy circuit.
Let us consider a simple situation where the noisy channel $\mathcal{E}$ is \hl{an} $N$-qubit ``\textit{global}'' depolarizing channel $\mathcal{D}_\lambda^{(N)}$ defined as follows:
\begin{equation}
     \mathcal{D}_p^{(N)} (\rho)= (1-p)\rho+ p \cdot 2^{-N} I_{N},
     \label{eq-noise-process}
\end{equation}
 where $p$ satisfies $0\le p \le 1$, and $I_{N}$ is the $2^N$ dimensional identity matrix. 
In this case, 
for all $N$-qubit unitary operator $G$, $[G]$ commutes with $\mathcal{D}_p^{(N)}$: 
$[G]\circ \mathcal{D}_p^{(N)} = \mathcal{D}_p^{(N)} \circ [G]$.
Then, we can evaluate the action of $G^\dagger GG$ as follows:
\begin{align}\label{commute global depolarizing}
\mathcal{D}_p^{(N)} \circ [G]^\dagger \circ \mathcal{D}_p^{(N)} \circ [G]
\circ \mathcal{D}_p^{(N)} \circ [G] =\mathcal{D}_{3p}^{(N)} \circ [G] +\mathcal{O}(p^2).
\end{align}
Therefore, when $p$ is sufficiently small, by replacing $G$ with $G^\dagger GG$, we can effectively amplify the noise probability of the global depolarizing channel from $p$ to $3p$. Similarly, replacing $G$ with $(G^\dagger G)^K G$ corresponds to  amplifying 
the noise probability of the global depolarizing channel from $p$ to $(2K+1)p$ in the first order of $p$. Hereafter, for unitary folding, in which $G$ is replaced by $(G^\dagger G)^K G$, we will refer to $K$ as the \textit{folding number} and $2K+1$ as the \textit{scale factor}. 

When all the noise in an ansatz circuit for a VQE can be described by a global depolarizing channel $\mathcal{D}_p^{(N)}$ as above, an energy estimation $E(p)$ of the VQE is a function of $p$, and 
can be approximated by a polynomial like Equation~\eqref{eq-noisy-expec-val} with $\lambda =p$.
In this case, we can derive an estimation of $E((2K+1)p)$ 
by replacing all $N$-qubit unitary gates, $G_i$ with $(G_i^\dagger G_i)^KG_i$ in the ansatz circuit. 
We refer to such a circuit containing noisy identities of the form $G^\dagger G$ as a {\em folded circuit}. We refer the reader to the work of Nachman \textit{et al.}~\cite{PhysRevA.102.012426} for a detailed analysis of ZNE via unitary folding under a global depolarizing channel.

In the above discussion, we have assumed that the quantum channel $\mathcal{E}$ is a global depolarizing channel $\mathcal{D}_p^{(N)}$.
On the other hand, a global depolarizing channel can be implemented by uniformly and randomly applying all $N$-qubit Pauli operators. This implies that the interaction from the environment \hl{(}which is the source of the noise affecting the $N$ qubits\hl{)} exhibits very strong non-local correlations. In contrast, in this paper, we consider implementing VQE by indirectly controlling a quantum many-body system. Given that quantum many-body systems are inherently local in nature, it is unreasonable to assume that the time-evolution of the quantum many-body system is subject to noise with strong non-locality, such as that of a global depolarizing channel.
Therefore, in this paper, we use a local depolarizing channel defined by $( \mathcal{D}_p^{(1)})^{\otimes N}$ rather than a global depolarizing channel $\mathcal{D}_p^{(N)}$ as the quantum channel $\mathcal{E}$.  Unlike a global depolarizing channel, a local depolarizing channel does not possess non-local correlations.
Therefore, the local depolarizing channel is a suitable simple noise model for quantum many-body systems, as we consider in this paper.

\hl{Also,} unlike \hl{the} global depolarizing channels, \hl{the} local depolarizing channels do not commute with unitary transformations of the first order with respect to $p$.
Therefore, when noise is represented by a local depolarizing channel, even after unitary folding, the simple scaling law for noise parameter shown in Equation~\eqref{commute global depolarizing} does not hold even in the first-order approximation with respect to $p$.
Nevertheless, as we will see in the next section, even when $p$ is small, ZNE via unitary folding is sufficiently effective for error mitigation.
In Appendix \ref{appendix local depolarizing}, we show that the commutation relation breaks down when noise is represented by a local depolarizing channel.

\section{Unitary folding in indirect-control VQE}

In this section, we describe how to implement ZNE using unitary folding when VQE is implemented using indirect-control on a one-dimensional XY spin chain.

\subsection{Hamiltonian of indirect-control system}
In this paper, we choose a one-dimensional XY spin chain as the quantum many-body system in which we implement VQE via indirect-control. As we will explain later, this choice of the system Hamiltonian is crucial for implementing unitary folding using indirect-control.
Thus, the system Hamiltonian $H_S$ is given by
\begin{equation}
    H_S = H_\text{XY},
\end{equation}
where
\begin{equation}
    H_\text{XY} = \frac{1}{2}\sum_{k=1}^{N-1} \left( X_k X_{k+1} + Z_k Z_{k+1} \right).
    \label{equation-hxz}
\end{equation}
The single-qubit depolarizing channel $\mathcal{D}_p^{(1)}$ is given by Equation \eqref{eq-noise-process},
where $p$ denotes the probability that a random Pauli operation acts on the state $\rho$. As we have explained in the last section, we assume that for multi-qubit gates each qubit is subject to \hl{a} local depolarizing channel (see Figure \ref{indirect_vqe_redun_ansatz}).

\Figure[t!](topskip=0pt, botskip=0pt, midskip=0pt){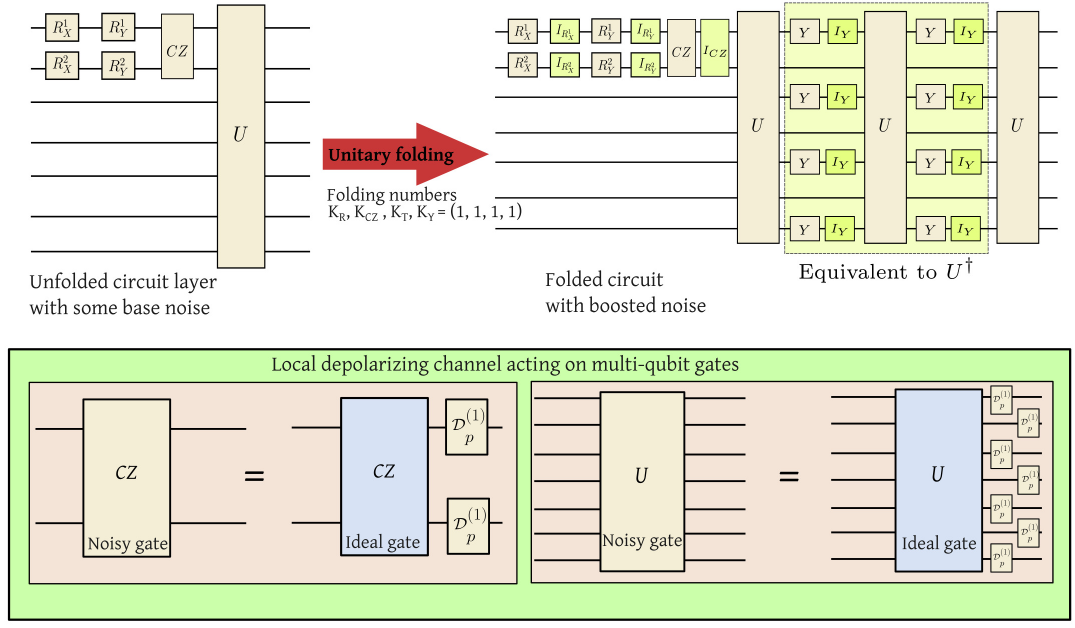}{\hl{Noise boosting of the standard indirect-control VQE circuit is achieved using unitary folding, where a noisy unitary gate $G$ is replaced by $GG^\dagger G$. We refer to the circuit containing such noisy $G^\dagger G$ identity gates as a \textit{folded circuit}. The $n$-th layer of a 7-qubit indirect-control folded circuit with noisy time-evolution is depicted, where $I_{R^{1}_{x}}$ is an abbreviation for $R_X^\dagger R_X$ on the 1st qubit, $I_{CZ}$ represents $CZ^\dagger CZ$, and $I_Y$ represents $Y^\dagger Y$. The highlighted section corresponds to $U^\dagger(t_n, t_{n+1})$}.\label{indirect_vqe_redun_ansatz}}

\subsection{Implementing negative time-evolution}
As explained in Section \ref{sec2}, our ansatz circuit is implemented by indirectly controlling a 1-D XY model.
Therefore, the system is constantly driven by a drift Hamiltonian $H_{XY}$, and under the assumption that the norms of the Hamiltonian elements realizing 1-qubit and 2-qubit gates are sufficiently larger than the norm of the drift Hamiltonian $H_{XY}$, an approximate circuit like the one shown in Figure \ref{fig-indirect-vqe} is obtained.
For example, \hlgreen{in Figure \ref{fig-indirect-vqe}}, the time-evolution gate $U(t_n, t_{n+1})$ is not a gate that is actively applied. Rather, it is a description of the time-evolution from time $t=t_n$ to time $t=t_{n+1}$ resulting from the instantaneous application of single-qubit and two-qubit gates at times $t=t_n$ and $t=t_{n+1}$, respectively. \hl{Nonetheless, throughout this study, we use the terms ``time evolution'' and ``time-evolution gate'' interchangeably to refer to this description of the system's evolution.}

Now, let us consider the case where decoherence associated with time-evolution is the primary source of noise, and that noise is described by a local depolarizing channel. In this case, to amplify the noise using unitary folding, we need to operate $U^\dagger U$.
For $U = U(t_n, t_{n+1})$, $U^\dagger$ can be written as $U^\dagger = U(t_{n+1}, t_n)$. Thus, $U^\dagger$ represents the reverse time-evolution of the drift Hamiltonian $H_{XY}$. However, in general, an indirect-control system is not time-reversible, and
it is difficult to implement reverse time-evolution in quantum many-body systems using indirect-control \cite{Burgarth10,Maruyama2016,KATO2020168046}.

Even in the 1-D XY model, we are not aware of any method for implementing reverse time-evolution simply by applying 1-qubit or 2-qubit gates to the two qubits at the ends, which serve as the control unit.
Therefore, in this \hl{study}, we aim to implement reverse time-evolution by adding a minimal number of operations to the inaccessible part.
In fact, on a 1D-XY spin chain, that is, when the drift Hamiltonian is $H_{XY}$ as given by Equation \eqref{equation-hxz}, it is possible to show that $U^\dagger (t_n, t_{n+1})$ can be implemented by simply adding an ability to perform $Y$ gates on odd-numbered qubits including the inaccessible part. 
This is due to the symmetry of 1-D XY model Hamiltonian $H_{XY}$ without local magnetic fields given by $\Pi_{j=0}^MY_{2j+1}\cdot H_{XY}\cdot \Pi_{j=0}^MY_{2j+1}=-H_{XY}$ when $N=2M+1$; the similar equation also holds when $N=2M$. This equation leads 
\begin{equation}
    \Pi_{j=0}^MY_{2j+1}\cdot U (t_n, t_{n+1})\cdot \Pi_{j=0}^MY_{2j+1}=U^\dagger (t_n, t_{n+1})
\end{equation}
Hence, for an arbitrary $n$-qubit XY time-evolution gate $U(t_n, t_{n+1})$, the inverse gate $U^\dagger(t_n, t_{n+1})$ can be implemented by setting $b_k = 0$ in Equation \eqref{eq-xy-hamiltonian} and \hlgreen{applying the Pauli $Y$ gate}
to the odd-numbered qubits before and after the time-evolution, as illustrated in the highlighted section of Figure \ref{indirect_vqe_redun_ansatz} for the $H_{XY}$ Hamiltonian. 
\hlgreen{Therefore, by adding the ability to apply $Y$ gates to the odd-numbered qubits, in addition to the 1-qubit and 2-qubit gates applied to the two end qubits, it is possible to implement ZNE via unitary folding in the one-dimensional XY model. Although we relaxed the indirect-control constraint by applying $Y$ gates to the odd sites, it is theoretically possible to selectively apply $Y$ gate operations to these otherwise inaccessible sites. This can be achieved by constructing a spin chain with two types of qubits having different resonance frequencies and applying a resonant control pulse for the odd-site qubits to the entire spin chain \cite{Menta_2026, PhysRevResearch.7.L012065}. We note that such a global control pulse could in principle introduce correlated errors across the odd-site qubits. Nonetheless, in our simulations we treat the $Y$-gate error on the same footing as the other single-qubit gate errors.}

\section{Numerical simulations}
\label{Numerical simulations}
\hl{In this section, we present the results of simulations using Qulacs \cite{Suzuki_2021}, a Python library for quantum circuit simulation, to demonstrate the implementation of ZNE in VQE using indirect-control for the 1-D XY model described in the previous section.}
Our target Hamiltonian is a one-dimensional transverse-field Ising spin chain expressed as
\begin{equation}
    H_\text{Ising} = \sum_{k=1}^{N-1} X_k X_{k+1} + \sum_{k=1}^{N} Z_k.
    \label{ising-hamiltonian}
\end{equation}
\hl{We simulated the system using density matrices, thereby avoiding finite-shot sampling noise.}
In the VQE simulation, the initial value of the angle parameter $\theta_n$ is selected uniformly at random from the interval $[0,2\pi]$.
In addition,
the initial value of the time-parameter $t_0$ is set $0$, and
the initial values of the time parameters $\{t_n \}_{n=2}^{30}$ are sampled uniformly from $[0, T_\text{max}]$, and then reindexed such that $t_{n+1} > t_n$.
The optimization of the angle parameter $\theta_n$ and the time-parameter $t_n$, subject to the constraint $t_{n+1}>t_n$, is performed using the \texttt{minimize} function from the Python library \texttt{scipy}, with Sequential Least Squares Programming (SLSQP) \cite{kraft1988software} specified as the solver. 
\hl{To keep the optimization time tractable, the precision tolerance \texttt{ftol} is set to $10^{-3}$ in the \texttt{scipy} optimization for all subsequent VQE optimizations. In the SLSQP algorithm, this parameter acts as a tolerance bound not only for the relative change in the objective function but, more strictly, for the residuals of the Karush–Kuhn–Tucker (KKT) optimality conditions derived from the Lagrange multipliers and the absolute constraint violations computed at each sequential quadratic programming subproblem \cite{nocedal2006numerical, lawson1995solving}.}

\subsection{Description of decoherence using the noisy ansatz circuit}
\label{Description of decoherence using the noisy ansatz circuit}
We first consider the case where all noise introduced into a 1-D XY model is due to decoherence, which occurs independently for each qubit.
 As mentioned in the previous section, the simplest model describing this situation is local depolarization noise. \hl{Further, local dephasing noise and local amplitude damping noise are also typical models.} 
The Lindblad master equation is the appropriate tool for describing such a situation \cite{breuer2002theory}:
\begin{align}\label{eq lindblad eq}
    \frac{d\rho}{dt}& =\mathcal{L}_\text{XY}(\rho) +\mathcal{L}_\text{noise}(\rho).
\end{align}
Here, $\mathcal{L}_\text{XY}$ and $\mathcal{L}_\text{noise}$ are  $N$-qubit Lindbladians describing a Hamiltonian dynamics and decoherence, respectively:
\begin{align}
\mathcal{L}_\text{XY}(\rho) &= -i[H_\text{XY},\rho], \nonumber \\ 
\mathcal{L}_\text{noise}(\rho) &=\gamma \sum_{i=1}^N\mathcal{L}_i(\rho),    
\end{align}
where $\gamma$ is a single-qubit decoherence rate, and $\mathcal{L}_i$ is a single-qubit Lindbladian representing decoherence of the $i$-th qubit.
When $\mathcal{L}_\text{noise}$
represents decoherence due to local depolarizing noise, using the Liouville representation, $\mathcal{L}_i$ can be written as $\mathcal{L}_i=- \Pi_{XYZ,i}$,
where $\Pi_{XYZ,i}$ is a projection onto 3-dimensional subspace of a single-qubit Hilbert-Schmidt space spanned by $\{ X_i, Y_i, Z_i\}$.
\hl{Similarly, when $\mathcal{L}_\text{noise}$
represents decoherence due to local dephasing noise, $\mathcal{L}_i$ can be written as $\mathcal{L}_i=- \Pi_{XY,i}$,
where $\Pi_{XY,i}$ is a projection onto 3-dimensional subspace of a single-qubit Hilbert-Schmidt space spanned by $\{ X_i, Y_i\}$.}

\Figure[t!](topskip=0pt, botskip=0pt, midskip=0pt){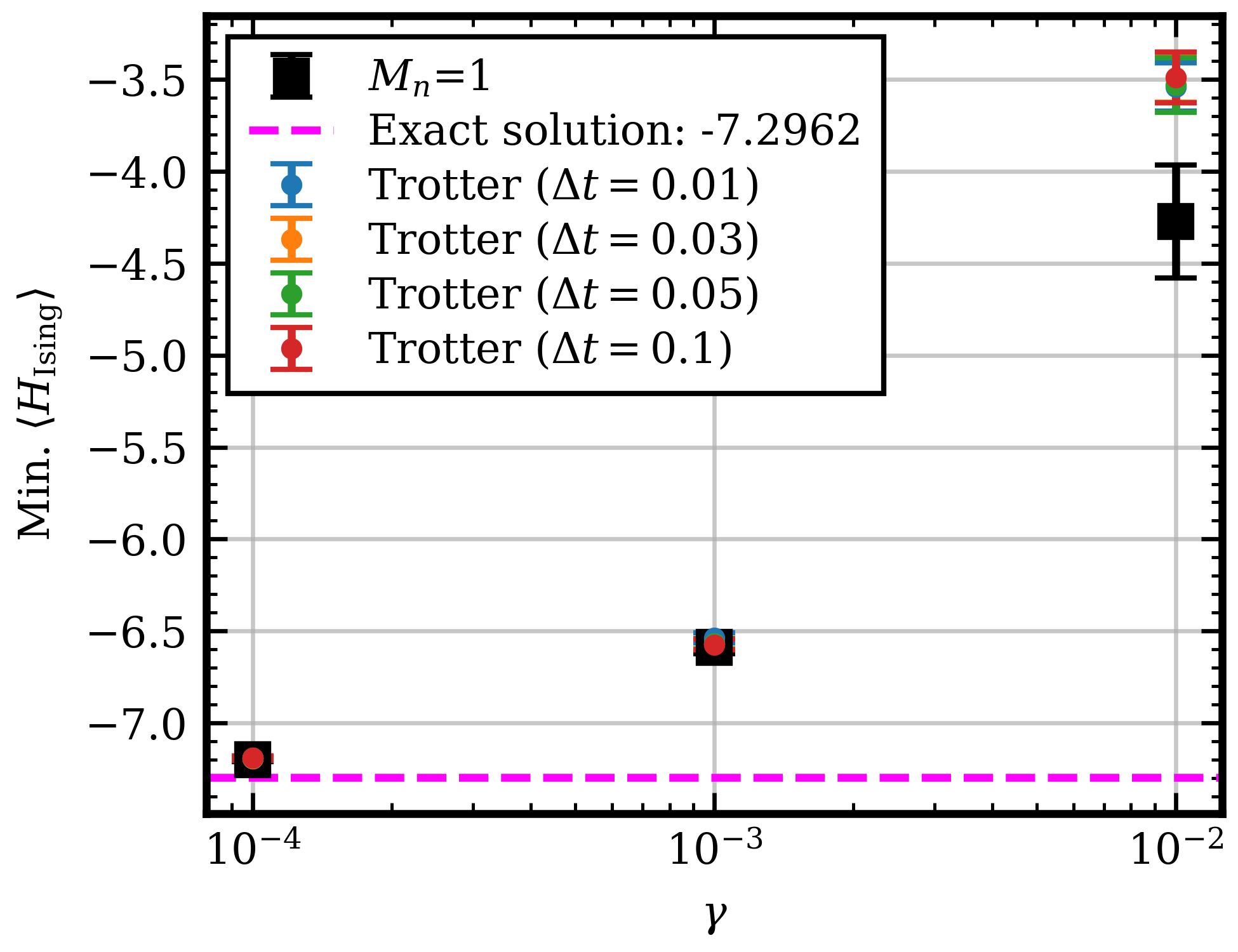}
{A 6-qubit, 30-layer indirect-control VQE ansatz with initial random time parameters sampled from the range $[0, T_\text{max}=20]$ is used. The minimum expectation value of the 6-qubit $H_{\text{Ising}}$ is estimated. The results compare Trotterized time-evolution with a single time-evolution gate subject to local depolarizing noise for various given decoherence rates $\gamma$, where the noise strength scales with the duration of the evolution. \hl{The error bars represent the standard deviation computed from 10 independent trials, each consisting of VQE optimization. The resulting standard deviation reflects the variability arising from the random initialization of the variational parameters.}\label{fig-trotter-vs-no-trotter-gamma}}

\Figure[t!](topskip=0pt, botskip=0pt, midskip=0pt){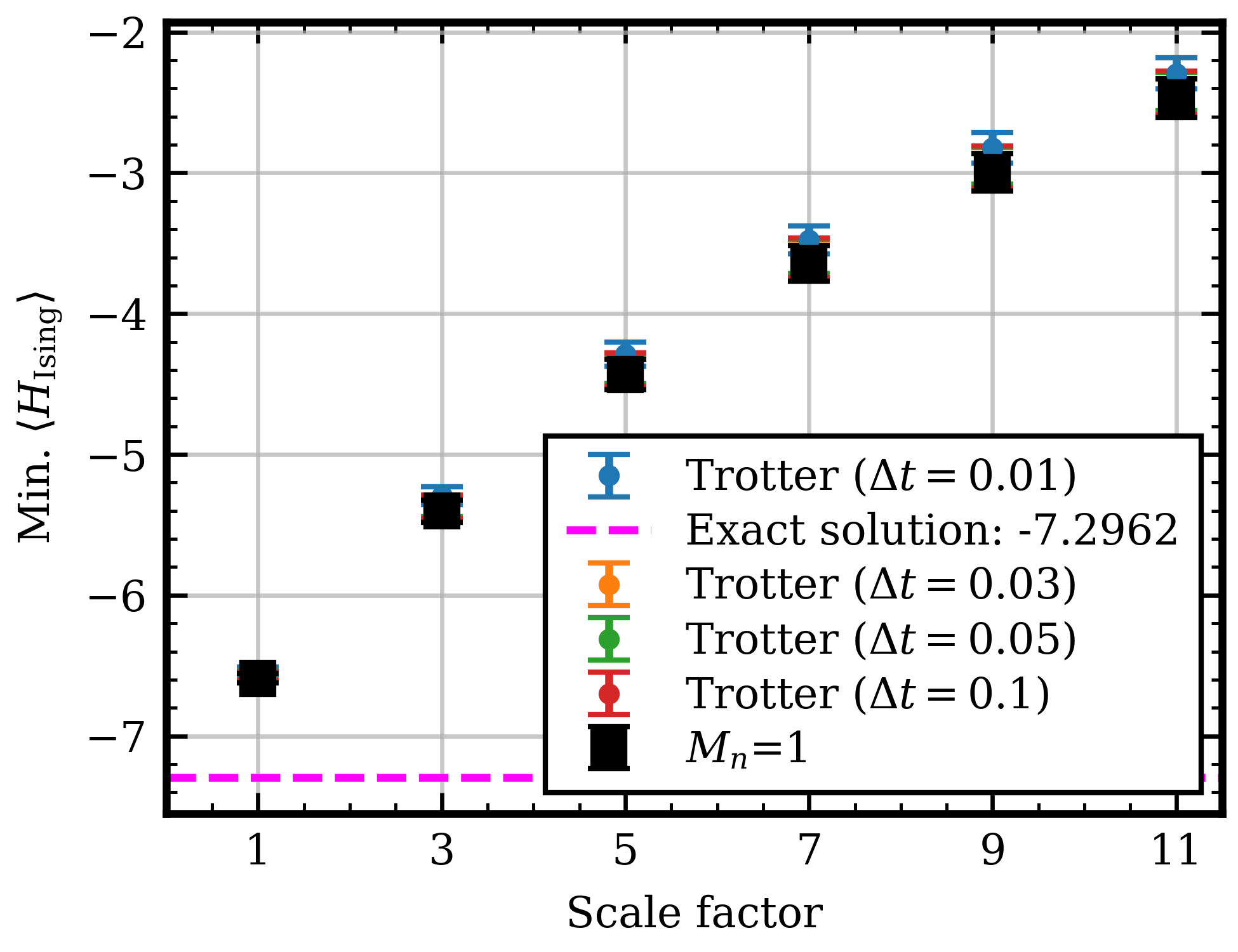}
{A 6-qubit, 30-layer indirect-control VQE ansatz with initial random time parameters sampled from the range $[0, T_\text{max}=20]$ and a decoherence rate of $\gamma = 0.001$ is used. The minimum expectation value of the 6-qubit $H_{\text{Ising}}$ is estimated. The results compare Trotterized time-evolution with a single time-evolution gate subject to local depolarizing noise, where the noise strength scales with the duration of the evolution. Noise levels are amplified using unitary folding. \hl{The error bars represent the standard deviation computed from 10 independent trials, each consisting of VQE optimization. The resulting standard deviation reflects the variability arising from the random initialization of the variational parameters.}\label{fig-trotter-vs-no-trotter}}

The solution of Equation \eqref{eq lindblad eq} can be written formally as 
$e^{(\mathcal{L}_\text{XY}+\mathcal{L}_\text{noise})t}\rho$. However, since computing this matrix numerically is time-consuming, we use an approximation via the Trotter decomposition \cite{Trotter59}:
\begin{equation}
e^{(\mathcal{L}_\text{XY}+\mathcal{L}_\text{noise})t}=\lim_{M\rightarrow \infty}
    \left(
e^{\mathcal{L}_\text{XY}\frac{t}{M}}\circ e^{\mathcal{L}_\text{noise}\frac{t}{M}}\right)^M.
\label{eq-trotter-lindblad}
\end{equation}
\hl{Here, $e^{\mathcal{L}_\text{XY}t}$ and  $e^{\mathcal{L}_\text{noise}t}$ correspond to a unitary time-evolution and a quantum channel representing decoherence, respectively.
They can be written as}
\begin{align}
e^{\mathcal{L}_\text{XY}t}\rho & = e^{-iH_\text{XY}t}\rho e^{iH_\text{XY}t}, \nonumber \\  e^{\mathcal{L}_\text{noise}t}\rho & =(e^{\gamma \mathcal{L}_it})^{\otimes N}(\rho).
\label{eq-lindblad-noise}
\end{align}

\hl{Let us focus on the case of local depolarizing noise, where $\mathcal{L}_i=-\Pi_{XYZ,i}$. 
In this case, we can easily see that $e^{\gamma \mathcal{L}_it}$ is given by $e^{\gamma \mathcal{L}_it}=\mathcal{D}_{1-e^{-\gamma t}}^{(1)}$,}
where $\mathcal{D}_{1-e^{-\gamma t}}^{(1)}$ is a single-qubit depolarizing channel defined by Equation \eqref{eq-noise-process}.
Here, we note that when $\gamma t$ is small, $\mathcal{D}_{1-e^{-\gamma t}}^{(1)}\approx \mathcal{D}_{\gamma t}^{(1)}$ for $\gamma t \ll 1$. 
Therefore, when $t$ is much smaller than $\gamma^{-1}$, that is, when the typical timescale is sufficiently smaller than the inverse of the decoherence rate, the error probability of the local depolarization noise corresponding to $e^{\mathcal{L}_\text{noise}t}$ is given by $\gamma t$.

Let us consider the indirect-control ansatz circuit given by Figure \ref{fig-indirect-vqe} under the decoherence given by Equation \eqref{eq lindblad eq}.
Since we implement single-qubit gates and \hl{$CZ$} gates instantaneously, only the time-evolution gates $U(t_n, t_{n+1})$ are affected by decoherence according to Equation \eqref{eq lindblad eq}. Thus,  
a time-evolution gate $U(t_n, t_{n+1})$ on $n$-th layer is replaced by 
\begin{equation}\label{eq trotterized time-evolution}
\left(
e^{\mathcal{L}_\text{XY}\frac{t_{n+1}-t_n}{M_n}}\circ
(\mathcal{D}_{\gamma (t_{n+1}-t_n)/M_n}^{(1)})^{\otimes N}\right)^{M_n},    
\end{equation}
where error probability $p_n$ is given by
\begin{equation}\label{eq p n gamma}
    p_n=\gamma (t_{n+1}-t_n)/M_n
\end{equation}
Here, we used the assumption that $\gamma (t_{n+1}-t_n)/M_n \ll 1$.
Since the error in the Trotter decomposition given by \hlgreen{Equation \eqref{eq-trotter-lindblad}} is of order \hl{$\mathcal{O}$}$(\frac{(t_{n+1}-t_n)^2}{M_n})$, we will fix $\Delta t:=\frac{t_{n+1}-t_n}{M_n}$ here and choose a different Trotter step $M_n$ for each layer.
Then, for a fixed $\Delta t$, if $t_{n+1}-t_n \ge \Delta t$, we choose $M_n$ such that $M_n = \left \lfloor \frac{t_{n+1}-t_n}{\Delta t} \right \rfloor$ , and if
$t_{n+1}-t_n < \Delta t$, we choose $M_n$ to be $1$. 
Here, since $M_n=1$ for $t_{n+1}-t_n < 2\Delta t$,  
the corresponding noisy time-evolution gate is given by 
$e^{\mathcal{L}_\text{XY}(t_{n+1}-t_n)}\circ
(\mathcal{D}_{\gamma (t_{n+1}-t_n)}^{(1)})^{\otimes N}$,
which is nothing but the noisy time-evolution gate given in Figure \ref{indirect_vqe_redun_ansatz} with noise parameter $p_n$ given by Equation \eqref{eq p n gamma} with $M_n=1$:
\begin{equation}
p_{n} = \gamma(t_{n+1} - t_n).
\label{eq-time-evo-noise}
\end{equation}
Choosing $M_n = 1$ in Equation \eqref{eq trotterized time-evolution} is equivalent to using the approximation
\begin{equation}\label{eq approximation Mn 1}
    e^{(\mathcal{L}_{XY}+\mathcal{L}_\text{noise})(t_{n+1}-t_n)} \approx e^{\mathcal{L}_\text{XY}(t_{n+1}-t_n)}e^{\mathcal{L}_\text{noise}(t_{n+1}-t_n)}.
\end{equation}

\hl{As explained in Appendix \ref{app:trotter_error_explicit}, the above error is bounded by $\frac{1}{2}\gamma (t_{n+1}-t_n)^2(N-1)c_\text{noise}$ in terms of the diamond norm,
where $c_\text{noise}$ is a constant of order $\mathcal{O}(1)$.
Therefore, as long as $t_{n+1}-t_n$ satisfies 
$t_{n+1}-t_n\ll
\left[c_\text{noise}(N-1)\gamma\right]^{-1/2}    
$,} we can use the approximation given by Equation \eqref{eq approximation Mn 1}.
Under this approximation, the noisy time-evolution between $t_n$ and $t_{n+1}$ is given by the unitary gate with local depolarizing noise introduced in Section \ref{zne}.

We present the results of numerical simulations of VQE using indirect-control ansatz with $30$-layer and $6$-qubit in Figures~\ref{fig-trotter-vs-no-trotter-gamma} and
\ref{fig-trotter-vs-no-trotter}, obtained by approximating decoherence using various values of $\Delta t$
and Equation \eqref{eq trotterized time-evolution}. 
Here, $T_\text{max}$ is set to $20$. Thus,
the initial values of the time parameters $\{t_n \}_{n=2}^{30}$ are sampled uniformly from $[0, 20]$, and then reindexed such that $t_{n+1} > t_n$. 
No error mitigation is performed here.

Figure~\ref{fig-trotter-vs-no-trotter-gamma} shows the results of indirect VQE using a standard (non-folded) circuit when the decoherence rate $\gamma$ is $10^{-4}$, $10^{-3}$, and $10^{-2}$.
The x-axis in Figure~\ref{fig-trotter-vs-no-trotter-gamma}  represents the decoherence rate $\gamma$, and the y-axis represents the average value of the target Hamiltonian $H_\text{Ising}$.
From Figure~\ref{fig-trotter-vs-no-trotter-gamma} and \ref{fig-trotter-vs-no-trotter},
we observe that for $\gamma = 10^{-4}$ and $\gamma = 10^{-3}$, 
the VQE results depend very little on $\Delta t$.
On the other hand, when $\gamma =10^{-2}$, it can be seen that differences in $\Delta t$ \hl{lead to} significant differences in the VQE results.
In particular, when $M_n = 1$ is chosen in the Trotter expansion of Equation \eqref{eq trotterized time-evolution} \hl{(}that is, when the approximation in Equation \eqref{eq approximation Mn 1} is used\hl{)} the results differ significantly from those obtained for small $\Delta t$.
\hl{In fact, as mentioned earlier, for the approximation in Equation \eqref{eq approximation Mn 1} to hold, the relationship $t_{n+1}-t_n \ll [c_\text{noise}(N-1)\gamma]^{-\frac{1}{2}}$ must hold. Since $c_\text{noise}\le 4$ for a local depolarizing channel (see Appendix \ref{app:trotter_error_explicit}), when $N=6$ and $\gamma =10^{-2}$, $t_{n+1}-t_n \ll [c_\text{noise}(N-1)\gamma]^{-\frac{1}{2}} \le \sqrt{5}\approx 2.2$, so this relationship is not satisfied, and the approximation in Equation \eqref{eq approximation Mn 1} is considered to fail.
As mentioned earlier, in the VQE algorithm, the time parameters $\{t_n\}_{n=2}^{30}$ are first sampled uniformly from $[0,20]$ and then optimized. }
In our observations, in many cases, even after optimization, the time-parameter $t_n$ does not take values that deviate significantly from the range $[0, 20]$. 
Consequently, $t_{n+1}-t_n$ is \hl{$\mathcal{O}$}$(1)$ in high probability, 
and in fact, in most cases, \hl{the condition $t_{n+1} -t_n \ll 2.2$ is not satisfied.
This is why the result with $M_n = 1$ differs significantly from those obtained for small $\Delta t$ when $\gamma=10^{-2}$.
On the other hand, for $\gamma = 10^{-3}$, $t_{n+1}-t_n \ll [c_\text{noise}(N-1)\gamma]^{-\frac{1}{2}} \le \sqrt{50}\approx 7$. This condition is satisfied in most case, since average of $t_{n+1}-t_n$ is smaller than $1$ in high probability. This is why the result with $M_n = 1$ is not so different from those obtained for small $\Delta t$ when $\gamma=10^{-3}$.}

In Figures~\ref{fig-trotter-vs-no-trotter},  we perform indirect-control VQE using folded circuits with various scale factors of $2K_T+1$, where unitary folding is applied only for time-evolution gates, and  $K_T$ is a folding number of time-evolution gates.
Here, the decoherence rate is fixed at \hl{$\gamma=10^{-3}$.}
As noted above, when $\gamma$ was set to this value, no significant dependence of the VQE results on $\Delta t$ was observed in the standard circuit.
The x-axis in Figure~\ref{fig-trotter-vs-no-trotter} represents the scale factor $2K_T+1$, and the y-axis represents the average value of the target Hamiltonian $H_\text{Ising}$ obtained when performing VQE using the corresponding folded circuit. 

From Figure~\ref{fig-trotter-vs-no-trotter}, we observe that for a given scale factor, 
the VQE results depend very little on $\Delta t$. In particular, the results  for very small $\Delta t$ are nearly identical to those obtained when $M_n$ is set to 1 for all layers.
This suggests that even when simulations are performed using a circuit with $M_n$ set to $1$, where decoherence is represented by adding a local depolarizing channel with $p_n$ given by Equation \eqref{eq-time-evo-noise} following the time-evolution gate as shown in Figure~\ref{indirect_vqe_redun_ansatz}, the results obtained are almost indistinguishable from those obtained by exactly simulating according to the Lindblad equation (Equation  \eqref{eq lindblad eq}).
\hl{On the other hand, the condition $t_{n+1}-t_n \ll [c_\text{noise}(N-1)\gamma]^{-\frac{1}{2}}$ depends
only on the square root of $N$. Hence, although the numerical calculations above were performed for $N=6$, it is expected that the VQE results for $M_n=1$ will be nearly identical to those obtained using the Trotter decomposition even for $N=7$ or $N=8$. }
Therefore, to save computation time, we fix $\gamma$ to be \hl{$10^{-3}$} and simulate decoherence using the folded circuit given by Figure~\ref{indirect_vqe_redun_ansatz} \hl{even for $N=7$ or $N=8$}.

\begin{table}
\caption{\textbf{List of \hl{the symbols}}}
\label{tab1-summary-hyper-param}
\setlength{\tabcolsep}{3pt}
\newcolumntype{L}[1]{>{\raggedright\arraybackslash}p{#1}}
\begin{tabular}{|L{25pt}|L{190pt}|}
\hline
Symbol & Quantity \\
\hline
$\gamma$ & Decoherence rate \\
$p_n$ & Time-dependent local depolarizing noise probability \\
$T_{\text{max}}$ & Upper bound for the initial random time parameters, uniformly sampled from $0$ to $T_{\text{max}}$ \\
$M_n$ & Number of Trotter steps \\
$\Lambda^{(i)}$ & Noise level in circuit $i$ \\
$K_G^{(i)}$ & Folding number for gate $G$ in circuit $i$ \\
$\vec{\Lambda}^{(i)}$ & Vector noise parameter in circuit $i$ \\
$\mathcal{N}_G^{(i)}$ & Number of gate $G$ operations in a layer in circuit $i$ \\
\hline
\end{tabular}
\end{table}
\subsection{Results of ZNE when noise is present only in the time domain}
\label{7q-uni-var-zne}

\hl{Below, we present the results of simulations performed using a $7$-qubit, 30-layer (i.e., 150 variational parameters) indirect-control VQE ansatz circuit followed by ZNE for error mitigation. Local depolarizing noise acts only on the time-evolution gates, while the $R$, $Y$, and $CZ$ gates are assumed to be noise-free. First, we obtain the optimized variational parameters $\boldsymbol{\xi}$ using the indirect-control VQE. To implement ZNE, it is necessary to execute the VQE circuit multiple times using the same optimized variational parameters $\boldsymbol{\xi}$ at different noise levels. We denote this group of circuits as:}
\begin{equation*}
    \left\{ \mathcal{C}(\Lambda^{(1)}, \boldsymbol{\xi}), \mathcal{C}(\Lambda^{(2)}, \boldsymbol{\xi}), \mathcal{C}(\Lambda^{(3)}, \boldsymbol{\xi}), \dots \right\},
\end{equation*}
together with their corresponding expectation value estimations,
\begin{equation*}
    \left\{ E(\Lambda^{(1)}), E(\Lambda^{(2)}), E(\Lambda^{(3)}), \dots \right\},
\end{equation*}
\hl{where $\Lambda^{(i)}$ denotes the corresponding noise level arising from the noisy time-evolution gates and can be quantified simply in terms of the number of time-evolution gates in a layer, since all layers are structurally identical within a given circuit. Equivalently, $\Lambda^{(i)}$ represents the scale factor resulting from folding the time-evolution gates and can be expressed in terms of the folding number as:}
\begin{equation}
    \Lambda^{(i)} = 2K_T^{(i)} + 1
\end{equation}
\hl{where $K_T^{(i)}$ is the folding number of the time-evolution gates when executing the $i$-th circuit. From Equation~\eqref{eq E Ric 0 equal sum}, we see that an extrapolation of order $n$ requires $(n+1)$ data points. We can further organize the measurements as a dataset of noise--expectation value pairs,}
\begin{equation}
    D = \left\{ (\Lambda^{(1)}, E(\Lambda^{(1)})), (\Lambda^{(2)}, E(\Lambda^{(2)})), \dots \right\}.
    \label{eq-zne-dataset}
\end{equation}

\hl{The following results were obtained:}

\Figure[t!](topskip=0pt, botskip=0pt, midskip=0pt){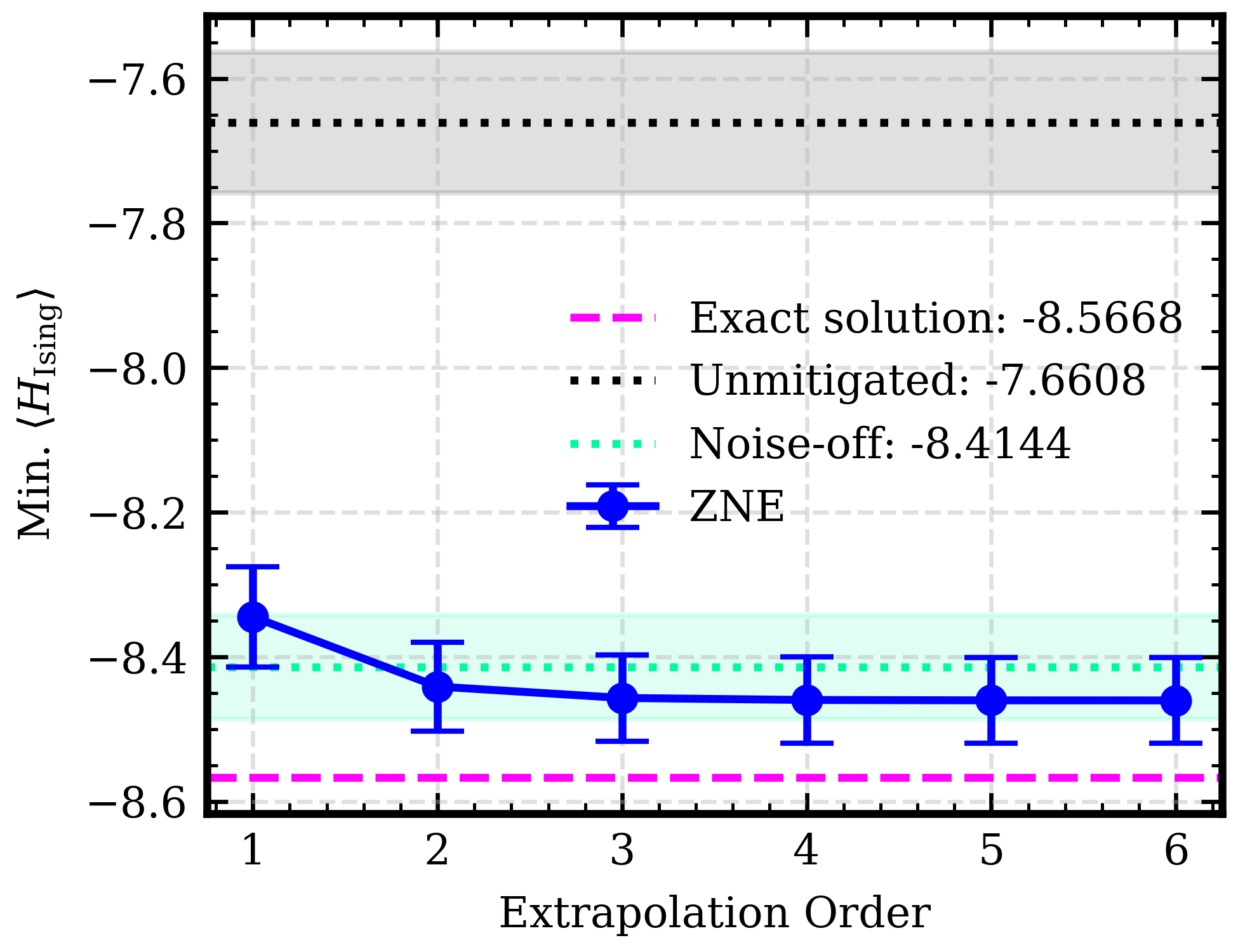}
{Single-variate Richardson extrapolation of different orders is shown. A 7-qubit, 30-layer indirect-control VQE ansatz with initial random time parameters sampled from the range $[0, T_\text{max}=20]$ is used. The minimum expectation value of the 7-qubit $H_{\text{Ising}}$ is estimated. The time-evolution gates are simulated using $M_n=1$ in Equation \eqref{eq trotterized time-evolution}. Local depolarizing noise with $\gamma = 0.001$ in Equation~\eqref{eq-time-evo-noise} is applied only to the multi-qubit time-evolution gate. \hl{The error bars represent the standard deviation computed from 10 independent trials, each consisting of VQE optimization followed by ZNE. The resulting standard deviation reflects the variability arising from the random initialization of the variational parameters.}\label{fig-plots-zne-diff-orders}}

\Figure[t!](topskip=0pt, botskip=0pt, midskip=0pt){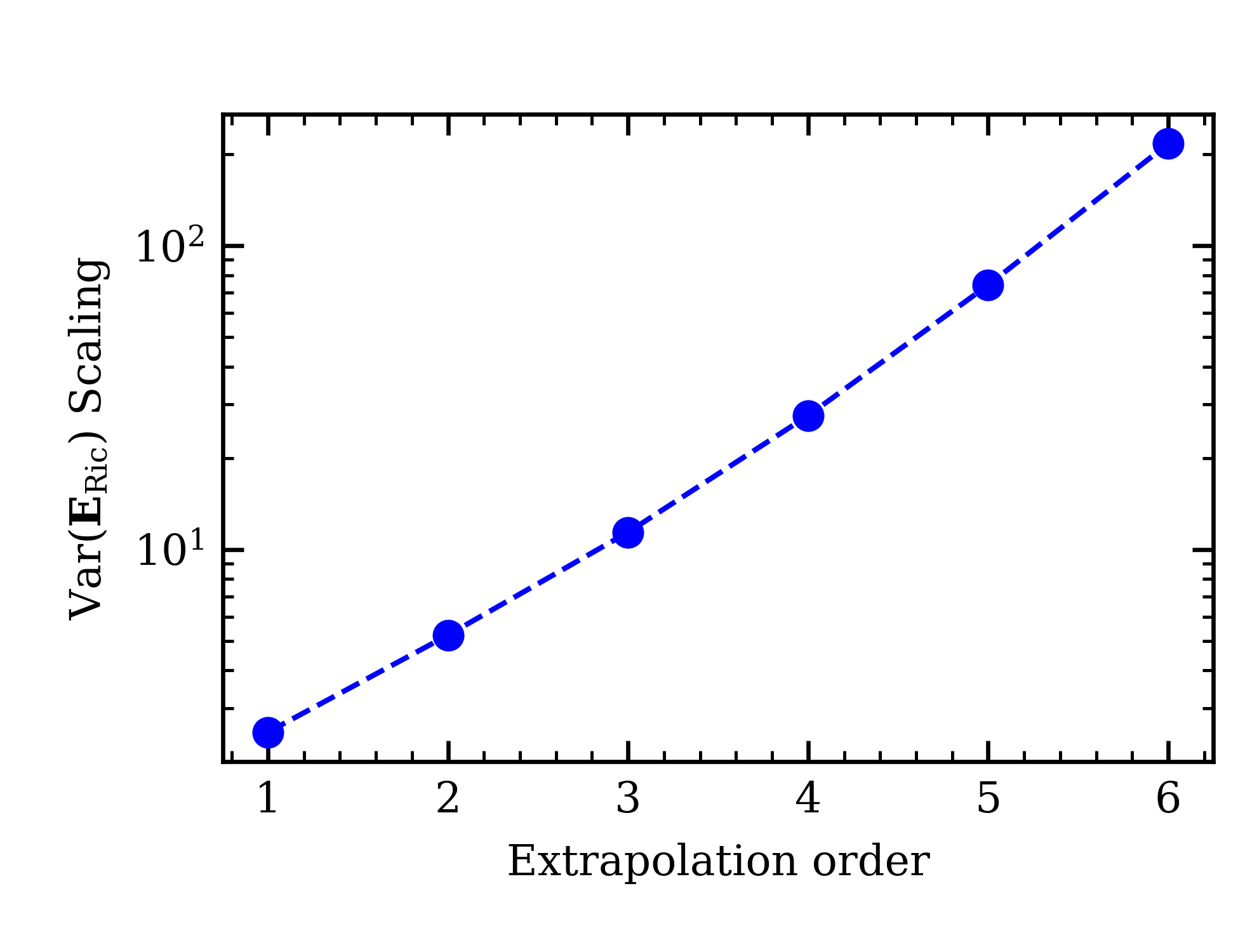}
{\hl{$\operatorname{Var}(\mathbf{E}_{\mathrm{Ric}}(0))$ scales according to Equation~\eqref{eq-vari-ric-uni}, assuming that all noisy expectation values in Equation~\eqref{eq-vari-ric-uni} have a fixed variance. Local depolarizing noise acts only on the time-evolution gates in the 7-qubit, 30-layer system, as discussed in Section~\ref{7q-uni-var-zne}. The minimum expectation value of the 7-qubit $H_{\text{Ising}}$ is estimated under this noise model.}\label{fig-plots-zne-vs-variance}}

\begin{enumerate}
    \item \textbf{Vary Over ZNE Orders:}  We performed indirect-control VQE to estimate the minimum expectation value (ground-state energy) of $\langle H_\text{Ising} \rangle$. Using the same optimized parameters, we applied unitary folding to construct various folded circuits and measured the expectation values at different noise levels. Once the measurement dataset in Equation \eqref{eq-zne-dataset} is obtained, classical post-processing techniques, such as Richardson extrapolation, can be applied. In Figure~\ref{fig-plots-zne-diff-orders} (refer to Table \ref{table-single-var-zne} ), the performance of different orders of Richardson extrapolation is shown for $\gamma =0.001$ and $T_\text{max}=20$. 
    As shown in this plot, when using Richardson \hl{extrapolation} of order two or higher, no significant difference was observed between the estimates obtained using ZNE and those obtained in the absence of noise.
    Thus, we confirm that when the decoherence rate $\gamma$ is sufficiently small (e.g., $\gamma=0.001$),  ZNE can mitigate errors in the expectation value within the statistical uncertainty \hl{range} of the estimation obtained from indirect-control VQE without noisy gates. \hlgreen{Additionally, Figure~\ref{fig-plots-zne-vs-variance} (refer to Table \ref{table-single-var-zne} ) illustrates the near-exponential scaling of the variance in ZNE estimation with increasing extrapolation order, assuming that all noisy expectation values in Equation~\eqref{eq-vari-ric-uni} have a fixed variance.}

\begin{table}
\caption{\hl{\textbf{ZNE results for the 7-qubit, 30-layer indirect-control VQE estimating $\langle H_{\text{Ising}} \rangle$. The ansatz contains 150 variational parameters: 30 time parameters and 120 angle parameters.}}}
\label{table-single-var-zne}
\setlength{\tabcolsep}{3pt}
\begin{tabular}{|p{90pt}|p{71pt}|p{54pt}|}
\hline
Method & Estimated Value (Mean $\pm$ Std. Dev.) $^{\mathrm{1}}$ & \hl{$\operatorname{Var}\big(\mathbf{E}_{\text{Ric}} (0)\big)$ Scaling $^{\mathrm{2}}$}\\

\hline
    Noise-free estimation & $-8.41 \pm 0.07 $ & -- \\
    \hline
    At base noise $\Lambda^{(1)}= 1$ & $-7.66 \pm 0.10 $ & -- \\
    At $\Lambda^{(2)}= 3$ & $-6.29 \pm 0.15 $ & -- \\
    At $\Lambda^{(3)}= 5$ & $-5.18 \pm 0.19 $ & -- \\
    At $\Lambda^{(4)}= 7$ & $-4.28 \pm 0.21 $ & -- \\
    At $\Lambda^{(5)}= 9$ & $-3.54 \pm 0.22 $ & -- \\
    At $\Lambda^{(6)}= 11$ & $-2.94 \pm 0.22 $ & -- \\
    At $\Lambda^{(7)}= 13$ & $-2.45 \pm 0.21 $ & -- \\
    \hline
    ZNE of order 1 & $-8.34 \pm 0.07 $ & $2.50 \times 10^{0}$ \\
    ZNE of order 2 & $-8.44 \pm 0.06 $ & $5.22 \times 10^{0}$ \\
    ZNE of order 3 & $-8.46 \pm 0.06 $ & $1.14 \times 10^{1}$ \\
    ZNE of order 4 & $-8.46 \pm 0.06 $ & $2.76 \times 10^{1}$ \\
    ZNE of order 5 & $-8.46 \pm 0.06 $ & $7.43 \times 10^{1}$ \\
    ZNE of order 6 & $-8.46 \pm 0.06 $ & $2.17 \times 10^{2}$ \\
\hline

\multicolumn{3}{p{215pt}}{\hl{All values are expressed in dimensionless units. The first-order Richardson extrapolation is constructed using the base noise $\Lambda^{(1)} = 1$ and a single amplified noise point at $\Lambda^{(2)} = 3$. The second-order Richardson extrapolation incorporates the base value together with two amplified noise points at $\Lambda^{(2)} = 3$ and $\Lambda^{(3)} = 5$. Higher-order Richardson estimates are obtained in a similar manner by progressively including additional data points corresponding to larger noise amplification factors.}} \\
\multicolumn{3}{p{215pt}}{\hl{$^{\mathrm{1}}$ The error bars represent the standard deviation computed from 10 independent trials, where each trial consists of VQE optimization with a different random initialization of the variational parameters, followed by ZNE using the optimized parameters.} }\\
\multicolumn{3}{p{215pt}}{\hl{$^{\mathrm{2}}$$\operatorname{Var}(\mathbf{E}_{\mathrm{Ric}}(0))$ scales according to Equation~\eqref{eq-vari-ric-uni}, assuming that all noisy expectation values in Equation~\eqref{eq-vari-ric-uni} have a fixed variance.}}\\
\hline
\end{tabular}
\end{table}

    \item \textbf{Varying the Time-Parameter Range:} We sampled the initial random time parameters uniformly from $[0, T_{\text{max}}]$ for $0 < T_{\text{max}} \leq 100$, and for each choice of $T_{\text{max}}$ performed \hl{3rd-order} Richardson extrapolation to obtain ZNE estimates, plotted in Figure~\ref{fig-plots-zne-vs-tmax}. Local depolarizing noise was applied exclusively to the time-evolution gate.  
    In this plot, the colored circles represent the \hl{3rd-order} ZNE estimates obtained from the standard and folding ansatz circuits with noise, while the colored squares represent the VQE estimates from the noise-free standard ansatz circuit. 
    
    First, the results show that the choice of $T_{\text{max}}$ significantly affects the accuracy of the ZNE estimate. Notably, the optimized time parameters remain close to their initial values across all choices of $T_{\text{max}}$ regardless of whether noise is present or absent. This is not imposed by the optimizer constraints, which only enforce positivity and monotonic ordering of the time parameters with no upper bound; rather, it reflects an intrinsic property of the cost landscape, where local minima exist across all time scales and the optimizer consistently converges to a basin near the initialization. This initialization-dependence suggests that the indirect-control VQE energy landscape contains multiple local minima distributed across a wide range of evolution times.

    Next, from this figure, we can see that both the estimate plots obtained from the noise-free circuit and the ZNE estimate plots obtained from the noisy circuit show a sharp initial decrease in estimate values \hl{(}that is, the accuracy of the estimates increases rapidly\hl{)} and then reach a minimum around $T_{\text{\hl{max}}} = 15$. This phenomenon can be explained by the fact that, when $T_{\text{\hl{max}}}$ is small, the time-evolution gate is close to the identity matrix, and the ansatz circuit does not possess sufficient computational power. 
    As $T_{\text{\hl{max}}}$ is increased from that minimum value, the plot of estimates obtained from the noise-free circuit appears to decrease slightly and converge to a constant value. Even when $T_{\text{\hl{max}}}$ is increased, the computational capacity of the ansatz circuit should not decrease. Therefore, this phenomenon indicates that the optimization algorithm has stopped at a local minimum.

    On the other hand, when $T_{\text{\hl{max}}}$ is increased from the minimum value, the plot of the estimates obtained using ZNE from the noisy folding circuits continues to decrease without converging to a constant value. As a result, as $T_{\text{\hl{max}}}$ increases, the difference between the estimates obtained from the noise-free circuit and those obtained using ZNE from the noisy folding circuits monotonically widens. This phenomenon can be explained by the fact that as $T_{\text{\hl{max}}}$ increases, the execution time of the circuits lengthens, and the total amount of decoherence experienced by the circuits increases. Consequently, it is considered that the total amount of error introduced into the circuits increases, causing the accuracy of error mitigation to decrease as $T_{\text{\hl{max}}}$ increases.

    \Figure[t!](topskip=0pt, botskip=0pt, midskip=0pt){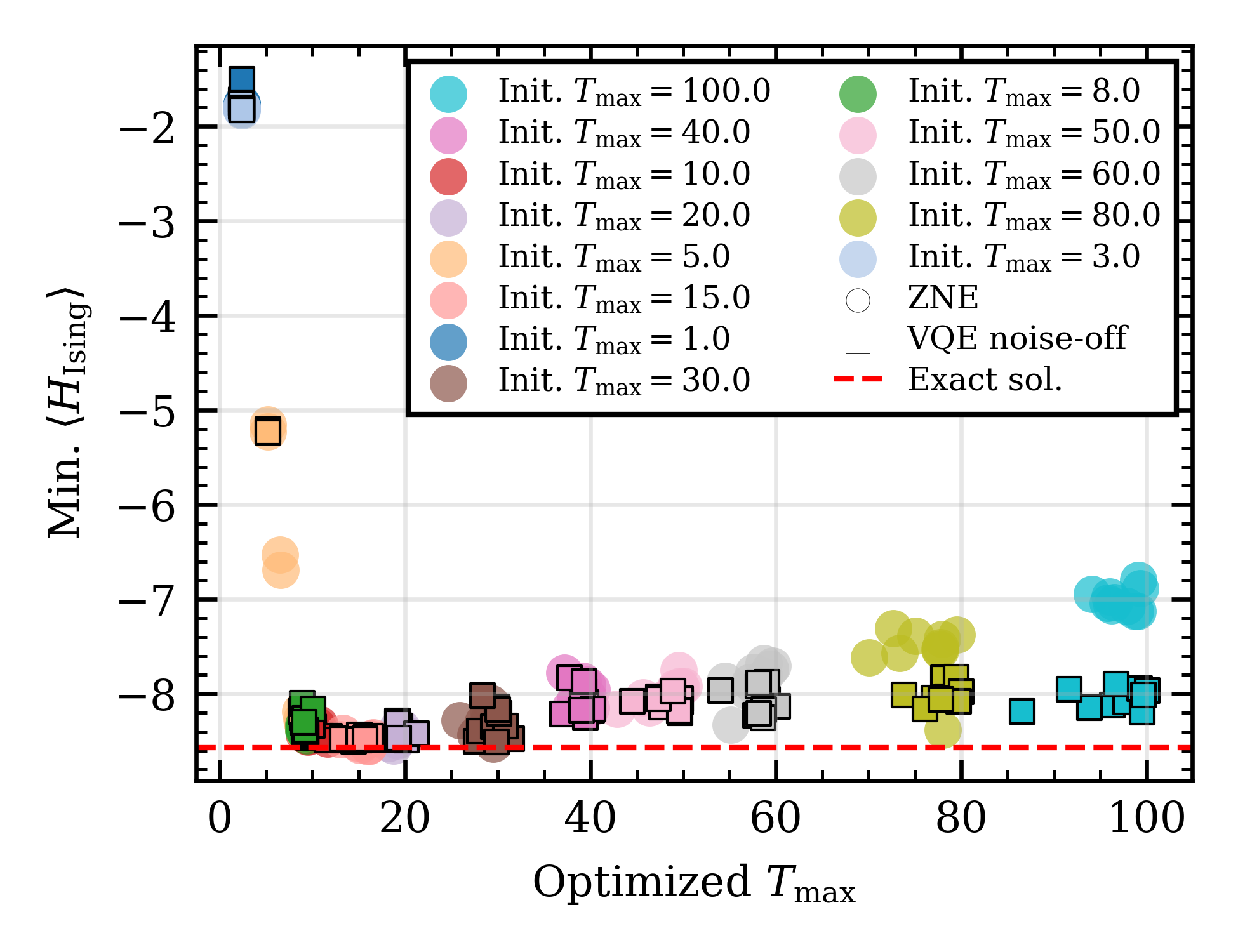}
    {A 7-qubit, 30-layer indirect-control VQE ansatz is used with $\gamma = 0.001$ in Equation~\eqref{eq-time-evo-noise}, while all other gates, including $R, Y$ and $CZ$, are assumed to be noise-free. The minimum expectation value of the 7-qubit $H_{\text{Ising}}$ is estimated. The VQE/ZNE performance for different choices of the initial time-parameter range $[0, T_{\text{max}}]$ is shown. The time-evolution gates are simulated using $M_n=1$ in Equation \eqref{eq trotterized time-evolution}. The colored circles denote \hl{3rd-order} ZNE estimates, whereas the colored squares represent noise-free ansatz VQE estimates. Each color corresponds to a distinct choice of $T_{\text{max}}$ when sampling the initial random time parameters.\label{fig-plots-zne-vs-tmax}}

    \item \textbf{Vary Over Noise Strength:} Subsequently, we fix $T_{\text{max}} = 20$ and vary the decoherence rate $\gamma$. We then apply ZNE using different orders of Richardson extrapolation to investigate the range of noise strengths over which accurate estimation can be recovered. We observe that as $\gamma$ increases, the ZNE estimates progressively deviate from the exact solution as can be confirmed from the plot in Figure \ref{fig-plots-zne-vs-delp}. On the other hand, the improvement in the estimates obtained by using ZNE \hl{(}that is, the difference between the estimates in the unmitigated case and those obtained using ZNE\hl{)} is minimal when $\gamma = 10^{-2}$. At $\gamma = 10^{-1}$, all estimates are nearly zero; this indicates that the decoherence experienced by the circuit is too large, causing the state to approach a completely mixed state. Consequently, it becomes practically impossible to perform any computations.

    \Figure[t!](topskip=0pt, botskip=0pt, midskip=0pt){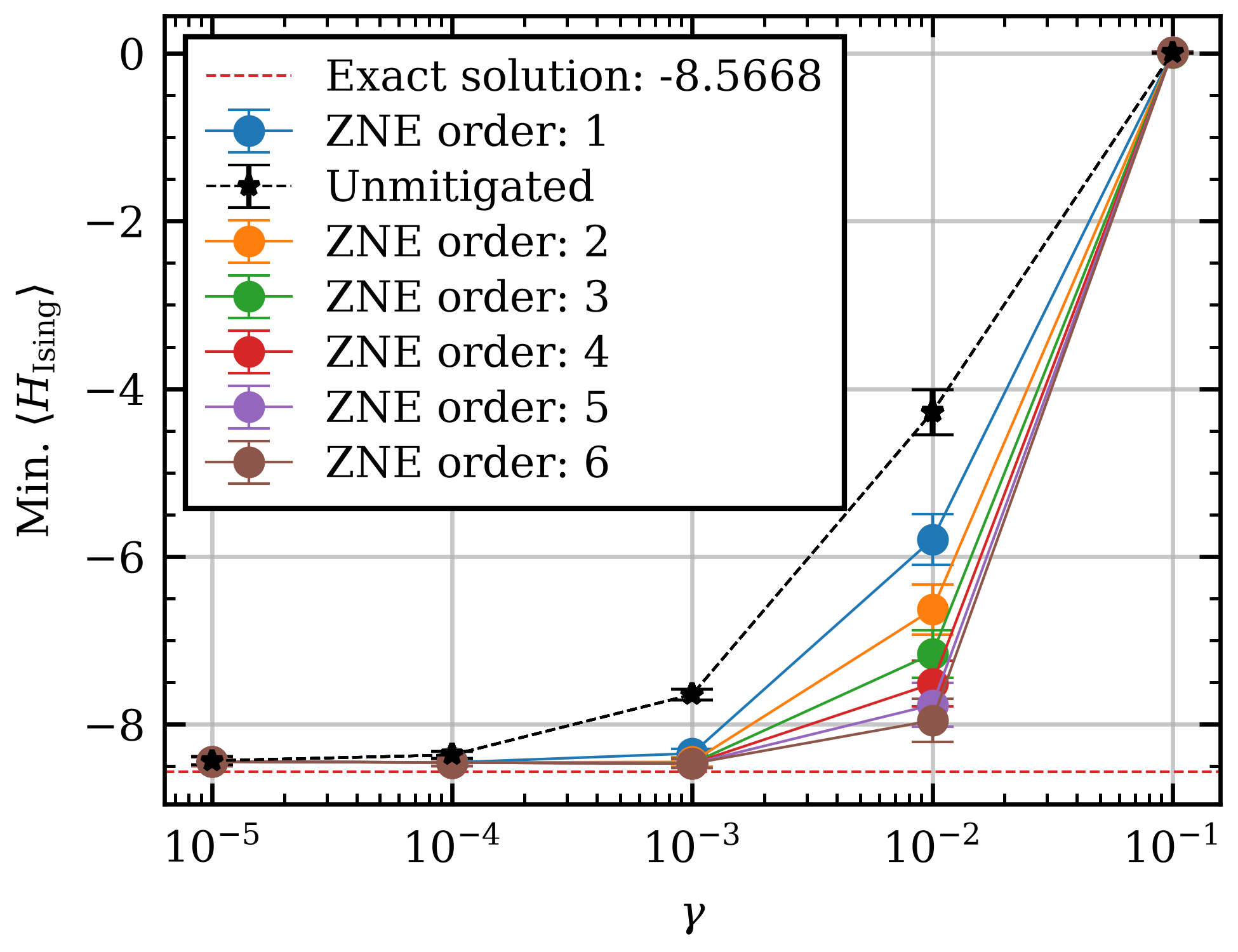}
    {Variation with respect to the parameter $\gamma$ in Equation~\eqref{eq-time-evo-noise}. A 7-qubit, 30-layer indirect-control VQE ansatz is used, while the $R$, $Y$, and $CZ$ gates are assumed to be noise-free. The minimum expectation value of the 7-qubit $H_{\text{Ising}}$ is estimated. Single-variate ZNE of different orders is shown as we vary $\gamma$ in Equation~\eqref{eq-time-evo-noise}, with $T_{\text{max}}=20$. The time-evolution gates are simulated using $M_n=1$ in Equation \eqref{eq trotterized time-evolution}. \hl{The error bars represent the standard deviation computed from 10 independent trials, each consisting of VQE optimization followed by ZNE. The resulting standard deviation reflects the variability arising from the random initialization of the variational parameters.} \label{fig-plots-zne-vs-delp}}

\end{enumerate}

\subsection{\hl{Multivariate ZNE with vector Noise Model}}

Finally, we consider a more realistic scenario in which all gates, including both control unit and time-evolution gates, are subject to local depolarizing noise channels. Single-qubit gates such as $R$ and $Y$ are assigned a depolarizing noise probability of $0.001$. We also set $\gamma = 0.001$ in Equation \eqref{eq-time-evo-noise}, thereby accounting for the noise associated with the time-evolution gate. In contrast, the local depolarizing noise acting on the two-qubit $CZ$ gate is assigned a higher depolarizing noise probability of $0.005$.

\begin{table}
\caption{\hl{\textbf{Multivariate ZNE results for the 8-qubit, 30-layer indirect-control VQE estimating $\langle H_{\text{Ising}} \rangle$. The ansatz contains 150 variational parameters: 30 time parameters and 120 angle parameters.}}}
\label{table-multi-var-zne}
\setlength{\tabcolsep}{3pt}

\begin{tabular}{|p{71pt}|p{71pt}|p{71pt}|}
\hline

Method & Estimated Value (Mean $\pm$ Std. Dev.) $^{\mathrm{1}}$ & ZNE Sampling Overhead $(c)$ $^{\mathrm{2}}$ \\
\hline
Noise-free estimation & $-9.66 \pm 0.05$ & -- \\
\hline
Unmitigated & $-7.03 \pm 0.16$ & -- \\
\hline
ZNE of order 1 & $-8.82 \pm 0.12$ & $1.75 \times 10^{1}$ \\
ZNE of order 2 & $-9.37 \pm 0.09$ & $3.51 \times 10^{2}$ \\
ZNE of order 3 & $-9.52 \pm 0.08$ & $4.35 \times 10^{3}$ \\
ZNE of order 4 & $-9.54 \pm 0.07$ & $3.87 \times 10^{4}$ \\
ZNE of order 5 & $-9.51 \pm 0.07$ & $3.15 \times 10^{5}$ \\
\hline
\multicolumn{3}{p{215pt}}{All values are expressed in dimensionless units. Due to the large number of data points, we do not include them explicitly here.}\\
\multicolumn{3}{p{215pt}}{\hl{$^{\mathrm{1}}$ The error bars represent the standard deviation computed from 10 independent trials, where each trial consists of VQE optimization with a different random initialization of the variational parameters, followed by ZNE using the optimized parameters.}}\\
\multicolumn{3}{p{215pt}}{\hl{$^{\mathrm{2}}$ Refer to Appendix \ref{sec smapling overhead}}}\\
\hline
\end{tabular}
\end{table}
\hlgreen{We scaled up the system to $8$ qubits. First, we perform the indirect-control VQE in the presence of local depolarizing noise on all gates and obtain the optimized variational parameters. We then perform ZNE using the multivariate noise framework. In this case, unitary folding is applied to all gates, with folding numbers $K_{R}^{(i)}$, $K_{CZ}^{(i)}$, $K_{T}^{(i)}$, and $K_{Y}^{(i)}$ for the $R$, $CZ$, time-evolution, and $Y$ gates, respectively, when executing the $i$-th circuit. In our simulation pipeline, folding is applied at the level of individual gates, allowing different scale factors for each gate type. For illustration, in Figure~\ref{indirect_vqe_redun_ansatz}, a 7-qubit folded circuit is specified by the folding numbers $K_R=1$, $K_{CZ}=1$, $K_T=1$, and $K_Y=1$, indicating that the $R$, $CZ$, time-evolution, and $Y$ gates are each folded once. }

To account for the heterogeneous noise contributions arising from different gate types, we employ a multivariate framework of Richardson extrapolation \cite{lre} by defining a vector-valued proxy noise parameter in terms of the number of noisy gates as

\begin{equation}
    \vec{\Lambda}^{(i)} := \bigl(\mathcal{N}_{R}^{(i)}  + \mathcal{N}_{Y}^{(i)},\, \mathcal{N}_{CZ}^{(i)},\, \mathcal{N}_{T}^{(i)}\bigr),
    \label{eq-vector-noise}
\end{equation}
where \hl{ $\mathcal{N}_{R}^{(i)}$, $\mathcal{N}_{CZ}^{(i)}$, $\mathcal{N}_{T}^{(i)}$, and $\mathcal{N}_{Y}^{(i)}$} denote the number of single-qubit rotation gates, two-qubit $CZ$ gates, time-evolution gates acting on \hl{$8$}-qubit, and $Y$ gates in a single layer, respectively\hl{, when executing $i$-th circuit}.  Here, \hl{$\mathcal{N}_{CZ}^{(i)}$ and  $\mathcal{N}_{T}^{(i)}$} are equivalent to the scale factors of a $CZ$-gate and a time-evolution gate, respectively.   
\hl{Following the notation convention used in Appendix~\ref{sec multi variate richardson}, the subscript $i$ is indexed as $i = 1,2,3,\dots$, with $i = 1$ corresponding to the noise level of the unfolded base circuit.}
A detailed explanation of the multivariate Richardson interpolation is provided in Appendix \ref{sec multi variate richardson}.

We plot the ZNE estimates obtained using different orders of extrapolation under this vector-valued noise model in Figure~\ref{fig-plots-zne-diff-orders-multivariate} (refer to Table \ref{table-multi-var-zne}) which confirms that ZNE can mitigate errors in the expectation value \hl{very close to} the statistical uncertainty \hl{range} of the estimation obtained from indirect-control VQE without noisy gates. Finally, we would like to note the total number of circuits used in the results obtained here. In the multivariate Richardson interpolation method proposed by Russo and Mari, interpolation using a $d$\hl{-th order} $l$-variable polynomial requires $\frac{(d+l)!}{d!l!}$ \hl{(refer to Equation \eqref{eq-lre-m})} data points. Consequently, the ZNE runs shown in Figure~\ref{fig-plots-zne-diff-orders-multivariate}  using 1st, 2nd, 3rd, 4th, and 5th-order Richardson interpolation use $4$, $10$, $20$, $35$, and $56$ data points, respectively. For example, in Figure~\ref{fig-plots-zne-diff-orders-multivariate} , the ZNE using 3rd-order Richardson interpolation yields results that are sufficiently close  to the estimate obtained in the absence of noise \hl{($98\%$ relative accuracy with respect to the noise-free estimate; see Equation~\eqref{eq-relative-accuracy})}; this estimate was obtained using  $20$ different folding circuits.

\Figure[t!](topskip=0pt, botskip=0pt, midskip=0pt){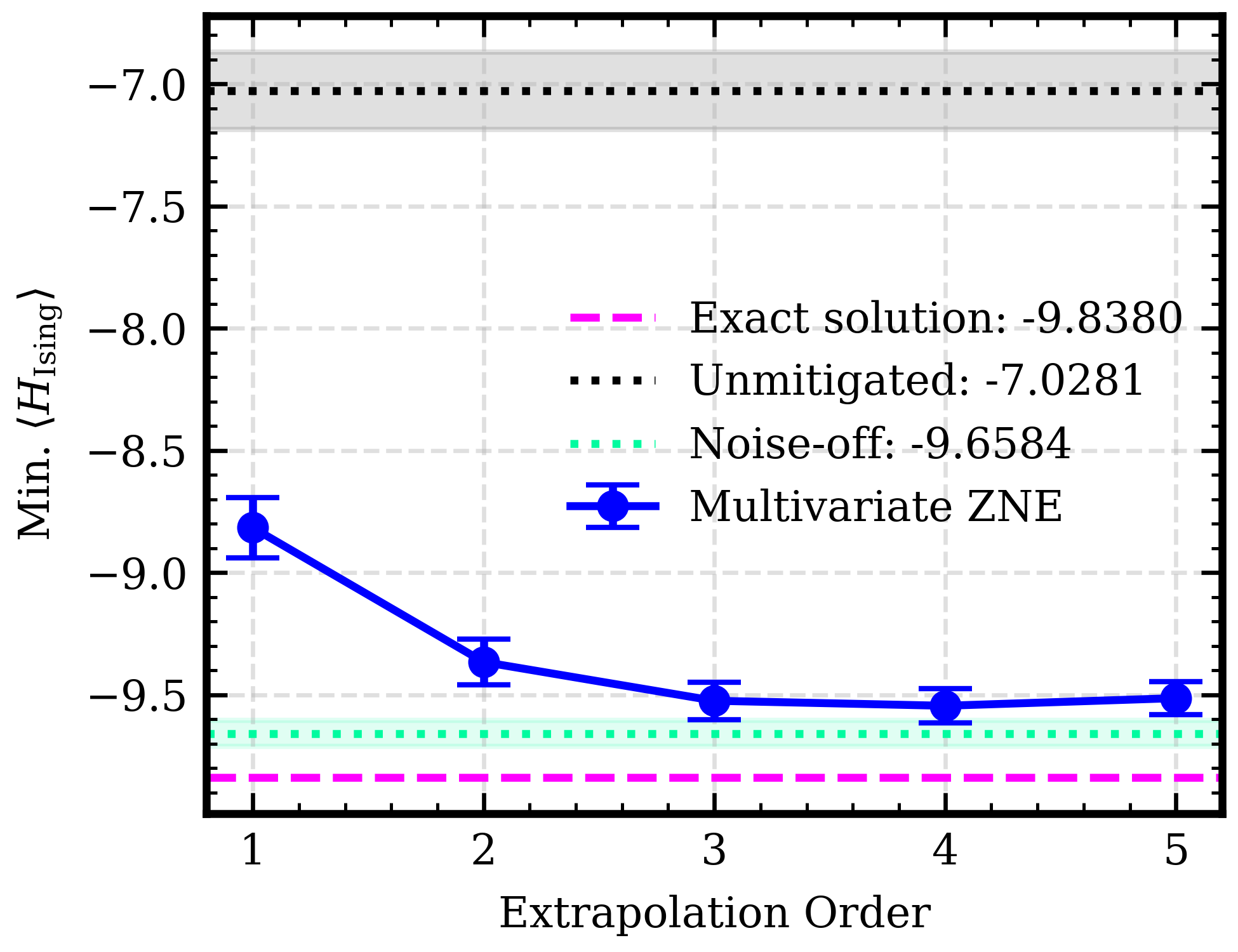}
{\hl{Multivariate Richardson extrapolation of different orders with vector-valued noise. An 8-qubit, 30-layer indirect-control VQE ansatz is used. The minimum expectation value of the 8-qubit $H_{\text{Ising}}$ is estimated. Local depolarizing noise is applied to the entire system with $\gamma = 0.001$ for the time-evolution gate, while the single-qubit gates ($R$, $Y$) are subjected to depolarizing noise with probability $0.001$, and the $CZ$ gate is subjected to local depolarizing noise with probability $0.005$. Initial random time parameters are sampled from the range $[0, T_{\text{max}}=20]$. The time-evolution gates are simulated using $M_n=1$ in Equation \eqref{eq trotterized time-evolution}. The error bars represent the standard deviation computed from 10 independent trials, each consisting of VQE optimization followed by ZNE. The resulting standard deviation reflects the variability arising from the random initialization of the variational parameters.}\label{fig-plots-zne-diff-orders-multivariate}}

\section{\hl{Summary and Conclusion}}\label{conclusion}

Indirect-control systems offer an alternative architecture for quantum computing. In this work, we conducted a comprehensive study of \hl{ZNE} error-mitigation technique for indirect-control VQE. Boosting noise in the ansatz for ZNE by applying unitary folding to the time-evolution gate presented a non-trivial challenge. However, by choosing the XY model as the Hamiltonian for time-evolution and relaxing the definition of indirect-control to allow an appropriate Pauli gate to be applied to odd-numbered qubits before and after the time-evolution gate, it became possible to implement unitary folding effectively. \hl{Although we relaxed the indirect-control constraint by applying $Y$ gates to the odd sites, these gates can, in principle, be selectively addressed by using two types of qubits with different resonance frequencies and applying a resonant pulse for the odd-site qubits to the entire spin chain. The simulation results confirmed that ZNE consistently improved the estimation of the noise-free ground-state energy for both the 7-qubit and 8-qubit transverse-field Ising Hamiltonians using the indirect-control VQE.} 

Unfortunately, the method used in this study clearly relies heavily on the symmetries of the one-dimensional XY model, which is the system Hamiltonian. Therefore, it is not easy to extend this method to quantum many-body systems, such as the Heisenberg model, which possess different symmetries.
Consequently, it is considered difficult to perform ZNE using unitary folding when implementing VQE via indirect-control for general quantum many-body systems that do not possess specific symmetries.
Therefore, to perform ZNE in such cases, it is necessary to employ methods other than unitary folding.
On the other hand, since modifying the drift Hamiltonian itself deviates from the purpose of indirect-control, using pulse stretching on the drift Hamiltonian is likely inappropriate.
Therefore, for quantum many-body systems without special symmetries, exploring methods to implement ZNE using indirect-control via techniques other than unitary folding or pulse stretching, as well as methods to implement probabilistic error cancellation \hl{(PEC) \cite{van_den_Berg_2023}} using indirect-control, is considered an important area for future research.
As a topic for future research,
it would be interesting to investigate other algorithms that can be implemented in indirect-control systems, as well as useful additional error \hl{mitigation} techniques. Furthermore, while implementing indirect-control VQE in experimental setups is valuable, it is first necessary to identify which types of many-body quantum systems are suitable for the indirect-control approach.

In this study, we considered the implementation of ZNE for VQE via indirect-control in the XY model; however, the method presented here can always be adapted to the observable estimates obtained using indirect-control in the XY model. Consequently, it is considered adaptable to variational quantum algorithms that use observable estimates other than those from VQE. 
To this end, it will be important to first conduct research on implementation methods for variational quantum algorithms other than VQE using indirect-control.

\section{\hlred{Further Discussion}}\label{further discussion}

\paragraph{Realistic noise sources}
\hlred{ZNE is a noise-agnostic technique. Our numerical results confirm that ZNE can effectively mitigate local depolarizing noise in indirect-control systems. We have also extended this analysis to local dephasing and local amplitude damping channels, with the corresponding results presented in Appendix~\ref{extending to other noise channels}.}
\begin{table}
\caption{\textbf{Noise types and their impact on indirect-control system.}}
\label{table-hardware-noise}
\setlength{\tabcolsep}{3pt}

\begin{tabular}{|p{65pt}|p{65pt}|p{82pt}|}
\hline
\textbf{Noise Type} & \textbf{Noise Source} & \textbf{Impact on Indirect-Control System} \\
\hline
Local depolarizing & Model noise in this work. & Multivariate ZNE estimates (orders $3$--$5$) agree with the noise-free estimate to $98\%$ relative accuracy.$^1$\\
\hline
Local dephasing$^2$ & Model noise in this work. & Multivariate ZNE estimates (order $3$) agree with the noise-free estimate to $98\%$ relative accuracy. \\
\hline
Local amplitude damping$^2$ & Model noise in this work. & Multivariate ZNE estimates (orders $2,3$) agree with the noise-free estimate to $99\%$ relative accuracy. \\
\hline
\noalign{\hrule height 0.8pt}
Crosstalk\cite{Murali_2020} & Shared microwave/optical control lines. & Would affect the control unit directly; $Y$ gates on odd sites may introduce additional crosstalk during unitary folding. \\
\hline
Leakage \cite{Brown_2019, battistel_mitigating_leakage} & Population transfer to higher levels (e.g., $|2\rangle$). & Corrupts the noise-scaling factor by leaking amplified noise outside the computational subspace. \\
\hline
Coherent errors \cite{Cai_2023} & Pulse over- or under-rotation. & Would affect the rotation gates in control unit. Produces nonlinear noise scaling with unitary folding, reducing ZNE extrapolation accuracy. \\
\hline
\multicolumn{3}{p{215pt}}{$^1$ Relative accuracy is defined in Equation \eqref{eq-relative-accuracy}}\\
\multicolumn{3}{p{215pt}}{$^2$ Refer to Appendix~\ref{extending to other noise channels}.}\\
\hline
\end{tabular}
\end{table}
\hlred{However, the assumptions underlying ZNE, particularly the requirement that the expectation value depends smoothly on the noise level, may not always hold on real quantum hardware, where error sources such as crosstalk, leakage, and coherent control errors may also be present. In the indirect-control setting, these noise sources primarily affect the control unit. Leakage into higher energy levels \cite{battistel_mitigating_leakage} of the control unit, for instance, can invalidate the effective noise-scaling factor assumed by ZNE, as the amplified noise may no longer remain confined to the intended two-level computational subspace spanned by $\ket{0}$ and $\ket{1}$, causing subsequent gate operations to deviate from their intended behavior. Coherent pulse over- or under-rotation errors~\cite{Cai_2023} in the control field can lead to a nonlinear scaling of the error under unitary folding, thereby affecting the accuracy of ZNE. Crosstalk \cite{Murali_2020} arising from shared control lines would in principle affect the control unit directly; additionally, the $Y$ gates used on odd sites during unitary folding may introduce further crosstalk beyond what is present in the unscaled circuit. Addressing these noise sources typically requires complementary error-suppression protocols, such as dynamical decoupling~\cite{Cai_2023} to suppress low-frequency coherent errors, crosstalk-adaptive scheduling (CAS)~\cite{Murali_2020} to suppress crosstalk, Pauli twirling~\cite{Cai_2023} to convert certain coherent crosstalk errors into more stochastic effective noise, and leakage reduction units (LRUs)~\cite{xin2025improvederrorcorrectionleakage} to remove leakage in superconducting quantum processors. However, incorporating such complex noise sources, such as extending the model beyond the two-level computational subspace to account for leakage and introducing physically motivated simulations of coherent errors and crosstalk, as well as adapting the corresponding error-suppression techniques to the indirect-control architecture, requires further investigation and is beyond the scope of the present study.}

\hlred{We would also like to emphasize that a direct experimental assessment of indirect-control VQE (along with ZNE) on quantum hardware is beyond the scope of the present study, as mainstream quantum computing platforms, such as those developed by IBM and Google, are not based on the indirect-control technique. Experimental implementations of indirect-control quantum systems have so far been demonstrated only on a small scale, such as through individual one- and two-qubit gate operations in heteronuclear spin systems~\cite{PhysRevA.91.012330} or a one-qubit quantum memory system~\cite{diamondnv, indirect}, rather than as part of a fully realized variational algorithm such as VQE. We therefore restrict our quantitative studies to indirect-control VQE simulations followed by ZNE under local depolarizing (Section \ref{Numerical simulations}), local dephasing, and local amplitude damping (Appendix \ref{extending to other noise channels}) noise channels. Table~\ref{table-hardware-noise} provides a brief overview of different noise types (both modeled and real) and their potential impact on indirect-control system.}

\paragraph{Contrasting discussion: ZNE, other VQE robustness architectures, and indirect-control}

\hlred{ZNE mitigates errors through classical post-processing and is an appealing technique due to its simplicity of implementation, given that the noise can be controllably scaled. However, its main disadvantage is the exponential growth of the estimator variance, which leads to an increased sampling overhead (see Appendix~\ref{sec smapling overhead}). In addition, bias can arise from a mismatch between the assumed extrapolation model and the true dependence of the observable on the noise strength. On the other hand, various architecture-level strategies have been proposed to address different challenges and performance limitations in VQE, including barren plateaus, optimization difficulties, and sensitivity to noise. For example, gradient evaluation can impose significant computational overhead in VQE (refer to Appendix \ref{comp-cost}). Recent approaches have therefore sought to reduce the dependence on quantum-backend evaluations; for instance, TensorHyper-VQC \cite{Qi_2026} uses low-rank Tensor-Train hypernetworks to generate variational parameters classically, thereby avoiding repeated quantum-backend parameter optimization. Furthermore, noise resilience is  enhanced because gradient updates are  decoupled from quantum measurement stochasticity and hardware imperfections without any additional mitigation overhead \cite{Qi_2026}. Other approaches include symmetry-preserving ansatz preparation with a minimal number of variational parameters \cite{Gard_2020} and symmetry-adapted VQE, which restores the spatial symmetry of the Hamiltonian when the circuit structure breaks this symmetry \cite{PhysRevA.101.052340}. The exponential growth of the Hilbert space and the complexity of gradient estimation also motivate the use of well-prepared initial variational parameters rather than random initialization \cite{McClean_2018, Grant_2019}. These approaches primarily modify the ansatz, parameterization, or optimization procedure to improve trainability, symmetry preservation, or resource efficiency, and may also provide noise resilience. In contrast, the indirect-control approach is a hardware-level control architecture that can improve robustness by restricting external control operations to a small subset of qubits. By restricting external operations to a small subset of qubits while allowing the entire system to evolve under its intrinsic many-body Hamiltonian, it limits the pathways through which control-induced noise can enter the system. These approaches therefore operate at different levels and are complementary rather than mutually exclusive. Investigating the compatibility and potential benefits of combining these strategies with indirect-control VQE is an interesting direction for future work.}

\section*{Experimental data availability}\label{data-availability}

\hl{Simulation data and plots are available at} \url{https://github.com/arijit-ship/indirect-zne-data/}

\section*{Acknowledgment}\label{Acknowledgment}
The \hl{authors} gratefully acknowledge Yoshifumi Nakata and Yasunari Suzuki for insightful discussions and comments. The \hl{authors} also thank Toshifumi Anan and Mako Honda for discussions related to this research project.

\appendices

\section{\break MultiVariate Richardson \hlgreen{Extrapolation}} \label{sec multi variate richardson}

Russo and Mari~\cite{lre} introduced Layer-wise Richardson Extrapolation (LRE) as a generalization of conventional Richardson extrapolation to a multivariate noise framework. Rather than characterizing the total circuit noise by a single global parameter, LRE assigns an independent noise variable to each circuit layer. Noise is then amplified layer by layer, and the resulting expectation values are combined through a linear combination whose coefficients are derived analytically from multivariate Lagrange interpolation theory. This finer-grained noise characterization has been shown to reduce estimation bias compared to standard single-parameter zero-noise extrapolation.
\subsection{Noise Scaling via Layer-wise Folding}
The noise amplification step in LRE is carried out through layer-wise unitary folding. Given an $n$-qubit circuit $\mathcal{C}$ built from $l$ sequential layers $L_1, L_2, \dots, L_l$, the circuit is written as
\begin{equation}
    \mathcal{C} = L_l L_{l-1} \cdots L_2 L_1.
\end{equation}
To scale noise independently across layers, each circuit execution is characterized by a vector of $l$ scale factors
\begin{equation}
    \vec{\lambda}^{(i)} = \left(\lambda_1^{(i)}, \lambda_2^{(i)}, \dots, \lambda_l^{(i)}\right), \qquad \lambda_k^{(i)} \geq 1,
\end{equation}
where $\lambda_k^{(i)}$ controls how aggressively layer $L_k$ is folded in the $i$-th circuit execution. A collection of $N$ such vectors is assembled into the set
\begin{equation}
    \Lambda = \left\{ \vec{\lambda}^{(1)}, \vec{\lambda}^{(2)}, \dots, \vec{\lambda}^{(N)} \right\},
\end{equation}
and each vector generates a distinct noise-scaled circuit, yielding the family
\begin{equation}
    \mathcal{C}^{\Lambda} = \left\{ \mathcal{C}^{\vec{\lambda}^{(1)}}, \mathcal{C}^{\vec{\lambda}^{(2)}}, \dots, \mathcal{C}^{\vec{\lambda}^{(N)}} \right\}.
\end{equation}
Executing every circuit in $\mathcal{C}^{\Lambda}$ and measuring the target observable $O$ yields one expectation value per noise configuration. Stacking these into a column vector gives
\begin{equation}
    \mathbf{z} = \left( \langle O(\vec{\lambda}^{(1)})\rangle,\; \langle O(\vec{\lambda}^{(2)})\rangle,\; \dots,\; \langle O(\vec{\lambda}^{(N)})\rangle \right)^{\!T}.
    \label{eq-lre-vec-exp-val}
\end{equation}

\subsection{Extrapolation via multivariate Lagrange Interpolation}
The post-processing step treats the expectation value as an unknown multivariate polynomial in the $l$ noise variables and exploits Lagrange interpolation to evaluate that polynomial at the zero-noise point $\vec{\lambda} = \vec{0}$. The expectation value is modelled by the most general $l$-variable polynomial of total degree at most $d$,
\begin{equation}
    P(\vec{\lambda}) = \sum_{j=1}^{M} c_j \, \mathcal{M}_j(\vec{\lambda}),
    \label{eq-lre-poly}
\end{equation}
where each $\mathcal{M}_j(\vec{\lambda})$ is a monomial in $l$ variables of total degree $\leq d$, and $c_j \in \mathbb{R}$ are unknown real coefficients. The total number of such monomials is
\begin{equation}
    M = \binom{d + l}{d} = \frac{(d + l)!}{d!\,l!}.
    \label{eq-lre-m}
\end{equation}
As an illustration, for $l = 2$ variables $\vec{\lambda} = (\lambda_1, \lambda_2)$ and degree $d = 2$, one obtains $M = 6$ monomials ordered in graded lexicographic order:
\begin{equation}
    \mathcal{M}(\vec{\lambda},\, d=2) = \{1,\; \lambda_1,\; \lambda_2,\; \lambda_1^2,\; \lambda_1\lambda_2,\; \lambda_2^2\}.
\end{equation}
The monomials are ordered so that the constant term $1$ occupies the first position; this convention is important for the zero-noise extrapolation step below.
Evaluating all $M$ monomials at every noise vector in $\Lambda$ defines the $N \times M$ sample matrix

\begin{equation}
    \mathbf{A}(\Lambda, d) = \bigl[a_{ij}\bigr]_{\substack{i=1,\ldots,N \\ j=1,\ldots,M}}, 
    \quad a_{ij} = \mathcal{M}_j(\vec{\lambda}^{(i)}, d).
\end{equation}
Requiring the polynomial in Equation~\eqref{eq-lre-poly} to pass through all $N$ measured points $\{(\vec{\lambda}^{(i)},\, \langle O(\vec{\lambda}^{(i)})\rangle)\}$ translates the fitting problem into the linear system
\begin{equation}
    \mathbf{A}\,\mathbf{c} = \mathbf{z},
    \label{eq-lre-linsys}
\end{equation}
where $\mathbf{c} = (c_1, c_2, \dots, c_M)^T$ collects the polynomial coefficients and $\mathbf{z}$ is the vector of measured expectation values defined in Equation \eqref{eq-lre-vec-exp-val}. For a unique solution to exist, the system must be square with a non-singular coefficient matrix, which requires
\begin{equation}
    N = M \quad \text{and} \quad \det(\mathbf{A}) \neq 0.
\end{equation}
The desired quantity is the zero-noise limit $O_{\text{LRE}} = P(\vec{0}) = c_1$, i.e., the coefficient of the constant monomial. Rather than solving the full system~\eqref{eq-lre-linsys} for all of $\mathbf{c}$, one can recover this single coefficient directly via the multivariate Lagrange interpolation formula~\cite{Saniee}. The interpolating polynomial can be written as
\begin{equation}
    P(\vec{\lambda}) = \sum_{i=1}^{M} \langle O(\vec{\lambda}^{(i)}) \rangle \, \frac{\det\!\left(\mathbf{M}_i(\vec{\lambda})\right)}{\det(\mathbf{A})},
\end{equation}
where $\mathbf{M}_i(\vec{\lambda})$ is the matrix formed by taking $\mathbf{A}$ and replacing its $i$-th row with the monomial row $(\mathcal{M}_1(\vec{\lambda},d),\dots,\mathcal{M}_M(\vec{\lambda},d))$ evaluated at the query point $\vec{\lambda}$. This construction guarantees exact interpolation: evaluating at any $\vec{\lambda}^{(j)} \in \Lambda$ collapses the determinant ratio to the Kronecker delta $\delta_{ij}$, since a repeated row makes $\det(\mathbf{M}_i(\vec{\lambda}^{(j)})) = 0$ for $i \neq j$ while $\det(\mathbf{M}_j(\vec{\lambda}^{(j)})) = \det(\mathbf{A})$.
Setting $\vec{\lambda} = \vec{0}$ to extract the zero-noise limit, all non-constant monomials vanish. With the constant term placed first in the monomial ordering, the $i$-th modified row reduces to the unit vector $\mathbf{e}_1 = (1, 0, \dots, 0)^T$, and the mitigated estimator becomes
\begin{equation}
    O_{\text{LRE}} = \sum_{i=1}^{M} \eta_i \,\langle O(\vec{\lambda}^{(i)}) \rangle,
    \qquad \eta_i = \frac{\det\!\left(\mathbf{M}_i(\vec{0})\right)}{\det(\mathbf{A})}.
    \label{eq-lre-estimator}
\end{equation}
The coefficients $\{\eta_i\}$ depend only on the chosen noise scaling vectors $\Lambda$ and the extrapolation degree $d$; they are entirely independent of the quantum hardware and can be precomputed classically before any circuit is executed.

Although we did not employ the LRE approach for noise scaling in our simulations, the formula in Equation~\eqref{eq-lre-estimator} remains applicable as long as a vector-valued noise parameter can be constructed, as discussed in Equation \eqref{eq-vector-noise}.

\section{\break Perturbative Commutation of Local Depolarizing Noise with Multi-Qubit Unitaries}
\label{appendix local depolarizing}

In this study, we performed simulations under the assumption that the noise entering the ansatz circuit is provided via a local depolarizing channel. In this appendix, we analyze how the scaling law for the noise parameter is derived when unitary folding is applied in such a scenario.

We consider an $N$-qubit circuit $G$ composed of $N$-qubits unitary gates $\{G_i\}_{i=1}^L$:
$G = G_L G_{L-1} \cdots G_1$.
After each gate $G_i$, we insert a local depolarizing channel
with noise parameter $p_i$
\[
{\mathcal D}_i := \bigotimes_{k=1}^{N} {\mathcal D}^{(1)}_{p_i},
\]
where the single-qubit depolarizing channel ${\mathcal D}^{(1)}_{p}$ is given by
Equation \eqref{eq-noise-process}; that is, we allow the noise parameter to depend on $i$. 
Then, the time-evolution on the noisy circuit is given by 
\begin{equation}\label{eq app full noisy circuit}
\mathcal E
:=
\mathcal D_L \circ [G_L]
\circ \cdots
\circ
\mathcal D_1 \circ [G_1].    
\end{equation}
Here, we assume that the total amount of noise entering the circuit is small:
$\sum _{i=1}^L p_i \ll 1$.

Here, we introduce the notations of the so-called Liouville representation.
As is well known, under appropriate normalization, 
the single-qubit Pauli matrices including an identity matrix form an orthonormal basis for the $4$-dimensional 
real Hilbert-Schmidt space consisting of all $2\times 2$ Hermitian matrices.
Here, let $\Pi$ denote the projector onto the 3-dimensional subspace of the Hilbert-Schmidt space consisting of the Pauli $X$, $Y$, and $Z$ matrices. 
Then, a single qubit depolarizing channel can be expressed using $\Pi$ as $\mathcal{D}^{(1)}(p)=\mathcal{I}_1-p\Pi$,
where $\mathcal{I}_1$ is an identity map on a single qubit Hilbert-Schmidt space.
Thus the full $i$-th local depolarizing channel $\mathcal{D}_i$ can be written as 
\[
\mathcal{D} _i
=
\mathcal{I}_N - p_i \sum_{k=1}^N \Pi_k + \mathcal{O}(p_i^2).
\]
where $\Pi_k$ is the projector 
onto Pauli matrices on qubit $k$, and $\mathcal{I}_N$ is 
an identity map on an $N$-qubit Hilbert-Schmidt space.
By defining $\Pi$ as $\Pi = \sum_{k=1}^n \Pi_k
$, we can rewrite the above equation as follows:
\[
\mathcal{D} _i = \mathcal{I}_N - p_i  \Pi + \mathcal{O}(p_i^2).
\]

First, we consider a single unitary $[G_i]$ and
compare ${\mathcal D}_i \circ [G_i]$ and $[G_i] \circ {\mathcal D}_i.
$
Since they can be expanded as 
\begin{align}\label{eq app mathcal D i circ G i}
{\mathcal D}_i \circ [G_i]
=
(\mathcal{I}_N - p_i\Pi)[G_i] + \mathcal{O}(p_i^2),\nonumber \\
[G_i] \circ {\mathcal D}_i
=
[G_i](\mathcal{I}_N - p_i\Pi) + \mathcal{O}(p_i^2).    
\end{align}
Thus, we can evaluate the commutator $[{\mathcal D}_i,[G_i]]$ in the first order of $p_i$ as
\[
[{\mathcal D}_i,[G_i]]
=
-p_i(\Pi[G_i] - [G_i]\Pi)
+ \mathcal{O}(\|\vec{p}\|^2),
\]
where $\vec{p}:=(p_1,\cdots, p_L)$.
Since $[G_i]$ acts as an orthogonal transformation
on the real Hilbert-Schmidt space, $\Pi[G_i]- [G_i]\Pi = \mathcal{O}(1)$.
Hence, we have $[{\mathcal D}_i,[G_i]]= \mathcal{O}(p_i)$, which means
the non-commutativity is first order in $p_i$.

Now, we consider the full circuit given by Equation \eqref{eq app full noisy circuit}.
Using Equation \eqref{eq app mathcal D i circ G i}
and keeping only first order terms,
$\mathcal E$ can be rewritten as
\begin{align}
\mathcal E
&=
[G]
-
 \sum_{i=1}^L
p_i [G_L] \cdots
[G_{i+1}]
\Pi
[G_i]
\cdots
[G_1]
+ \mathcal{O}(\|\vec{p}\|^2) \nonumber \\
&=
\left(\mathcal{I}_N-
 \sum_{i=1}^L
p_i 
\tilde{\Pi}_i
 \right)[G]
+ \mathcal{O}(\|\vec{p}\|^2),
\end{align}
where $\tilde{\Pi}_i$ is given by
\begin{equation}
\tilde{\Pi}_i
=
[G_L] \cdots
[G_{i+1}]
\Pi
[G_{i+1}]^\dagger
\cdots
[G_L]^\dagger.    
\end{equation}
Therefore, in a first-order approximation for parameter $\vec{p}$, the effect of the local depolarizing channel inserted after each unitary gate $G_i$ in a noisy circuit is equivalent to the effect of the following effective noise channel $\mathcal D_{\text{eff}}(\vec{p})$ inserted at the end of the corresponding noise-free circuit:
\begin{equation}
    \mathcal D_{\text{eff}}(\vec{p})
=\mathcal{I}_N-
 \sum_{j=1}^L
p_j 
\tilde{\Pi}_j.
\end{equation}

\subsection{Implication for Unitary Folding}
Next, we consider the implication for unitary folding.
Under $(2K+1)$-fold unitary folding,
$\mathcal{D}_i\circ [G_i]$ is replaced by 
$(\mathcal{D}_i\circ [G_i]^\dagger \circ \mathcal{D}_i\circ [G_i])^K\circ [G_i]$. 
This map can be rewritten as 
\begin{align}
   &(\mathcal{D}_i\circ [G_i]^\dagger \circ \mathcal{D}_i\circ [G_i])^K\circ \mathcal{D}_i[G_i] \nonumber \\
   =&
   \left((\mathcal{I}_N-p_i\Pi) [G_i]^\dagger  (\mathcal{I}_N-p_i\Pi) [G_i]\right)^K
   \nonumber \\
& \qquad   \circ (\mathcal{I}_N-p_i\Pi) [G_i]
+\mathcal{O}(p_i^2)\nonumber \\
=&\left(\mathcal{I}_N
-p_i\left(
(K+1)\Pi+K[G_i]^\dagger \Pi [G_i]\right)\right)[G_i]
+\mathcal{O}(p_i)^2
\end{align}
Then, in this case, the time-evolution $\mathcal{E}_{2K+1}$ on the noisy circuit under $2K+1$-fold unitary folding is given by 
\begin{align}\label{eq app mathcal E unitary folding}
   & \mathcal{E}_{2K+1} \nonumber \\
=&[G]
-\sum_{i=1}^L p_i
[G_L]\cdots [G_{i+1}]\big(
(K+1)\Pi+K[G_i]^\dagger \Pi [G_i]\big)\nonumber \\
&\qquad \qquad \cdot [G_i]\cdots
[G_1]
+\mathcal{O}(p_i)^2\nonumber \\
=& \Big(\mathcal{I}_N
-\sum_{i=1}^L
p_i((K+1)\tilde{\Pi}_i + K\breve{\Pi}_i)
\Big)[G],
\end{align}
where $\breve{\Pi}_{i}$ is defined by 
\begin{align}
\breve{\Pi}_{i}
=
[G_L] \cdots
[G_{i+1}]
[G_i]^\dagger\Pi[G_i]
[G_{i+1}]^\dagger
\cdots
[G_L]^\dagger.    
\end{align}

Equation \eqref{eq app mathcal E unitary folding} shows that 
in a first-order approximation for parameter $\vec{p}$, the effect of the local depolarizing channel in the noisy circuit with $2K+1$-fold unitary folding is equivalent to the effect of the following effective noise channel ${\mathcal D}_{\text{eff},K}(\vec{p})$ inserted at the end of the corresponding noise-free circuit:
\begin{equation}
{\mathcal D}_{\text{eff},K}(\vec{p})=    \mathcal{I}_N
-\sum_{i=1}^L
p_i((K+1)\tilde{\Pi}_i + K\breve{\Pi}_i).
\end{equation}

Thus, the effect of local depolarizing noise on the circuit without folding is given by ${\mathcal D}_{\text{eff}}(\vec{p})$, and that on the circuit with $2K+1$ folding is given by 
${\mathcal D}_{\text{eff},K}(\vec{p})$.
However, it is clear that, in general, ${\mathcal D}_{\text{eff},K}(\vec{p})$ is not equal to 
${\mathcal D}_{\text{eff}}((2K+1)\vec{p})$.
Taking into account that 
both ${\mathcal D}_{\text{eff},K}(\vec{p})$
and ${\mathcal D}_{\text{eff}}((2K+1)\vec{p})$
contain only first-order terms in $\vec{p}$,
we can conclude that the noise strength does not scale
linearly with the number of insertions even in the first order of noise parameter $\vec{p}$ in our problem setting, where the ansatz circuit consisting of $N$-qubit gates and noise is given by the local depolarizing channel.

\section{\break \hlred{Quantitative analysis}} \label{sec Quantitative analysis}
\hlred{
To perform multivariate ZNE at a fixed extrapolation order, we require $M$ circuits (refer to Equation~\eqref{eq-lre-m}). In this section, we provide an additional quantitative analysis of the overhead associated with the execution of these additional circuits for error mitigation of the indirect-control VQE.
}

\subsection{\hlred{Sampling Overhead for Multivariate ZNE}} \label{sec smapling overhead}
\hlred{All simulations presented in this study were performed using density matrices, and the expectation values were evaluated deterministically from the resulting noisy quantum states, without finite-shot sampling. However, on real quantum hardware, the number of measurement shots required to achieve a fixed statistical error generally increases with the extrapolation order.
Let $s_u$ denote the number of shots required to estimate the unmitigated expectation value of a target observable using the noisy VQE estimator. For a given order of ZNE, to evaluate Equation \eqref{eq-lre-estimator}, the total number of shots, $s_{\mathrm{tot}}$, required to achieve the same statistical error is greater than $s_u$. The sampling overhead of ZNE is defined as the ratio $s_{\mathrm{tot}}/s_u$ \cite{lre}.

From Equation~\eqref{eq-lre-estimator}, the variance of ZNE estimation is
}
\begin{equation}
    \operatorname{Var}[O_{\mathrm{LRE}}]
    =
    \sum_{i=1}^{M}
    \eta_i^2\,
    \operatorname{Var}\!\left[\langle O(\vec{\lambda}^{(i)}) \rangle\right].
    \label{eq-lre-var}
\end{equation}
\hlred{
Assuming that $\operatorname{Var}[\langle O(\vec{\lambda}^{(i)}) \rangle]$ is identical for all circuits, and that each of the $M$ circuits is measured using $s_\mathrm{eq} = s_{\mathrm{tot}}/M$ shots, the sampling overhead is given by \cite{lre,Cai_2023}}
\begin{equation}
    c := \frac{s_{\text{tot}}}{s_{\text{u}}} = M\ \Gamma^2, \quad \Gamma := \left( \sum_{i=1}^M |\eta_i|^2 \right)^{\frac{1}{2}}, \quad s_\mathrm{eq} = \frac{s_{\text{tot}}}{M}.
    \label{eq-lre-sampling-c-bar}
\end{equation}
\hlred{
In \cite{lre}, an optimal shot-allocation strategy was proposed to further reduce the sampling overhead. Rather than distributing the measurement budget equally among the $M$ circuits, the number of shots allocated to the $i$-th circuit is chosen to be proportional to the magnitude of its extrapolation coefficient, i.e. $s_i \propto |\eta_i|$. Now the optimal sampling overhead can be written as \cite{lre, Cai_2023, PhysRevLett.119.180509}}
\begin{equation}
    c_{\mathrm{opt}} := \frac{s_{\text{tot}}}{s_{\text{u}}} = \Gamma_{\mathrm{opt}}^2, \quad \Gamma_{\mathrm{opt}} := \sum_{i=1}^M |\eta_i|, \quad s_i = \frac{s_{\text{tot}} |\eta_i|}{\Gamma_{\mathrm{opt}}}.
    \label{eq-lre-sampling-c}
\end{equation}
\hlred{
We note that the unmitigated VQE expectation value obtained using $s_u$ shots cannot be directly used as a data point for ZNE. Instead, at the final converged variational parameters, the unfolded circuit must be re-executed with an increased shot allocation: $s_{\mathrm{eq}} = \Gamma^2 s_u$ per circuit under equal shot allocation, or $s_1 = \Gamma_{\mathrm{opt}}|\eta_1| s_u$ under optimal allocation, to satisfy the total sampling budget requirements ($s_{\mathrm{tot}} =c s_u$ or $s_{\mathrm{tot}} = c_{\mathrm{opt}} s_u$) established in Equations~\eqref{eq-lre-sampling-c-bar} and \eqref{eq-lre-sampling-c}. 

Figure~\ref{fig-plots-zne-overhead-multivariate} (see Table~\ref{table-multi-var-zne}) illustrates the ZNE sampling overhead for an 8-qubit, 30-layer indirect-control ansatz with the Ising Hamiltonian, $H_{\text{Ising}}$, as the target observable under a local depolarizing noise channel. }

\Figure[t!](topskip=0pt, botskip=0pt, midskip=0pt){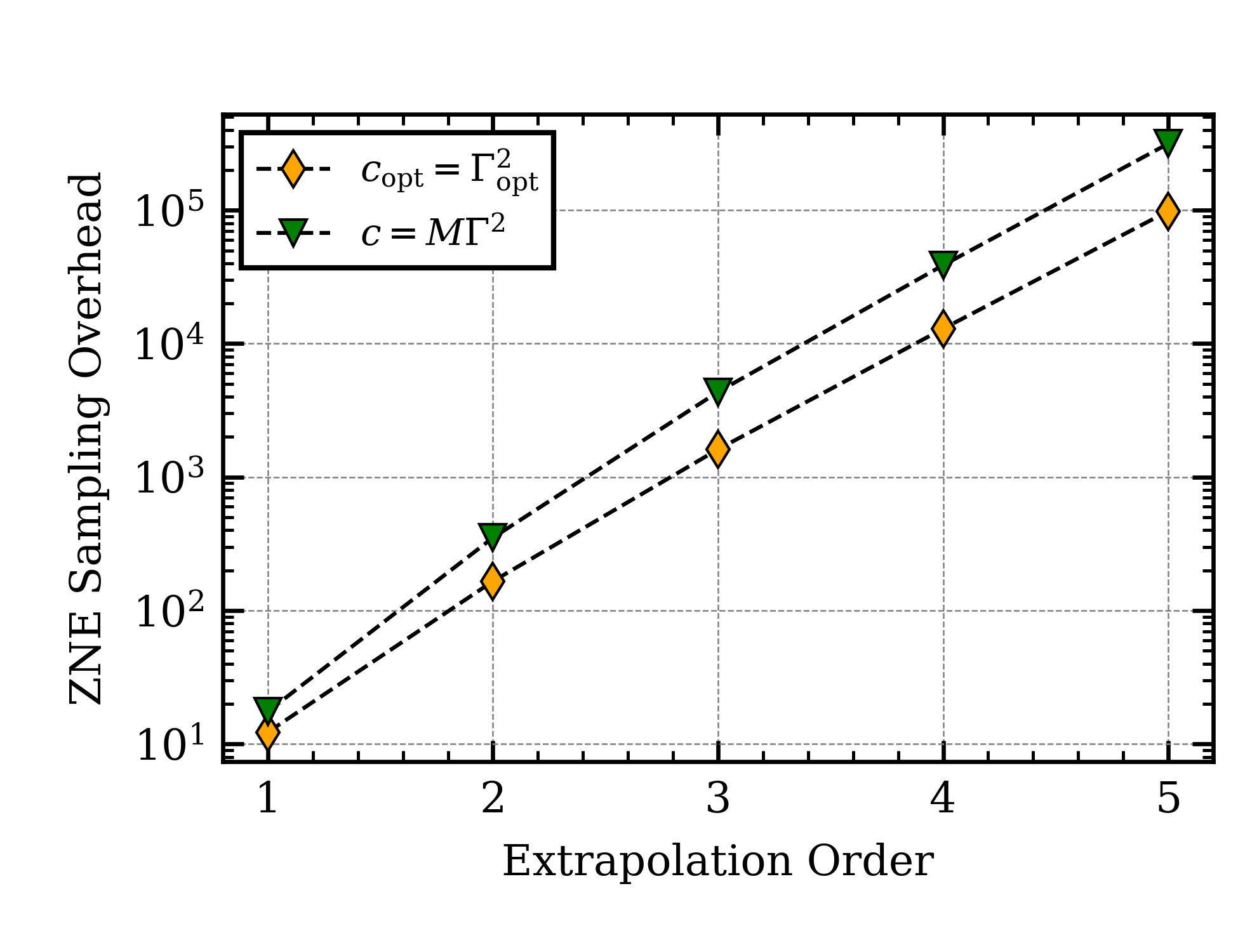}
{\hlred{Exponential scaling of the sampling overhead with the multivariate ZNE extrapolation order for an $8$-qubit, 30-layer system. The minimum expectation value of the $8$-qubit $H_{\text{Ising}}$ is estimated.} \label{fig-plots-zne-overhead-multivariate}}

\subsection{\hlred{ZNE Circuit-layer Depth and Gate Overhead}}
\Figure[t!](topskip=0pt, botskip=0pt, midskip=0pt){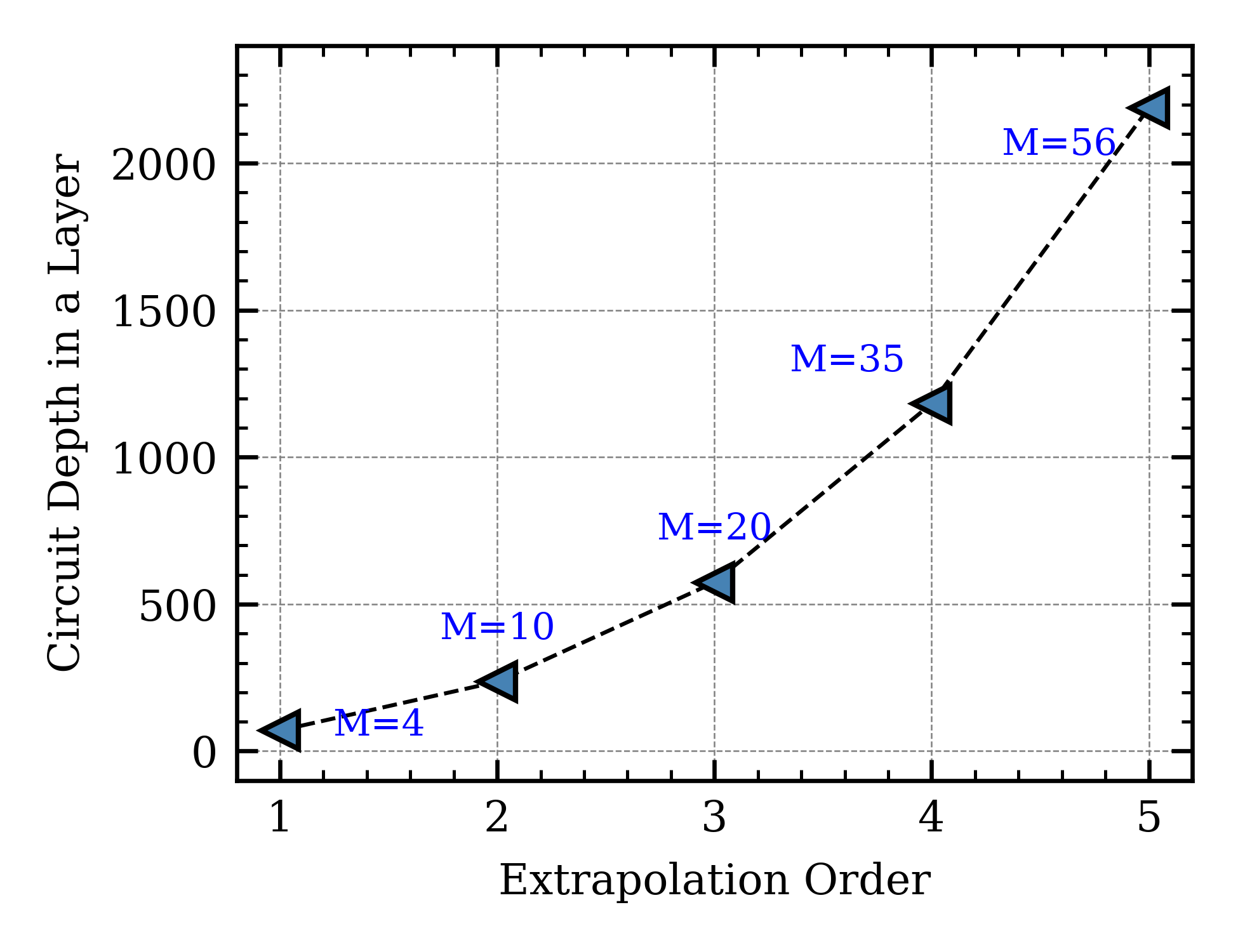}
{\hlred{Circuit-layer depth analysis for multivariate ZNE for 8-qubit, 30-layer system. The annotation $M$ denotes the total number of circuits required for the ZNE algorithm at a fixed extrapolation order.}\label{fig-plots-zne-circ-depth-multivariate}}

\hlred{
Our ansatz consists of structurally identical layers. The depth of layer is determined by the number of sequential gate operations, with gates acting on different qubits executed simultaneously. For example, the unfolded circuit in Figure~\ref{indirect_vqe_redun_ansatz} has a depth of 4 per ansatz layer, whereas the folded circuit next to it has a depth of 18. The total circuit-layer depth across all $M$ circuits (refer to Equation \eqref{eq-lre-m}) is given by}
\begin{equation}
    d_{\text{layer}} = \sum_{i=1}^{M} \left(
    \frac{\mathcal{N}_{R}^{(i)}}{2}
    + \mathcal{N}_{CZ}^{(i)}
    + \mathcal{N}_{T}^{(i)}
    + \frac{\mathcal{N}_{Y}^{(i)}}{n_{\mathrm{odd}}}
    \right),
    \label{eq-circuit-layer-depth}
\end{equation}
\hlred{
where $\mathcal{N}_{R}^{(i)}$, $\mathcal{N}_{CZ}^{(i)}$, $\mathcal{N}_{T}^{(i)}$, and $\mathcal{N}_{Y}^{(i)}$ denote the number of single-qubit rotation gates, two-qubit $CZ$ gates, time-evolution gates acting on $8$-qubit, and $Y$ gates in a single layer, respectively, when executing $i$-th circuit. The factor of $1/2$ accounts for the parallel execution of $R$ gates on the control unit qubits, while the division by $n_{\mathrm{odd}}$ accounts for the simultaneous application of $Y$ gates to all odd-numbered qubits. %
Figure~\ref{fig-plots-zne-circ-depth-multivariate} shows how the circuit-layer depth increases with the ZNE extrapolation order. Additionally, Figure~\ref{fig-plots-zne-circ-gate-op-multivariate} illustrates the total number of gate operations (per circuit layer) across all $M$ circuits for a given order of ZNE. }
\Figure[t!](topskip=0pt, botskip=0pt, midskip=0pt){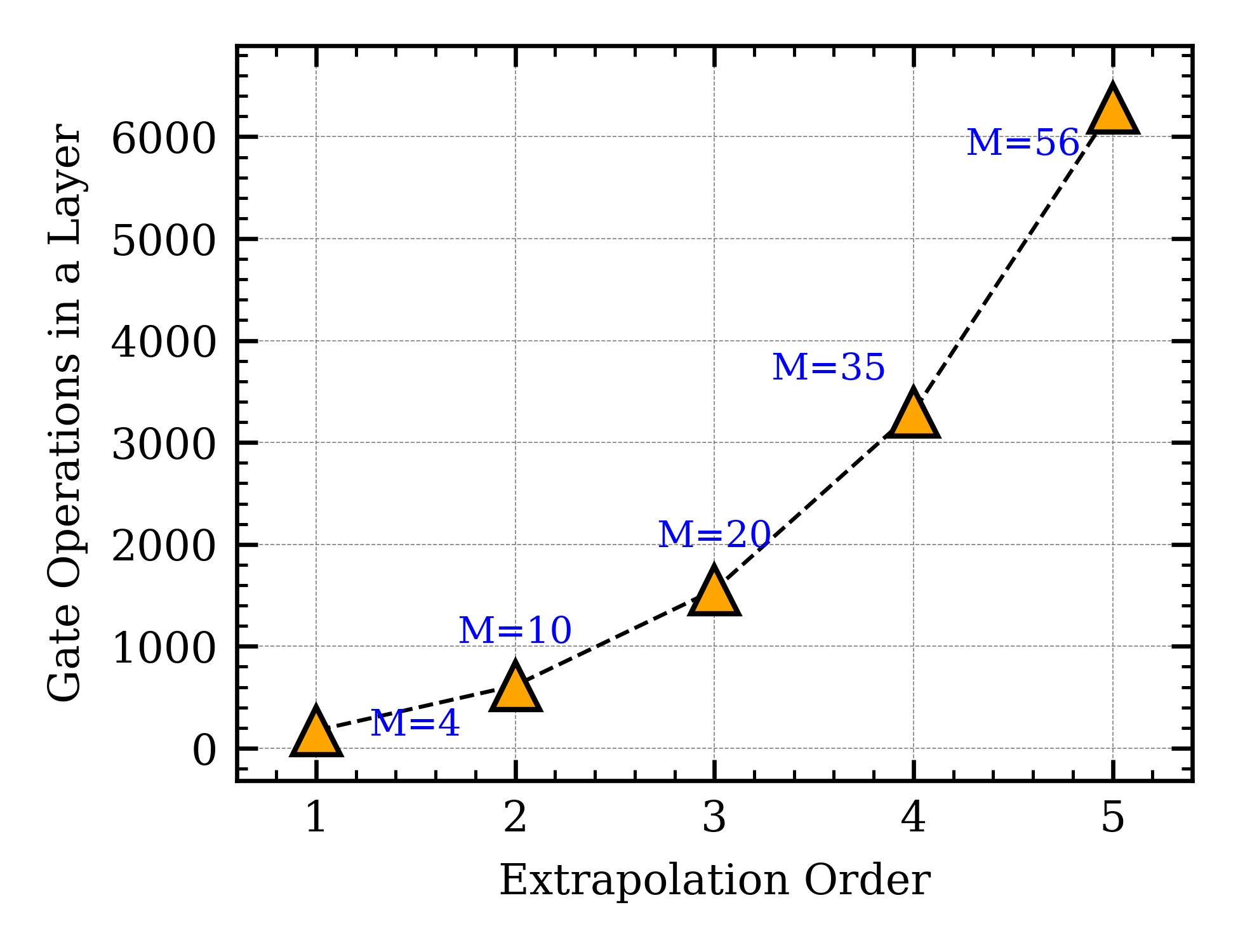}
{\hlred{Number of gate operations per layer summed over all $M$ circuits at a fixed extrapolation order for multivariate ZNE in an $8$-qubit, 30-layer system.}\label{fig-plots-zne-circ-gate-op-multivariate}}

\subsection{\hlred{Runtime}}
\hlred{We assume that the $R$, $Y$, and $CZ$ gate operations are instantaneous, so the total duration of the time-evolution operations determines the runtime. Unitary folding increases this duration by adding time-evolution gates. Since the circuit layers are structurally identical, the folding number $K_T^{(i)}$ for the time-evolution gates is the same across all layers of the $i$-th circuit, resulting in $2K_T^{(i)}+1$ time-evolution gates per layer. Thus, the total runtime can be expressed in terms of the number of time-evolution gates per layer, and the total runtime across all $M$ circuits required for ZNE at a given extrapolation order is defined as}
\begin{equation}
T_{\text{Run}}
= T_{\text{max}}^{\text{opt}} \mathcal{G_P} c s_u \sum_{i=1}^{M} \mathcal{N}_{T}^{(i)}
\label{eq-runtime}
\end{equation}
\hlred{where $T_{\text{max}}^{\text{opt}}$ is the optimized total evolution time (i.e., the final time parameter in the optimized VQE parameter set), $\mathcal{G_P}$ represents the number of measurement groupings of Pauli components for the target observable ($\mathcal{G}_P = 2$ for $H_{\text{Ising}}$, as all $X_k X_{k+1}$ terms and all $Z_k$ terms can be measured simultaneously in two orthogonal bases), $c$ is the ZNE sampling overhead defined in Equation~\eqref{eq-lre-sampling-c-bar}, $s_u$ is the number of measurement shots required for the unmitigated noisy VQE estimator, and $\mathcal{N}_{T}^{(i)}$ is the number of time-evolution gates per ansatz layer in the $i$-th circuit. For simplicity, we assume that all $M$ circuits are allocated an equal number of shots (as discussed in Appendix~\ref{sec smapling overhead}). $T_{\text{Run}}$ is a dimensionless quantity in our simulations. Figure~\ref{fig-plots-zne-circ-runtime-multivariate} illustrates the scaling of the runtime with the ZNE extrapolation order for an 8-qubit, 30-layer system.}

\Figure[t!](topskip=0pt, botskip=0pt, midskip=0pt){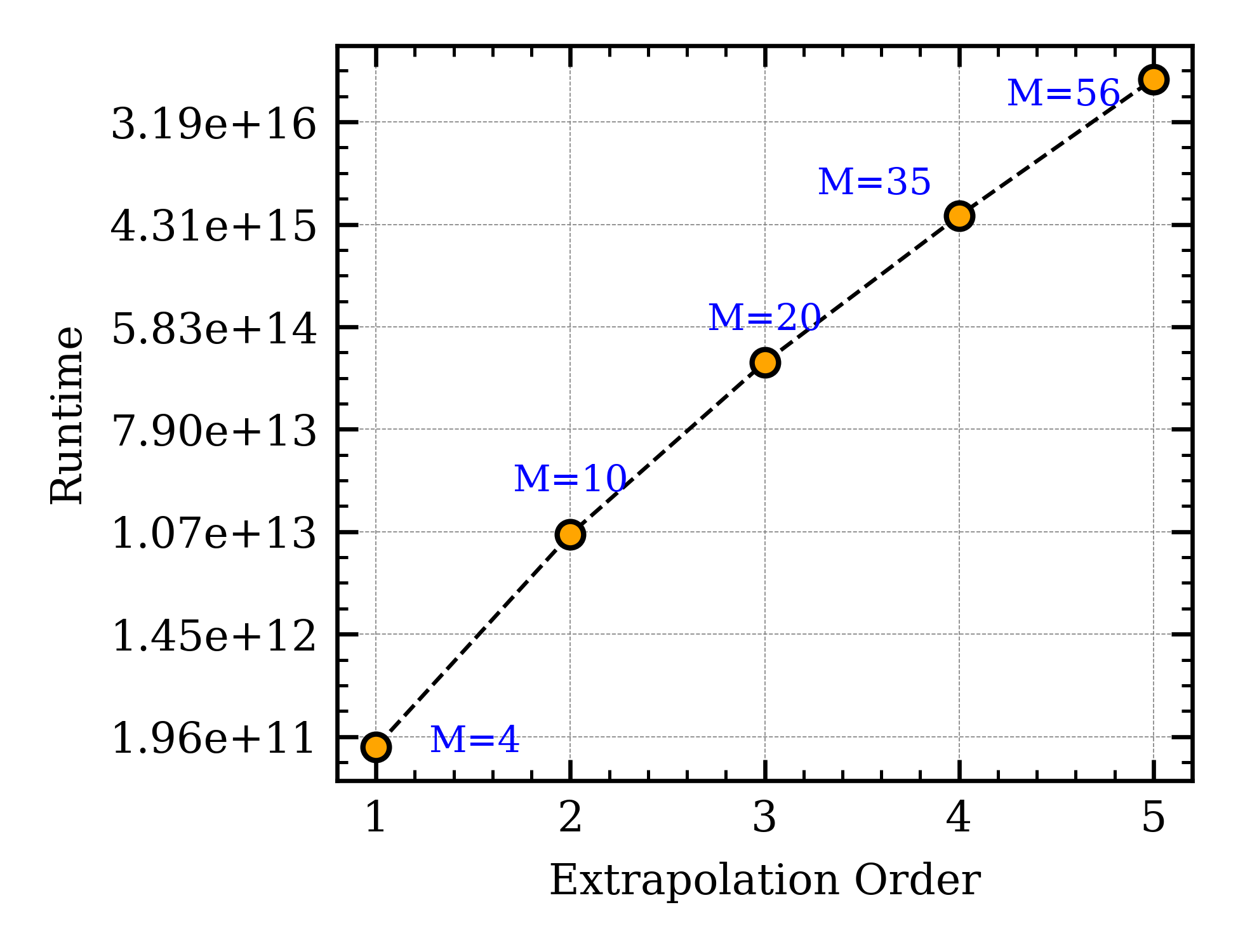}
{\hlred{Runtime versus ZNE extrapolation order for an 8-qubit, 30-layer system under a local depolarizing noise channel. The target observable is the minimum expectation value of the 8-qubit Ising Hamiltonian, $H_{\text{Ising}}$. Local depolarizing noise is applied to the entire system with $\gamma = 0.001$ for the time-evolution gate, while the single-qubit gates ($R$, $Y$) are subjected to depolarizing noise with probability $0.001$, and the $CZ$ gate is subjected to local depolarizing noise with probability $0.005$. Initial random time parameters are sampled from the range $[0, T_{\text{max}}=20]$. The time-evolution gates are simulated using $M_n=1$ in Equation \eqref{eq trotterized time-evolution}. The mean $T^{\text{opt}}_{\text{max}}$ is $19.217$, computed from 10 independent VQE runs with different initializations of the variational parameters. For $H_{\text{Ising}}$, $\mathcal{G_P}=2$ in Equation \eqref{eq-runtime}. The annotation $M$ denotes the total number of circuits required for the ZNE algorithm at a fixed extrapolation order. For a target precision of $\epsilon = 0.001$, the required number of measurement shots is $s_u \approx 1.98 \times 10^{7}$ using Equation \eqref{eq-est-su}. The sampling overhead $c$ is provided in Table~\ref{table-multi-var-zne}.}\label{fig-plots-zne-circ-runtime-multivariate}}

\subsection{\hlred{Error vs. computation complexity}}
\label{comp-cost}

\hlred{The error in error mitigated expectation value is defined as }
\begin{equation}
    E_{\mathrm{err}} = | E_{\mathrm{noise\text{-}off}} - E_{\mathrm{ZNE}} |,
    \label{eq-error-rate}
\end{equation}
\hlred{where $E_{\mathrm{noise\text{-}off}}$ denotes the VQE expectation value obtained in the absence of noise, and $E_{\mathrm{ZNE}}$ denotes the corresponding zero-noise extrapolated estimate for a given extrapolation order. Also, the relative accuracy with respect to $E_{\mathrm{noise\text{-}off}}$ is defined as}
\begin{equation}
    E_{\mathrm{{acc}}} = 100\% \times \left(1 - \frac{|E_{\mathrm{noise\text{-}off}} - E_{\mathrm{ZNE}} |}{|E_{\mathrm{noise\text{-}off}}|}\right).
    \label{eq-relative-accuracy}
\end{equation}
\Figure[t!](topskip=0pt, botskip=0pt, midskip=0pt){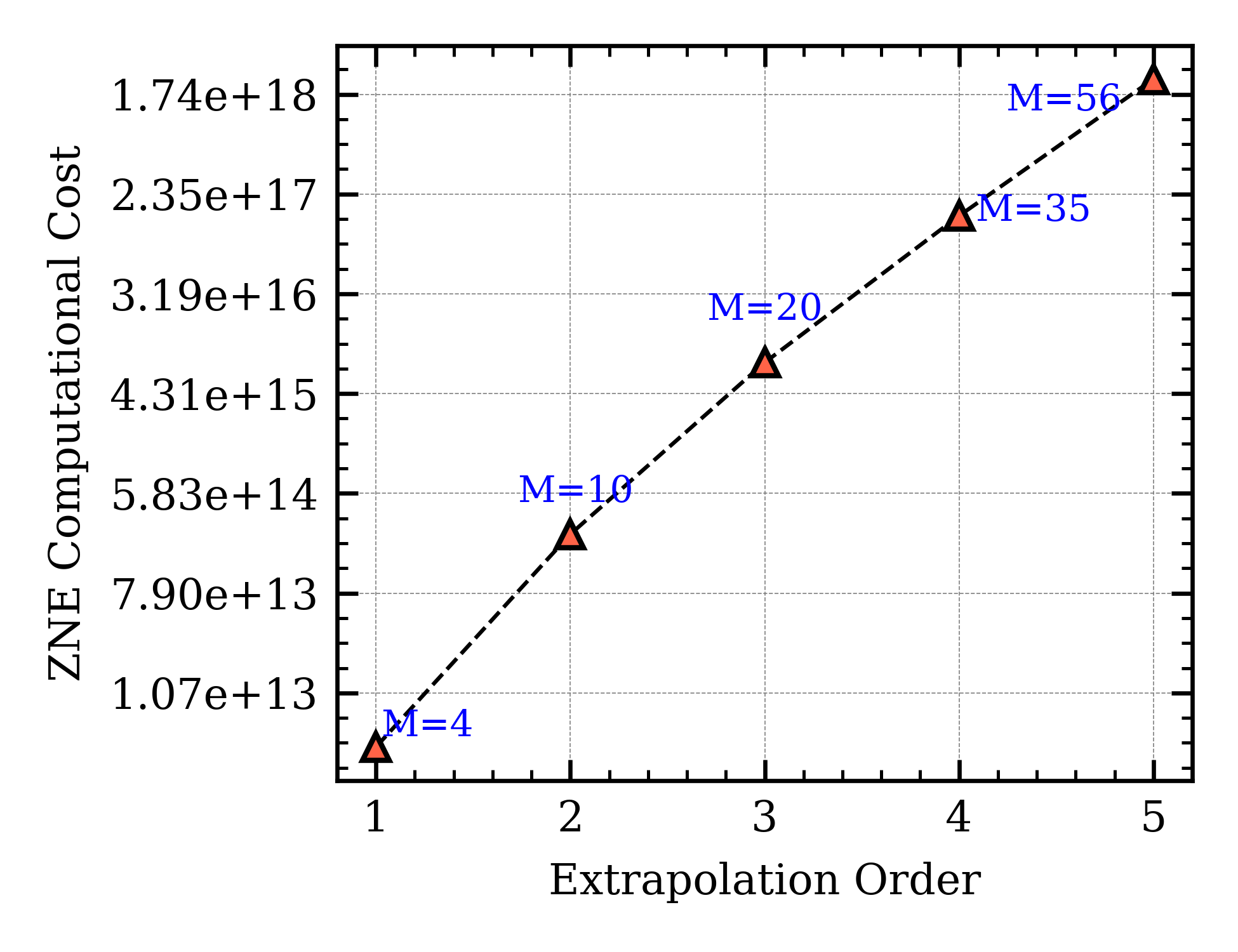}
{\hlred{ZNE computation cost versus ZNE extrapolation order for an 8-qubit, 30-layer system under a local depolarizing noise channel. The target observable is the minimum expectation value of the 8-qubit Ising Hamiltonian, $H_{\text{Ising}}$. Local depolarizing noise is applied to the entire system with $\gamma = 0.001$ for the time-evolution gate, while the single-qubit gates ($R$, $Y$) are subjected to depolarizing noise with probability $0.001$, and the $CZ$ gate is subjected to local depolarizing noise with probability $0.005$. Initial random time parameters are sampled from the range $[0, T_{\text{max}}=20]$. The time-evolution gates are simulated using $M_n=1$ in Equation \eqref{eq trotterized time-evolution}. The annotation $M$ denotes the total number of circuits required for the ZNE algorithm at a fixed extrapolation order. For a target precision of $\epsilon = 0.001$, the required number of measurement shots is $s_u \approx 1.98 \times 10^{7}$. In Equation~\eqref{eq-comp-cost}, $\mathcal{G_P}=2$, overhead $c$ is provided in Table \ref{table-multi-var-zne}.}\label{fig-plots-zne_comp_cost_vs_order}}

\Figure[t!](topskip=0pt, botskip=0pt, midskip=0pt){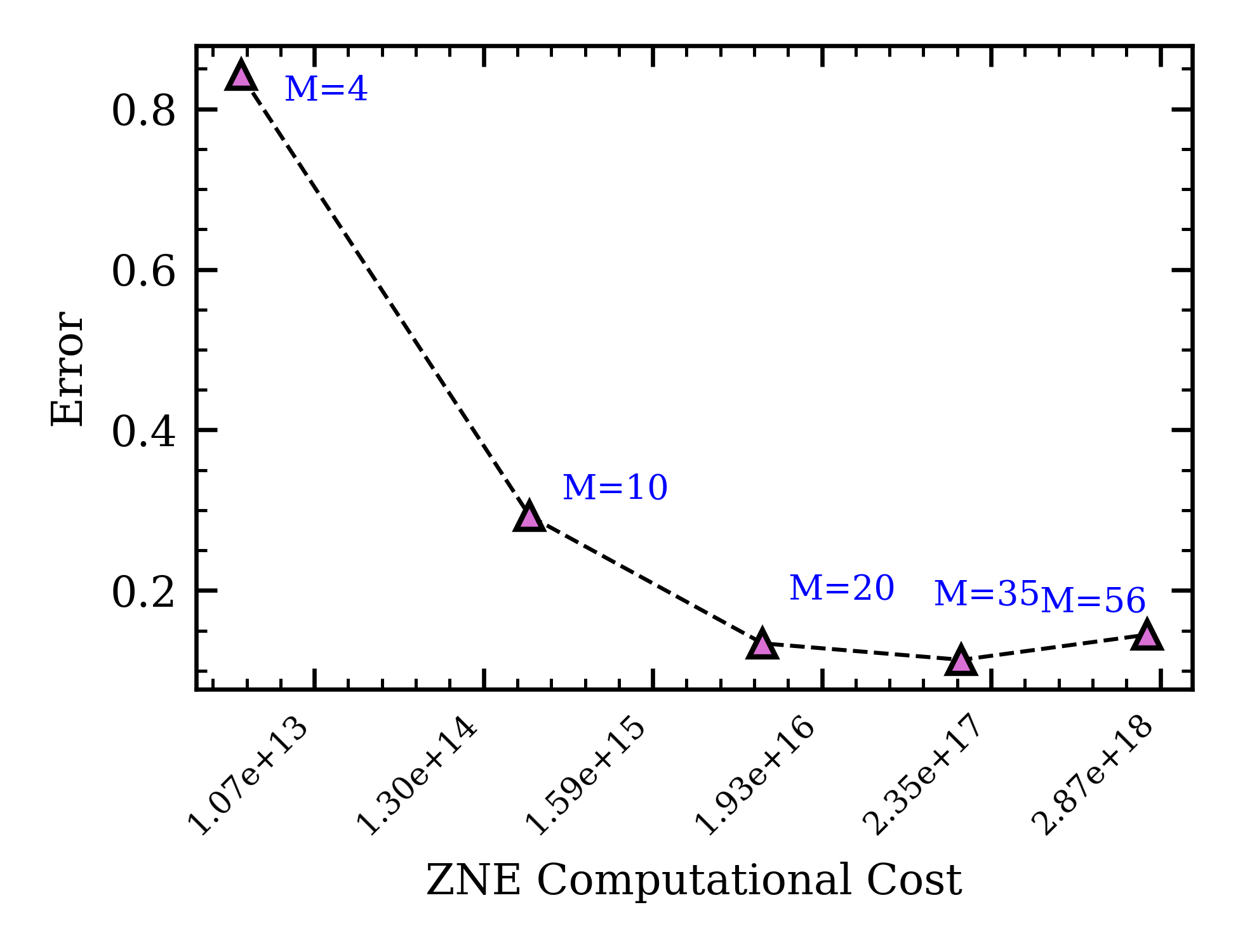}
{\hlred{Error vs. ZNE computational cost for an 8-qubit, 30-layer system under a local depolarizing noise channel. The target observable is the minimum expectation value of the 8-qubit Ising Hamiltonian, $H_{\text{Ising}}$. The annotation $M$ denotes the total number of circuits required for the ZNE algorithm at a fixed extrapolation order.}
\label{fig-plots-zne_comp_cost_vs_error}}
\Figure[t!](topskip=0pt, botskip=0pt, midskip=0pt){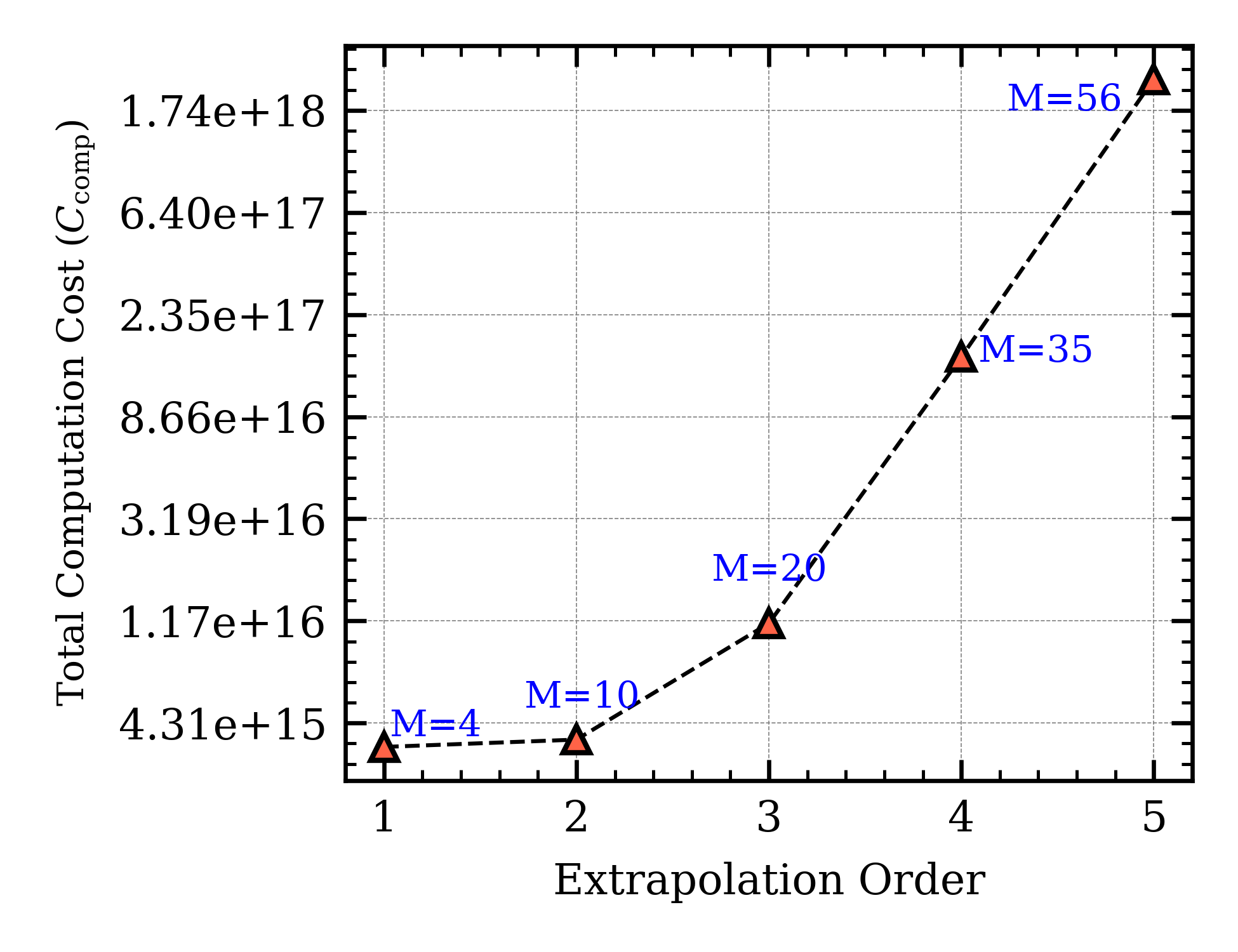}
{\hlred{The total computation cost versus ZNE extrapolation order for an 8-qubit, 30-layer system under a local depolarizing noise channel. The target observable is the minimum expectation value of the 8-qubit Ising Hamiltonian, $H_{\text{Ising}}$. Local depolarizing noise is applied to the entire system with $\gamma = 0.001$ for the time-evolution gate, while the single-qubit gates ($R$, $Y$) are subjected to depolarizing noise with probability $0.001$, and the $CZ$ gate is subjected to local depolarizing noise with probability $0.005$. Initial random time parameters are sampled from the range $[0, T_{\text{max}}=20]$. The time-evolution gates are simulated using $M_n=1$ in Equation \eqref{eq trotterized time-evolution}. The annotation $M$ denotes the total number of circuits required for the ZNE algorithm at a fixed extrapolation order. For a target precision of $\epsilon = 0.001$, the required number of measurement shots is $s_u \approx 1.98 \times 10^{7}$. In Equation~\eqref{eq-comp-cost}, $\mathcal{G_P}=2$, overhead $c$ is provided in Table \ref{table-multi-var-zne}. The mean value of $N_{\mathrm{iter}} = 117.1 \pm 28.29$ is obtained from 10 independent VQE runs with different initial variational parameters. Although the actual iteration count may differ on real hardware, the simulation-derived $N_{\mathrm{iter}}$ serves as a representative proxy for estimating the computational cost.}\label{fig-plots-comp_cost_vs_order}}

\Figure[t!](topskip=0pt, botskip=0pt, midskip=0pt){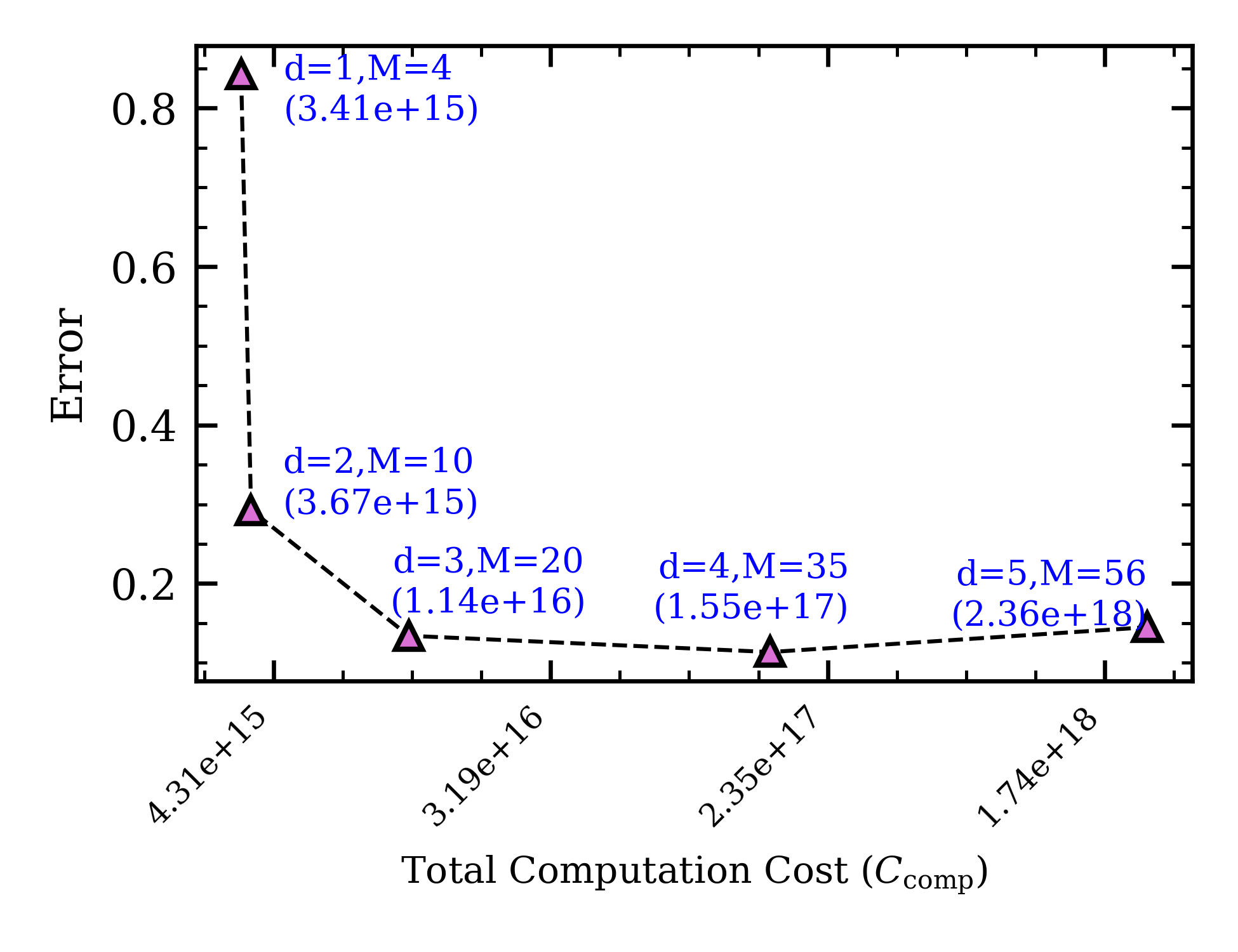}
{\hlred{Error vs. total computational cost for an 8-qubit, 30-layer system under a local depolarizing noise channel. The target observable is the minimum expectation value of the 8-qubit Ising Hamiltonian, $H_{\text{Ising}}$. The annotation $M$ denotes the total number of circuits required for the ZNE at a fixed extrapolation order $d$. The numerical values in parentheses indicate the corresponding total computational cost $C_{\text{Comp}}$.}
\label{fig-plots-comp_cost_vs_error}}

\hlred{To run an indirect-control VQE and subsequently mitigate errors using multivariate ZNE at a given extrapolation order under an equal shot allocation strategy (for the sake of simplicity), we define the total computational cost (in terms of total gate executions) as}
\begin{equation}
C_{\mathrm{Comp}} = \underbrace{N_{\mathrm{iter}} \, Q_{\mathrm{grad}} \, s_u \, \mathcal{N}_{\mathrm{gate}}^{(1)}}_{\text{VQE Optimization Part}} + \underbrace{\mathcal{G_P} \, c \, s_u \sum_{i=1}^{M} \mathcal{N}_{\mathrm{gate}}^{(i)}}_{\text{ZNE Mitigation Part}},
\label{eq-comp-cost}
\end{equation}
\hlred{where $N_{\mathrm{iter}}$ is the total number of classical optimizer iterations required to achieve convergence. We note that $N_{\mathrm{iter}}$ is highly sensitive to the choice of initial variational parameters. $Q_{\mathrm{grad}}$ is the number of hardware circuit evaluations per iteration required for time- and angle-parameter gradient estimations, $s_u$ is the base shot allocation per observable grouping for unmitigated expectation values, and $\mathcal{N}_{\mathrm{gate}}^{(i)}$ denotes the total gate count of the $i$-th scaled circuit across all layers ($i=1$ being the unfolded original circuit). $\mathcal{G_P}$ is the number of measurement groupings. The factor $c = M \Gamma^2$ accounts for the total equal-allocation sampling overhead defined in Equation~\eqref{eq-lre-sampling-c-bar}, such that the total shot budget across all $M$ circuits is $s_{\mathrm{tot}} = c s_u$. In the following paragraphs, we describe in detail how $s_u$ and $Q_{\mathrm{grad}}$ were estimated.}

\paragraph{\hlred{Estimating $s_u$}}
\label{estimating su}
\hlred{Since our simulations were performed using density matrices to compute expectation values deterministically, we estimated the number of shots $s_u$ required for a given target precision $\epsilon$ using all ansatz states encountered during the VQE optimization. To estimate the variance of the target observable, $\langle H_{\text{Ising}} \rangle$, under a local depolarizing noise channel, we first compute the variances of the $XX$ and $Z$ Pauli terms separately across all ansatz output states. The total variance is then}
\[
\operatorname{Var}\!\left[\langle H_{\text{Ising}} \rangle\right]
=
\operatorname{Var}\!\left[ \left\langle \sum_{k=1}^{N-1} X_k X_{k+1} \right\rangle \right]
+
\operatorname{Var}\!\left[ \left\langle \sum_{k=1}^{N} Z_k \right\rangle \right].
\]

\hlred{For a target precision of $\epsilon$, the required number of measurement shots is estimated as}
\begin{equation}
    s_u = \frac{\operatorname{Var}\!\left[\langle H_{\text{Ising}} \rangle\right] }{\epsilon^2}.
    \label{eq-est-su}
\end{equation}

\paragraph{\hlred{Estimating $Q_{\mathrm{grad}}$}}
\hlred{For a standard VQE, the classical optimizer can, in principle, compute the gradient using the parameter-shift rule, provided that the quantum circuit can be executed a sufficient number of times. However, when implementing an indirect-control VQE on real hardware, the conventional parameter-shift rule cannot be applied to the time parameters. The number of quantum circuit evaluations required to compute the gradient of the cost function with respect to all variational parameters is denoted by $Q_{\mathrm{grad}}$:}
\begin{equation}
Q_{\mathrm{grad}} = Q_{t} + Q_{\theta},
\end{equation}
\hlred{where $Q_{t}$ is the number of circuit evaluations required to compute the gradients with respect to the time parameters, and $Q_{\theta}$ is the number of circuit evaluations required to compute the gradients with respect to the angle parameters. Reference \cite{anan_indirect_vqe} provides a method of computing the gradient with respect to time parameters and we can show that for an indirectly controlled system, the number of quantum circuit evaluations required to compute the gradients with respect to the time parameters in a single optimization iteration, $Q_t$, is given by \cite{Suzuki_private}}
\begin{equation}
    Q_t = 2 \times 4^{|C'|} \times \mathcal{G}_P \times P_{\text{time}},
    \label{eq-qt-indirect-measurements}
\end{equation}
\hlred{where $|C'|$ denotes the cardinality of the set of controlled qubits, including the nearest-neighbor qubit. In our setup, because the control unit consists of two qubits located at the boundary of a one-dimensional spin chain (refer to Figure~\ref{indirect-spinchain}), we have $|C'| = 3$. The parameter $\mathcal{G}_P$ represents the number of Pauli measurement groupings required for the target observable. Finally, $P_{\text{time}}$ denotes the total number of VQE time parameters ($P_{\text{time}} = 30$ for a 30-layer system). We note that the exponent term, $4^{|C'|}$, depends solely on the control unit size and remains constant as the total system size scales.

On the other hand, the number of quantum circuit evaluations required to compute the gradients with respect to the angle parameters in a single optimization iteration, $Q_{\theta}$, is calculated using the standard parameter-shift \cite{Wierichs_2022} rule:}
\begin{equation}
    Q_{\theta} = 2 \times P_{\text{angle}} \times \mathcal{G}_P,
    \label{eq-qtheta-indirect-measurements}
\end{equation}
\hlred{where $P_{\text{angle}}$ denotes the total number of VQE angle parameters. For a 30-layer ansatz, $P_{\text{angle}} = 120$ (each layer contains four angle parameters; refer to Figure~\ref{fig-indirect-vqe}).

While $Q_{\mathrm{grad}}$ accounts for the hardware circuit evaluation cost per iteration, we do not have direct access to $N_{\mathrm{iter}}$ from actual hardware executions, that is, the hybrid quantum-classical optimization loop, when computing the total computational cost $C_{\mathrm{Comp}}$ in Equation~\eqref{eq-comp-cost}. Consequently, we adopt the average iteration count obtained from classical simulations as a benchmark proxy. Although the actual value of $N_{\mathrm{iter}}$ in the presence of real quantum hardware may vary depending on the gradient evaluation scheme and implementation-specific factors, the simulation-derived value provides a representative baseline estimate rather than an exact correspondence to hardware execution.

Figures~\ref{fig-plots-zne_comp_cost_vs_order} and~\ref{fig-plots-comp_cost_vs_order} illustrate the rapid growth of the ZNE computational cost and total computational cost, respectively, with increasing extrapolation order. Figure~\ref{fig-plots-zne_comp_cost_vs_error} shows that the estimation error decreases as the ZNE computational cost increases (up to extrapolation order $4$) for an 8-qubit, 30-layer system estimating $\langle H_{\text{Ising}} \rangle$. The corresponding trend with total computational cost is shown in Figure~\ref{fig-plots-comp_cost_vs_error}. A clear trade-off between error mitigation and computational cost emerges from these plots. Moreover, beyond a certain extrapolation order, the estimation error may increase rather than decrease. Therefore, the choice of the ZNE extrapolation order needs to be carefully assessed for practical implementation on real quantum hardware, since high extrapolation orders can introduce numerical instabilities associated with high-order polynomial interpolation.}

\section{\break \hlred{Extension to the Heisenberg Spin Chain Observable and Other Noise Channels}} \label{extending to other noise channels}

\hlred{We extend our simulation further. In this section, additional results using local dephasing and local amplitude damping are presented in an 8-qubit system. Using these two noise models, we estimated the expectation value of the 8-qubit Heisenberg spin chain, $H_{\text{Heisenberg}}$ using the indirect-control ansatz. $H_{\text{Heisenberg}}$ is given by}
\begin{equation}
    H_{\text{Heisenberg}} = \sum_{i=1}^{7} \left( X_i X_{i+1} + Y_i Y_{i+1} + Z_i Z_{i+1} \right)
    \label{eq-heisenberg-spin-chain}
\end{equation}

\hlred{The Trotterization framework presented in Equations \eqref{eq trotterized time-evolution}, \eqref{eq approximation Mn 1} is not restricted to local depolarizing noise and applies equally to local dephasing and local amplitude damping noise. The operator-splitting error in terms of the diamond norm in Equation~\eqref{eq approximation Mn 1} is
 bounded by $\frac{1}{2}\gamma (t_{n+1}-t_n)^2(N-1)c_\text{noise}$, where
$c_\text{noise}$ is a constant depending on both the spin chain Hamiltonian and the noise model.
For local depolarizing noise, local dephasing noise, and local amplitude damping noise, the upper bounds of $c_\text{noise}$ are $4$, $2$, and $6$, respectively for the XY model (see Appendix~\ref{app:trotter_error_explicit} for detail). Under the $M_n=1$ approximation, the layer-dependent noise probability has the same functional form, $p_n = \gamma(t_{n+1}-t_n)$, for all three noise channels, with $p_n^{\text{deph}}$ and $p_n^{\text{AD}}$ defined explicitly below.}

\subsection{\hlred{Dephasing Channel}}
The Lindbladian of the local dephasing channel can be written as $\mathcal{L}_\text{noise}(\rho) =-\gamma \sum_{i=1}^N \Pi_{XY,i}$,
where $\Pi_{XY,i}$ is a projection onto 3-dimensional subspace of a single-qubit Hilbert-Schmidt space spanned by $\{ X_i, Y_i\}$.
Hence, $e^{t\mathcal{L}_\text{noise}}=\bigotimes _i e^{-\gamma t \Pi_{XY,i}}$.
On the other hand, we can easily see that $e^{-\gamma t \Pi_{XY,i}}$ satisfies
\begin{equation}
    e^{-\gamma t \Pi_{XY,i}}=\mathcal{E}^{(1)}_{p}(\rho),
\end{equation}
\hlred{where $\mathcal{E}^{(1)}_{p}$ is 
a single-qubit dephasing channel $\mathcal{E}^{(1)}_p$ with phase-flip probability $p$ acting on a density matrix $\rho$ defined as}
\begin{equation}
    \mathcal{E}^{(1)}_{p}(\rho) = (1-p)\rho + p Z^{\dagger} \rho Z,
    \label{eq-dephasing-channel}
\end{equation}
\hlred{and $p$ is given by}
\begin{equation}
    p=\frac{1-e^{-\gamma t}}{2}.
\end{equation}
\hlred{Thus, when $\gamma t$ is small, $p\approx \frac{\gamma t}{2}$,
and we derive the following relation:}
\begin{equation}
    e^{\mathcal{L}_\text{noise}t}\rho \approx \left(\mathcal{E}^{(1)}_{\gamma t/2}\right)^{\otimes N} (\rho).
\end{equation}
\hlred{Under the single-step ($M_n = 1$) approximation introduced in Equation \eqref{eq approximation Mn 1}, the time-evolution gate in layer $n$ is followed by local dephasing channels with an effective noise probability}
\begin{equation}
    p_n^{\text{deph}} = \frac{\gamma}{2}(t_{n+1} - t_n) = \gamma\prime(t_{n+1} - t_n),
    \label{eq-time-evo-noise-deph}
\end{equation}
where we define a new parameter $\gamma\prime$ by $\gamma\prime:=\frac{\gamma}{2}$.

\subsection{\hlred{Amplitude Damping Channel}}
\hlred{We next consider local amplitude damping, which models spontaneous energy relaxation of each qubit. The Lindblad master equation retains the general form of Equation \eqref{eq lindblad eq}, with the noise superoperator given by}
\begin{equation}
    \mathcal{L}_\text{noise}(\rho) = \gamma\sum_{i=1}^N \mathcal{D}_{\sigma_-^{(i)}}(\rho),\quad
    \mathcal{D}_L(\rho) = L\rho L^\dagger - \tfrac{1}{2}\{L^\dagger L,\rho\},
    \label{eq-AD-lindblad}
\end{equation}
\hlred{where $\sigma_-^{(i)} = |0\rangle\langle 1|$ is the lowering operator acting on qubit $i$, and $\gamma$ denotes the relaxation rate. Unlike depolarizing noise and pure dephasing, amplitude damping is non-unital, driving each qubit towards its ground state $|0\rangle$.}

\hlred{The continuous evolution of a single qubit under $\mathcal{L}_\text{noise}$ yields the standard single-qubit amplitude damping channel $\mathcal{A}^{(1)}_{p(t)}$ \cite{Nielsen_Chuang_2010}:}
\begin{equation}
    \mathcal{A}^{(1)}_{p(t)}(\rho) = E_0 \rho E_0^\dagger + E_1 \rho E_1^\dagger,
    \label{eq-AD-kraus}
\end{equation}
\hlred{where the corresponding Kraus operators are defined as}
\begin{equation}
    E_0 = \begin{pmatrix} 1 & 0 \\ 0 & \sqrt{1-p(t)} \end{pmatrix}, \qquad
    E_1 = \begin{pmatrix} 0 & \sqrt{p(t)} \\ 0 & 0 \end{pmatrix},
    \label{eq-AD-kraus-2}
\end{equation}
\hlred{with dampening probability $p(t) = 1 - e^{-\gamma t} \approx \gamma t$ for $\gamma t \ll 1$. For $N$ independent qubits, the multi-qubit noise process over duration $t$ factors into the tensor product $e^{\mathcal{L}_\text{noise}t} = \left(\mathcal{A}^{(1)}_{p(t)}\right)^{\otimes N}$. Under the single-step ($M_n = 1$) approximation introduced in Equation \eqref{eq approximation Mn 1}, the time-evolution gate in layer $n$ is followed by local amplitude damping channels with effective damping probability}
\begin{equation}
    p^{\text{AD}}_n = \gamma(t_{n+1} - t_n).
    \label{eq-time-evo-noise-AD}
\end{equation}

 \hlred{Figures~\ref{fig-plots-zne-diff-orders-multivariate-dephasing} and \ref{fig-plots-zne-diff-orders-multivariate-ampdamp} demonstrate that the multivariate ZNE approach can significantly mitigate errors arising from local dephasing and local amplitude damping noise in the 8-qubit indirect-control ansatz. Under the local dephasing noise channel, as shown in Figure~\ref{fig-plots-zne-diff-orders-multivariate-dephasing}, the ZNE estimate improves as the extrapolation order increases up to the 3rd order. The 3rd-order multivariate extrapolation achieves an error-mitigated value very close to the noise-free estimate. Under the local amplitude damping noise channel, as shown in Figure~\ref{fig-plots-zne-diff-orders-multivariate-ampdamp}, the 2nd-order extrapolation achieves an error-mitigated value within the statistical uncertainty range of the noise-free estimate. These ZNE values, along with the additional quantitative analyses presented in Appendix~\ref{sec Quantitative analysis} for the two noise channels, are provided in Table~\ref{table-multi-var-zne-full}. While these results highlight the effectiveness of multivariate ZNE for indirect-control VQE, extending error mitigation beyond depolarizing noise introduces additional subtleties. In particular, non-unital and non-Pauli channels, such as amplitude damping, may exhibit nonlinear noise scaling under unitary folding. For local amplitude damping, this dependence may be more complex. Such nonlinear scaling, together with numerical instability in higher-order multivariate polynomial extrapolation, may affect the ZNE estimation.  }
 
\Figure[t!](topskip=0pt, botskip=0pt, midskip=0pt){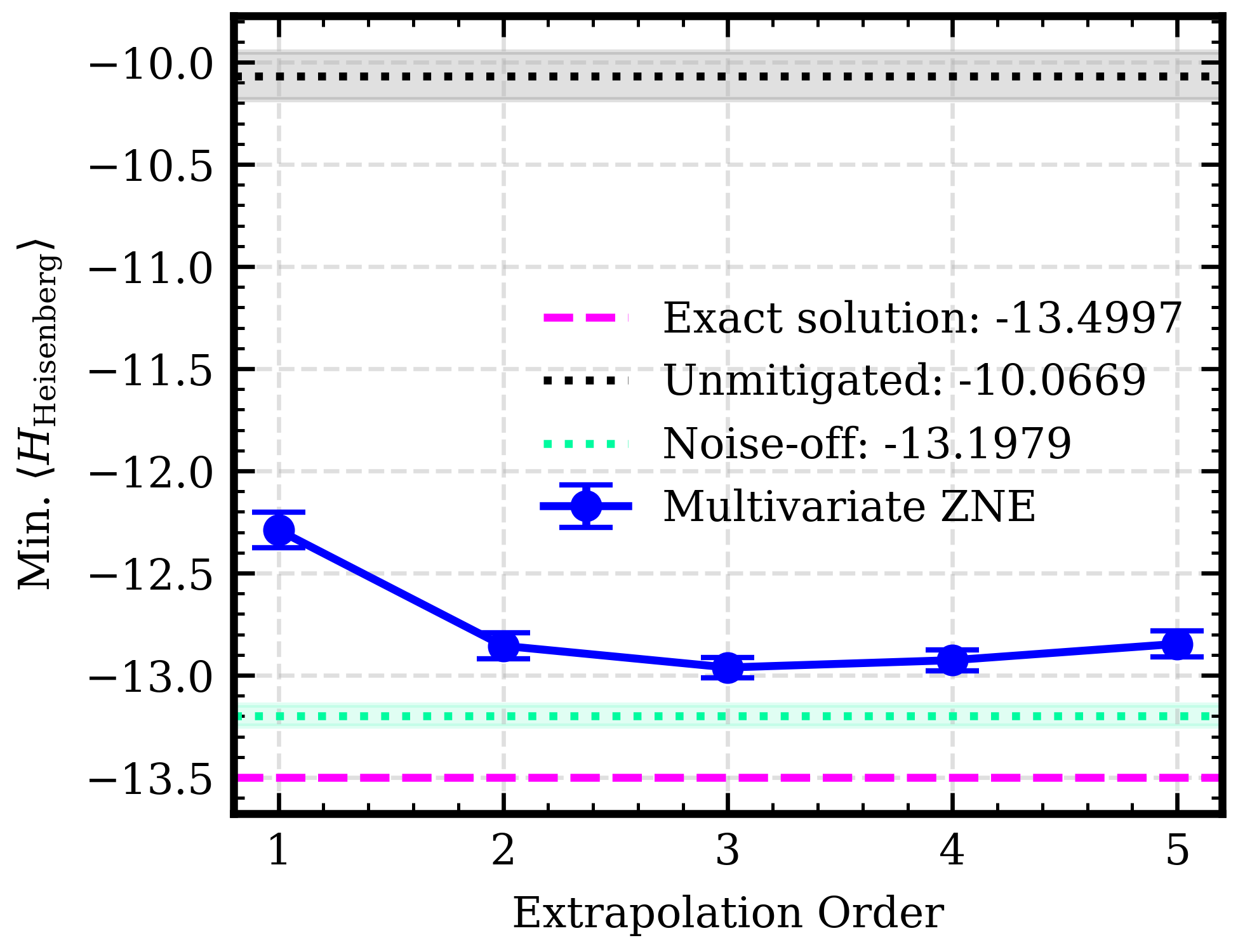}
{\hlred{Multivariate Richardson extrapolation of different orders with vector-valued noise. An 8-qubit, 30-layer indirect-control VQE ansatz is used. The minimum expectation value of the 8-qubit $H_{\text{Heisenberg}}$ is estimated. Local dephasing noise is applied to the entire system with $\gamma\prime = 0.001$ for the time-evolution gate, while the single-qubit gates ($R$, $Y$) are subjected to dephasing noise with probability $0.001$, and the $CZ$ gate is subjected to local dephasing noise with probability $0.005$. Initial random time parameters are sampled from the range $[0, T_{\text{max}}=20]$. The time-evolution gates are simulated using $M_n=1$ in Equation \eqref{eq trotterized time-evolution}. The error bars represent the standard deviation computed from 10 independent trials, each consisting of VQE optimization followed by ZNE. The resulting standard deviation reflects the variability arising from the random initialization of the variational parameters.} \label{fig-plots-zne-diff-orders-multivariate-dephasing}}

\Figure[t!](topskip=0pt, botskip=0pt, midskip=0pt){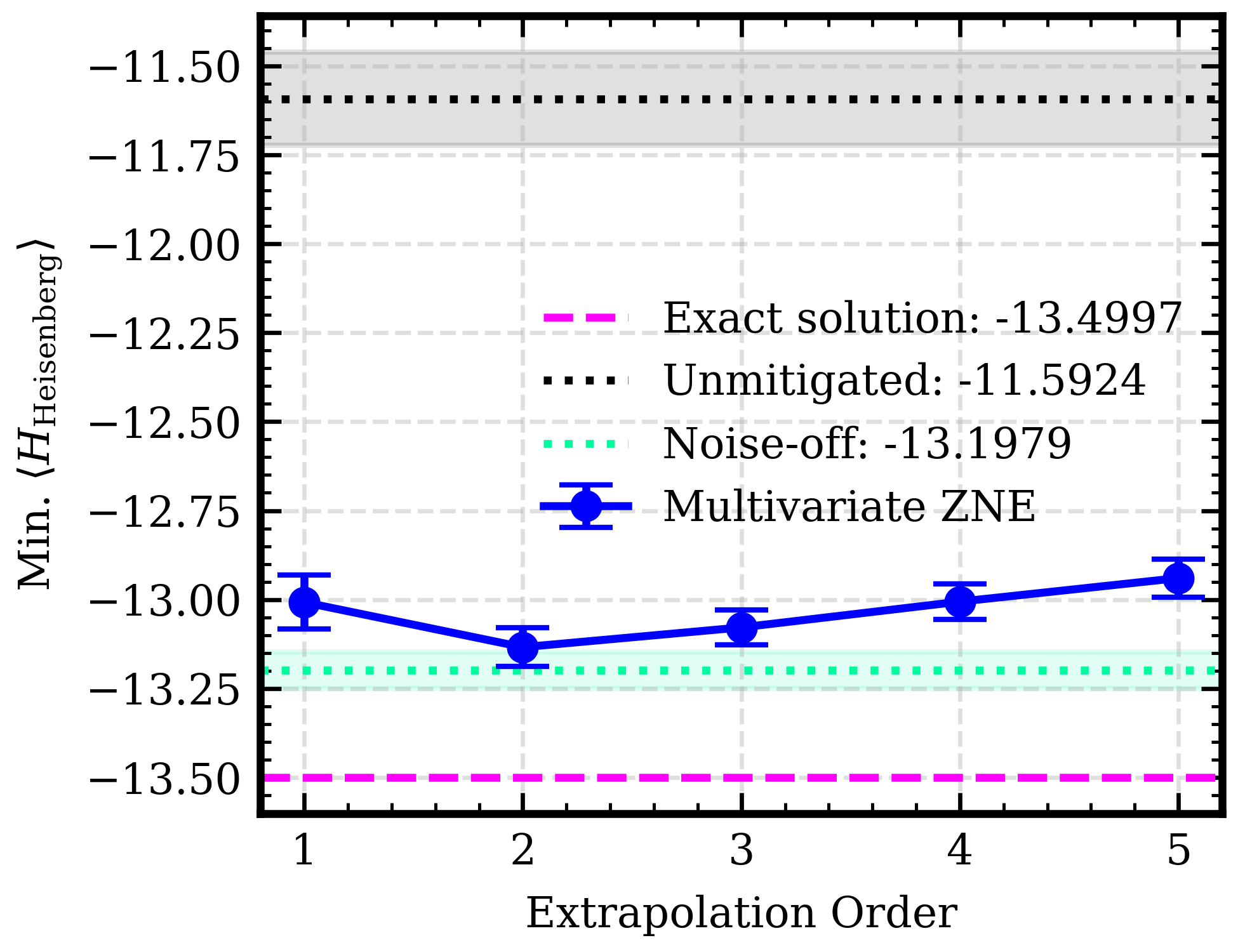}
{\hlred{Multivariate Richardson extrapolation of different orders with vector-valued noise. An 8-qubit, 30-layer indirect-control VQE ansatz is used. The minimum expectation value of the 8-qubit $H_{\text{Heisenberg}}$ is estimated. Local amplitude damping noise is applied to the entire system with $\gamma = 0.001$ for the time-evolution gate, while the single-qubit gates ($R$, $Y$) are subjected to amplitude damping noise with probability $0.001$, and the $CZ$ gate is subjected to local amplitude damping noise with probability $0.005$. Initial random time parameters are sampled from the range $[0, T_{\text{max}}=20]$. The time-evolution gates are simulated using $M_n=1$ in Equation \eqref{eq trotterized time-evolution}. The error bars represent the standard deviation computed from 10 independent trials, each consisting of VQE optimization followed by ZNE. The resulting standard deviation reflects the variability arising from the random initialization of the variational parameters.} \label{fig-plots-zne-diff-orders-multivariate-ampdamp}}

\begin{table*}[t]
\caption{\hlred{\textbf{Summary of multivariate ZNE results for additional noise models. An 8-qubit, 30-layer indirect-control VQE is used to estimate $\langle H_{\text{Heisenberg}} \rangle$. The ansatz contains 150 variational parameters: 30 time parameters and 120 angle parameters. For both local dephasing and amplitude damping, single-qubit gates such as $R$ and $Y$ are assigned a noise probability of $0.001$. We also set $\gamma\prime$ and $\gamma$ to $0.001$ in Equations~\eqref{eq-time-evo-noise-deph} and \eqref{eq-time-evo-noise-AD}, respectively, to account for noise associated with the time-evolution gates. In contrast, the two-qubit $CZ$ gates are assigned a higher noise probability of $0.005$. For $H_{\text{Heisenberg}}$ as the target observable, $\mathcal{G}_P =3$ in Equation \eqref{eq-runtime}, \eqref{eq-qt-indirect-measurements} and \eqref{eq-qtheta-indirect-measurements}.   }}}
\label{table-multi-var-zne-full}
\centering
\small
\setlength{\tabcolsep}{2pt}
\begin{tabular}{|>{\raggedright\arraybackslash}p{64pt}|>{\raggedright\arraybackslash}p{43pt}|>{\centering\arraybackslash}p{60pt}|>{\centering\arraybackslash}p{50pt}|>{\centering\arraybackslash}p{40pt}|>{\centering\arraybackslash}p{50pt}|>{\centering\arraybackslash}p{57pt}|>{\centering\arraybackslash}p{45pt}|>{\centering\arraybackslash}p{27pt}|>{\centering\arraybackslash}p{17pt}|}
\hline
\textbf{Method} & \textbf{Noise Model} & \textbf{Estimated Value Mean $\pm$ Std. Dev.$^{\mathrm{1}}$} & \textbf{Error $E_{\mathrm{err}}$} & \textbf{Accuracy$^{\mathrm{2}}$ $E_{\mathrm{{acc}}}$} & \textbf{Runtime$^{\mathrm{3}}$ $T_{\text{Run}}$} & \textbf{Computational Cost$^{\mathrm{4}}$ $C_{\text{Comp}}$} & \textbf{Overhead $c$} & \textbf{Depth$^{\mathrm{5}}$} & \textbf{Gate$^{\mathrm{6}}$} \\
\hline
\textbf{Noise-free VQE estimation} & -- & $-13.20 \pm 0.05$ & -- & -- & -- & -- & -- & -- & -- \\
\hline
\multirow{2}{64pt}{\textbf{Unmitigated VQE estimation}} & Dephasing & $-10.07 \pm 0.12$ & -- & -- & -- & -- & -- & -- & -- \\
 & AD & $-11.59 \pm 0.13$ & -- & -- & -- & -- & -- & -- & -- \\
\hline
\multirow{2}{*}{\textbf{ZNE of order 1}} & Dephasing & $-12.29 \pm 0.09$ & $9.085 \times 10^{-1}$ & $93.12$ & $1.786 \times 10^{11}$ & $3.918 \times 10^{15}$ & \multirow{2}{*}{$1.750 \times 10^{1}$} & \multirow{2}{*}{70} & \multirow{2}{*}{174} \\
 & AD & $-13.01 \pm 0.08$ & $1.904 \times 10^{-1}$ & $98.56$ & $1.591 \times 10^{11}$ & $3.918 \times 10^{15}$ &  &  &  \\
\hline
\multirow{2}{*}{\textbf{ZNE of order 2}} & Dephasing & $-12.86 \pm 0.06$ & $3.422 \times 10^{-1}$ & $97.41$ & $1.136 \times 10^{13}$ & $4.193 \times 10^{15}$ & \multirow{2}{*}{$3.514 \times 10^{2}$} & \multirow{2}{*}{236} & \multirow{2}{*}{616} \\
 & AD & $-13.13 \pm 0.05$ & $6.407 \times 10^{-2}$ & $99.51$ & $1.012 \times 10^{13}$ & $4.160 \times 10^{15}$ &  &  &  \\
\hline
\multirow{2}{*}{\textbf{ZNE of order 3}} & Dephasing & $-12.96 \pm 0.05$ & $2.360 \times 10^{-1}$ & $98.21$ & $3.255 \times 10^{14}$ & $1.262 \times 10^{16}$ & \multirow{2}{*}{$4.349 \times 10^{3}$} & \multirow{2}{*}{574} & \multirow{2}{*}{1554} \\
 & AD & $-13.08 \pm 0.05$ & $1.196 \times 10^{-1}$ & $99.09$ & $2.900 \times 10^{14}$ & $1.158 \times 10^{16}$ &  &  &  \\
\hline
\multirow{2}{*}{\textbf{ZNE of order 4}} & Dephasing & $-12.93 \pm 0.05$ & $2.717 \times 10^{-1}$ & $97.94$ & $5.697 \times 10^{15}$ & $1.685 \times 10^{17}$ & \multirow{2}{*}{$3.871 \times 10^{4}$} & \multirow{2}{*}{1182} & \multirow{2}{*}{3302} \\
 & AD & $-13.01 \pm 0.05$ & $1.922 \times 10^{-1}$ & $98.54$ & $5.075 \times 10^{15}$ & $1.489 \times 10^{17}$ &  &  &  \\
\hline
\multirow{2}{*}{\textbf{ZNE of order 5}} & Dephasing & $-12.85 \pm 0.06$ & $3.515 \times 10^{-1}$ & $97.34$ & $8.212 \times 10^{16}$ & $2.556 \times 10^{18}$ & \multirow{2}{*}{$3.155 \times 10^{5}$} & \multirow{2}{*}{2188} & \multirow{2}{*}{6282} \\
 & AD & $-12.94 \pm 0.05$ & $2.583 \times 10^{-1}$ & $98.04$ & $7.316 \times 10^{16}$ & $2.252 \times 10^{18}$ &  &  &  \\
\hline

\multicolumn{10}{p{460pt}}{\footnotesize AD: Amplitude damping. ``Dephasing'' and ``AD'' denote local channels.} \\
\multicolumn{10}{p{500pt}}{\footnotesize All values are expressed in dimensionless units. Due to the large number of data points, they are not included explicitly here. All noise models for multivariate ZNE use identical folded circuits.}\\
\multicolumn{10}{p{500pt}}{\footnotesize $^{\mathrm{1}}$ The error bars represent the standard deviation computed from 10 independent trials, where each trial consists of VQE optimization with a different random initialization of the variational parameters, followed by ZNE using the optimized parameters.} \\
\multicolumn{10}{p{500pt}}{\footnotesize $^{\mathrm{2}}$ Relative accuracy in \% with respect to the noise-free VQE estimation.} \\
\multicolumn{10}{p{500pt}}{\footnotesize $^{\mathrm{3}}$ In Equation~\eqref{eq-runtime}, the mean values of $T^{\text{opt}}_{\text{max}}$ over 10 independent simulation runs for each noise model, with different initializations of the variational parameters, are $19.82 \pm 0.99$ for dephasing noise and $20.04 \pm 0.92$ for amplitude damping noise.} \\
\multicolumn{10}{p{500pt}}{\footnotesize $^{\mathrm{4}}$ $\operatorname{Var}\!\left(\langle H_{\text{Heisenberg}} \rangle\right)$ is $21.457$ and $18.903$ under dephasing and amplitude damping noise, respectively. For $\epsilon = 0.001$, the corresponding unmitigated shot counts are $s_u = 2.145 \times 10^{7}$ and $1.890 \times 10^{7}$, respectively. In Equation~\eqref{eq-comp-cost}, the mean values of $N_{\mathrm{iter}}$ over 10 independent simulation runs for each noise model, with different initial variational parameters, are $124.2 \pm 24.87$ for dephasing noise and $141.0 \pm 28.44$ for amplitude damping noise. Although the actual iteration count may differ on real hardware, the simulation-derived $N_{\mathrm{iter}}$ serves as a representative proxy for estimating the total computational cost.} \\
\multicolumn{10}{p{460pt}}{\footnotesize $^{\mathrm{5}}$ Depth: Circuit-layer depth in Equation \eqref{eq-circuit-layer-depth}.} \\
\multicolumn{10}{p{460pt}}{\footnotesize $^{\mathrm{6}}$ Gate: Total gate operations per ansatz layer summed over all $M$ circuits, $M$ is defined in Equation \eqref{eq-lre-m}.} \\
\hline
\end{tabular}
\end{table*}

\section{\break \hl{Univariate and Multivariate ZNE for a 9-Qubit System}} \label{Univariate and Multivariate ZNE for a 9-Qubit System}

\hl{In this section, we present additional ZNE results for a 9-qubit, 30-layer (i.e. 150 variational parameters) indirect-control VQE system. Both standard univariate and multivariate ZNE are performed. For univariate ZNE of different extrapolation orders, all gates are affected by local depolarizing noise. For multivariate ZNE, the vector of noise parameters is defined as in Equation~\eqref{eq-vector-noise}. In the univariate case, the effective noise level $\Lambda^{(i)}$ is expressed as a weighted sum of the numbers of noisy gates in an $N$-qubit circuit:}
\begin{equation}
    \Lambda^{(i)} := \mathcal{N}_{R}^{(i)}  + 2\mathcal{N}_{CZ}^{(i)} + 9\mathcal{N}_{T}^{(i)} + \mathcal{N}_{Y}^{(i)},
    \label{eq-noise-all-gates-noisy}
\end{equation}
\hl{where  $\mathcal{N}_{R}^{(i)}$, $\mathcal{N}_{CZ}^{(i)}$, $\mathcal{N}_{T}^{(i)}$, and $\mathcal{N}_{Y}^{(i)}$ denote the number of single-qubit rotation gates, two-qubit $CZ$ gates, time-evolution gates acting on $9$-qubit, and $Y$ gates in a single layer, respectively, when executing $i$-th circuit. The index $i = 1,2,3,\dots$ labels different circuits in the noise-scaling procedure. The case $i=1$ corresponds to the unfolded circuit with base noise, while $i>1$ denotes folded circuits constructed via noise amplification, each associated with a larger effective noise level. The target observable is 9-qubit transverse-field Ising Hamiltonian described by Equation \eqref{ising-hamiltonian} with $N=9$. The weighting factors $2$ and $9$ reflect the size of the corresponding gates. 
In our simulations, single-qubit gates such as $R$ and $Y$ are assigned a depolarizing noise probability of $0.001$. We also set $\gamma = 0.001$ in Equation \eqref{eq-time-evo-noise}, thereby accounting for the noise associated with the time-evolution gate. In contrast, the local depolarizing noise acting on the two-qubit $CZ$ gate is assigned a higher depolarizing noise probability of $0.005$. In the VQE algorithm, the time parameters $\{t_n\}_{n=2}^{30}$ are first sampled uniformly from $[0,20]$ and then optimized. We note that, under the $M_n=1$ Trotter-step approximation, the condition $t_{n+1}-t_n \ll [c_\text{noise}(N-1)\gamma]^{-\frac{1}{2}} \le \sqrt{31.25} \approx 5.5$ is slightly more stringent than that for the 7- or 8-qubit systems. Nonetheless, the results shown in Figure~\ref{fig-plots-zne-diff-orders-9q-uni-vs-multivariate} demonstrate that multivariate ZNE can provide better error mitigation than univariate ZNE when noise affects different types of gates, at the cost of executing a larger number of deeper folded circuits. 

Scaling to 9 qubits demands substantially greater computational resources, primarily because the SLSQP optimizer requires repeated cost-function and gradient evaluations, each involving density-matrix evolution in a Hilbert space of dimension $2^N$, where $N$ is the number of qubits. Using Qulacs for density-matrix evolution on a consumer-grade workstation equipped with 32 logical CPU cores and 137.2~GB of RAM, three concurrent simulation runs of the 9-qubit system required 9~days and 6.5~hours of continuous execution ($\approx 654$ total CPU-hours), corresponding to an average of 9.08~days per run. Although software-level optimizations could improve computational efficiency, these results highlight the practical limitations of our simulation on consumer-grade hardware. Further increasing the system size would impose substantially greater computational demands and may require access to high-performance computing (HPC) resources and fault-tolerant checkpoint for practical execution.}

\Figure[t!](topskip=0pt, botskip=0pt, midskip=0pt){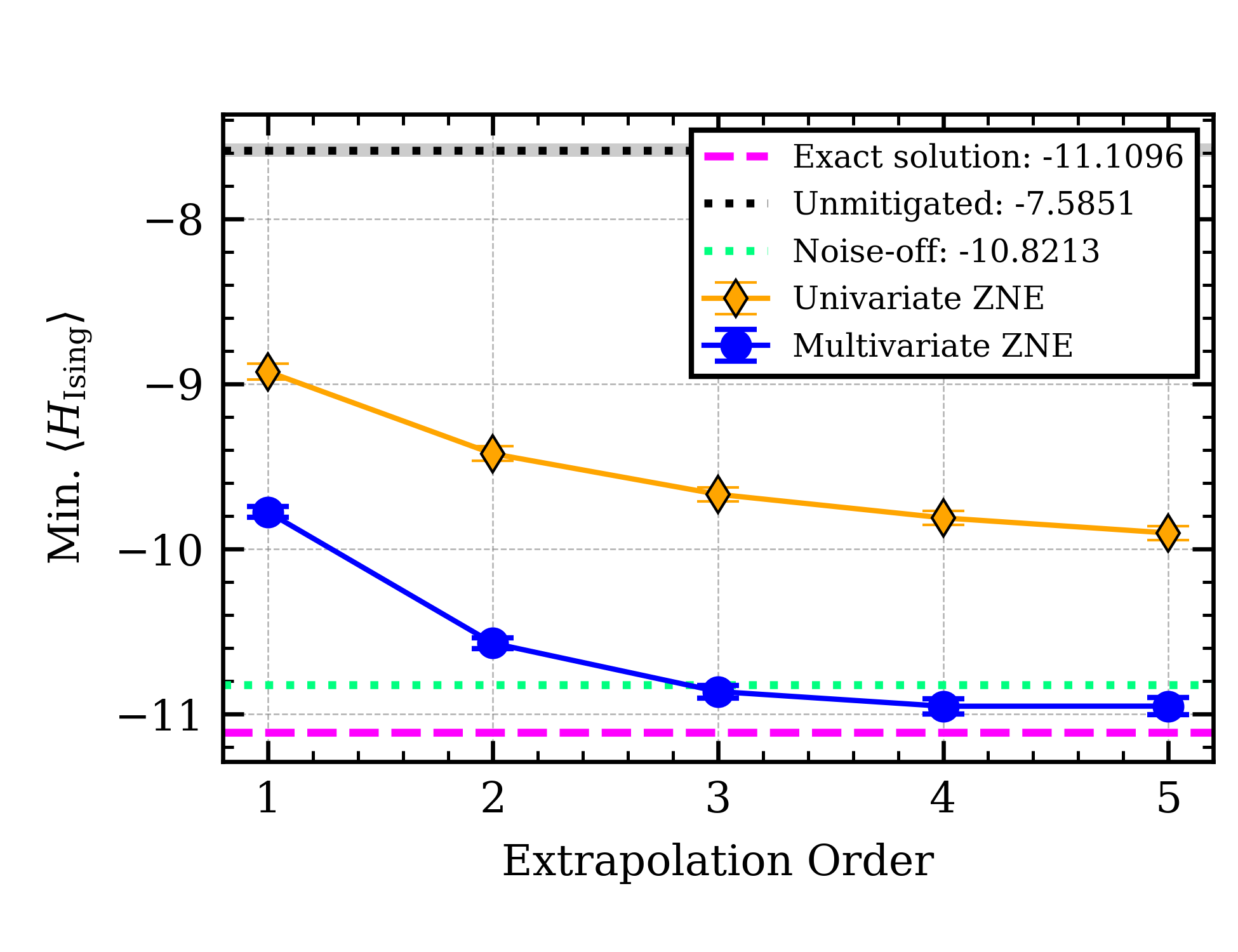}
{\hl{Univariate and multivariate Richardson extrapolation with 9-qubit, 30-layer indirect-control VQE ansatz is used. The minimum expectation value of the 9-qubit $H_{\text{Ising}}$ is estimated. Local depolarizing noise is applied to the entire system with $\gamma = 0.001$ for the time-evolution gate, while the single-qubit gates ($R$, $Y$) are subjected to depolarizing noise with probability $0.001$, and the $CZ$ gate is subjected to local depolarizing noise with probability $0.005$. Initial random time parameters are sampled from the range $[0, T_{\text{max}}=20]$. The time-evolution gates are simulated using $M_n=1$ in Equation \eqref{eq trotterized time-evolution}. The error bars represent the standard deviation of the unmitigated and corresponding ZNE estimates, computed from 3 independent trials with different initial variational parameters. The number of trials is limited by the substantial computational time required to complete each VQE simulation on our workstation. The resulting standard deviations therefore reflect the uncertainty arising from random initialization. The noise-free (``noise-off'') value corresponds to a single VQE run and is shown without a mean or standard deviation, serving only as a reference for the ZNE estimates.}\label{fig-plots-zne-diff-orders-9q-uni-vs-multivariate}}

\section{\break Explicit form of the Trotter splitting error} \label{app:trotter_error_explicit}

We give  a first-order product-formula error bound in the
diamond norm, which is a direct consequence of the
variation-of-parameters approach to product-formula errors
developed in~\cite{Childs2021}, combined with the contractivity of
completely positive trace-preserving maps in the diamond norm.
A related commutator-based analysis specifically for Lindbladian
simulation was recently developed in~\cite{Wang2026}.
For the reader’s convenience, we provide the complete proof of the lemma here.
\begin{lemma}[First-order Trotter error in the diamond norm]
\label{lem:first_order_lindblad_trotter}
Let $\mathcal{L}_1$ and $\mathcal{L}_2$ be Lindblad generators on a
finite-dimensional quantum system. Then, for any $t \geq 0$,
\begin{equation}
\left\|
e^{t\mathcal{L}_1}e^{t\mathcal{L}_2}
-
e^{t(\mathcal{L}_1+\mathcal{L}_2)}
\right\|_{\diamond}
\leq
\frac{t^2}{2}
\left\|
[\mathcal{L}_1,\mathcal{L}_2]
\right\|_{\diamond}.
\label{eq:first_order_lindblad_bound}
\end{equation}
\end{lemma}

\begin{proof}[Proof]
We first define $F(u)$ for $0\leq u\leq t$ by
\begin{equation}
F(u)
=
e^{(t-u)(\mathcal{L}_1+\mathcal{L}_2)}
e^{u\mathcal{L}_1}
e^{u\mathcal{L}_2}.
\label{eq:F_definition}
\end{equation}
Then, it satisfies $F(0)=e^{t(\mathcal{L}_1+\mathcal{L}_2)},
$ and $F(t)=e^{t\mathcal{L}_1}e^{t\mathcal{L}_2}$.
Differentiating~\eqref{eq:F_definition}, we derive
\begin{align}
F'(u)
={}&
e^{(t-u)(\mathcal{L}_1+\mathcal{L}_2)}
\left[
e^{u\mathcal{L}_1},\mathcal{L}_2
\right]
e^{u\mathcal{L}_2}.
\label{eq:Fprime}
\end{align}
We further define $G(v)$ by $G(v)
=e^{v\mathcal{L}_1}\mathcal{L}_2e^{(u-v)\mathcal{L}_1}$ for $0\leq v\leq u$.
Then, $G'(v)$ can be written as 
\begin{align}
G'(v)
&=
e^{v\mathcal{L}_1}
[\mathcal{L}_1,\mathcal{L}_2]
e^{(u-v)\mathcal{L}_1}.
\label{eq:Gprime}
\end{align}
Since $G(u)=e^{u\mathcal{L}_1}\mathcal{L}_2
$ and $G(0)=\mathcal{L}_2e^{u\mathcal{L}_1}$,
the fundamental theorem of calculus leads to
\begin{equation}
\left[
e^{u\mathcal{L}_1},\mathcal{L}_2
\right]
=
\int_0^u
e^{v\mathcal{L}_1}
[\mathcal{L}_1,\mathcal{L}_2]
e^{(u-v)\mathcal{L}_1}
\,dv.
\label{eq:commutator_integral}
\end{equation}
Substituting~\eqref{eq:commutator_integral} into~\eqref{eq:Fprime}, we
obtain
\begin{align}
F'(u)
=
\int_0^u
&
e^{(t-u)(\mathcal{L}_1+\mathcal{L}_2)}
e^{v\mathcal{L}_1}
[\mathcal{L}_1,\mathcal{L}_2]
\nonumber\\
&\times
e^{(u-v)\mathcal{L}_1}
e^{u\mathcal{L}_2}
\,dv.
\label{eq:Fprime_integral}
\end{align}
Integrating from $u=0$ to $u=t$ gives the error representation
\begin{align}
&
e^{t\mathcal{L}_1}e^{t\mathcal{L}_2}
-
e^{t(\mathcal{L}_1+\mathcal{L}_2)}
\nonumber\\
&\quad=
\int_0^t du
\int_0^u dv\,
e^{(t-u)(\mathcal{L}_1+\mathcal{L}_2)}
e^{v\mathcal{L}_1}
[\mathcal{L}_1,\mathcal{L}_2]
e^{(u-v)\mathcal{L}_1}
e^{u\mathcal{L}_2}.
\label{eq:first_order_error_representation}
\end{align}

Because $\mathcal{L}_1$, $\mathcal{L}_2$, and
$\mathcal{L}_1+\mathcal{L}_2$ are Lindblad generators, their
forward-time evolutions are completely positive and trace preserving.
Consequently, for all nonnegative times $s$,
\begin{equation}
\left\|e^{s\mathcal{L}_j}\right\|_{\diamond}
=
\left\|e^{s(\mathcal{L}_1+\mathcal{L}_2)}\right\|_{\diamond}
=1.
\end{equation}
Using the triangle inequality and the submultiplicativity of the
diamond norm in~\eqref{eq:first_order_error_representation}, we obtain
\begin{align}
&
\left\|
e^{t\mathcal{L}_1}e^{t\mathcal{L}_2}
-
e^{t(\mathcal{L}_1+\mathcal{L}_2)}
\right\|_{\diamond}
\nonumber\\
&\quad\leq
\int_0^tdu\int_0^u dv\,
\left\|
[\mathcal{L}_1,\mathcal{L}_2]
\right\|_{\diamond}
\nonumber\\
&\quad=
\frac{t^2}{2}
\left\|
[\mathcal{L}_1,\mathcal{L}_2]
\right\|_{\diamond},
\end{align}
which proves~\eqref{eq:first_order_lindblad_bound}.
\end{proof}
The formula for the $r$-step version of the Lindbladian is obtained by adapting the exact same argument used by Childs {\em et al.} in~\cite{Childs2021} for the formula regarding Hermitian operators to Lemma~\ref{lem:first_order_lindblad_trotter}.
Thus, when $\mathcal{L}_1$ and $\mathcal{L}_2$ are Lindblad generators,
for any $t \geq 0$ and any positive integer $r$,
\begin{equation}
\left\|
\left(
e^{\frac{t}{r}\mathcal{L}_1}
e^{\frac{t}{r}\mathcal{L}_2}
\right)^r
-
e^{t(\mathcal{L}_1+\mathcal{L}_2)}
\right\|_{\diamond}
\leq
\frac{t^2}{2r}
\left\|
[\mathcal{L}_1,\mathcal{L}_2]
\right\|_{\diamond}.
\label{eq:global_first_order_lindblad_bound}
\end{equation}

Now, we assume that $\mathcal{L}_1$ corresponds to the time evolution under a 1-D spin Hamiltonian $\mathcal{L}_\text{spin}=\sum_{k=1}^{N-1}\mathcal{L}_{\text{spin},k}$.
Here, $\mathcal{L}_{\text{spin},k}$ is defined by 
$\mathcal{L}_{\text{spin},k}(\rho)=-i[H_k,\rho]$, where $H_k$ is a Hamiltonian 
describing the nearest-neighbor
coupling between $k$-th and $k+1$-th qubits.
For example, when the system is a 1-D XY model, $H_k$ is written down as
$H_k=\frac{1}{2}(X_kX_{k+1}+Z_kZ_{k+1})$. 
On the other hand, we assume that $\mathcal{L}_2$ is a Lindblad corresponding to a local noise
given by $\mathcal{L}_2=\mathcal{L}_\text{noise}=\gamma\sum_{j=1}^N \mathcal{L}_j$ with each
$\mathcal{L}_j$ acting only on qubit $j$. 
Applying Equation~\eqref{eq:global_first_order_lindblad_bound} to 
$\mathcal{L}_1=\mathcal{L}_\text{spin}$ and $\mathcal{L}_2=\mathcal{L}_\text{noise}$,
the error of $r$-step Trotter decomposition is evaluated as follows: 
\begin{equation}
\left\|
\left(
e^{\frac{t}{r}\mathcal{L}_\text{spin}}
e^{\frac{t}{r}\mathcal{L}_\text{noise}}
\right)^r
-
e^{t(\mathcal{L}_\text{spin}+\mathcal{L}_\text{noise})}
\right\|_{\diamond}
\leq
\frac{t^2}{2r}
\left\|
[\mathcal{L}_\text{spin},\mathcal{L}_\text{noise}]
\right\|_{\diamond}.
\label{eq-app-diamond}
\end{equation}
Because $\mathcal{L}_{\text{spin},k}$ acts only on qubits $\{k,k+1\}$ and $\mathcal{L}_j$
acts only on qubit $j$, these two superoperators commute exactly whenever
$j\notin\{k,k+1\}$. Hence, the double sum defining
$[\mathcal{L}_\text{spin},\mathcal{L}_\text{noise}]$ collapses to at most
$2(N-1)$ nonzero terms — one for each qubit adjacent to each of the $N-1$
bonds:
\begin{equation}
    [\mathcal{L}_\text{spin},\mathcal{L}_\text{noise}]
    = \gamma\sum_{k=1}^{N-1}\Big([\mathcal{L}_{\text{spin},k},\mathcal{L}_k]+[\mathcal{L}_{\text{spin},k},\mathcal{L}_{k+1}]\Big).
\end{equation}
Further, we assume that the couplings of the spin chain
and the decoherence rates of noise are uniform.
Then, we further derive
\begin{equation}
\begin{aligned}
    \|[\mathcal{L}_\text{spin},\mathcal{L}_\text{noise}]\|_\diamond
    &\le 2(N-1)\gamma\,\|[\mathcal{L}_{\text{spin},1},\mathcal{L}_1]\|_\diamond
\end{aligned}
\label{eq-app-cloc}
\end{equation}
where $\left\|[\mathcal{L}_{\text{spin},1},\mathcal{L}_1]\right\|_\diamond$ is a two-qubit quantity that does \emph{not} grow with
$N$. Combining Equations~\eqref{eq-app-diamond}--\eqref{eq-app-cloc}, 
the error of $r$-step Trotter decomposition takes the explicit form
\begin{equation}
\begin{aligned}
&\left\|
\left(
e^{\frac{t}{r}\mathcal{L}_\text{spin}}
e^{\frac{t}{r}\mathcal{L}_\text{noise}}
\right)^r
-
e^{t(\mathcal{L}_\text{spin}+\mathcal{L}_\text{noise})}
\right\|_{\diamond}\\
    &\qquad \le
    \frac{\gamma t^2(N-1)}{2r}
    \|[\mathcal{L}_{\text{spin},1},\mathcal{L}_1]\|_\diamond.
\end{aligned}
\label{eq-app-final}
\end{equation}
We note that the above inequality is valid as long as \(\mathcal{L}_j\) is an identical Lindbladian acting on the \(j\)-th qubit, i.e., the noise is local and identical for all qubits. 

For convenience, we define the two-qubit constant
$$
c_\text{noise}
:=
\left\|[\mathcal{L}_{\text{spin},1},\mathcal{L}_1]\right\|_\diamond .
$$
Then, Equation \eqref{eq-app-final} guarantees that $\frac{1}{2}\gamma t^2(N-1)c_\text{noise}$ is 
an upperbound of the one-step splitting error. Requiring this local splitting error to remain small gives the step-size condition
\begin{equation}
t\ll
\left[c_\text{noise}(N-1)\gamma\right]^{-1/2}.    
\end{equation}
The constant \(c_\text{noise}\) depends only on the local noise model and can be evaluated on a single two-qubit bond, independently of \(N\).

To evaluate \(c_\text{noise}\), we distinguish the local Lindbladian \(\mathcal{L}_j\) from its Lindblad jump operator. For a single jump operator \(L\), define the corresponding dissipator by
$$
\mathcal{D}_L(\rho)
=
L\rho L^\dagger
-\frac{1}{2}\{L^\dagger L,\rho\}.
$$
A straightforward calculation leads to
$$
[\mathcal{L}_{\text{spin},k},\mathcal{D}_L](\rho)
=
K\rho L^\dagger
+
L\rho K^\dagger
-\frac{1}{2}
\{K^\dagger L+L^\dagger K,\rho\},
$$
where $K$ is defined by \(K\equiv-i[H_k,L]\).
Let us define $\Gamma_{A,B}$ by $\Gamma_{A,B}(\rho)=A\rho B$.
Then,  $\Gamma_{A,B}$ satisfies 
$\|\Gamma_{A,B}\|_\diamond
=\|A\|_{\rm op}\|B\|_{\rm op}$.
By using this equality, 
$c_\text{noise}=\left\|[\mathcal{L}_{\text{spin},k},\mathcal{D}_L]\right\|_\diamond$
can be evaluated as
\begin{align}
&\left\|[\mathcal{L}_{\text{spin},k},\mathcal{D}_L]\right\|_\diamond \nonumber\\
=&\left\|\Gamma_{K,L^\dagger}+\Gamma_{L,K^\dagger}
-\frac{1}{2}(\Gamma_{K^\dagger L,I}+\Gamma_{L^\dagger K,I}
+\Gamma_{I,K^\dagger L}+\Gamma_{I,L^\dagger K})\right\|_\diamond 
\nonumber \\
\le&
4\left\|[H_k,L]\right\|_{\rm op}
\left\|L\right\|_{\rm op}.
\end{align}

For the XY model, since $H_k$ is given by
$H_k=\frac{1}{2}(X_kX_{k+1}+Z_kZ_{k+1})$,
the commutator between $H_k$ and a Pauli matrix can be easily derived as
$$
\|[H_k,X_k]\|_{\rm op}
=
\|[H_k,Z_k]\|_{\rm op}
=1,
\qquad
\|[H_k,Y_k]\|_{\rm op}=2.
$$

In the case of local depolarizing noise, 
$\mathcal{L}_i$ is given by 
$\mathcal{L}_i=\frac{1}{4}(\mathcal{D}_{X_i}+\mathcal{D}_{Y_i}+\mathcal{D}_{Z_i})$.
Consequently, $c_\text{noise}
$ can be evaluated as
$c_\text{noise}
\le
\frac{1}{4}\cdot 4(1+2+1)
=4$ for the XY model.

In the case of local dephasing noise, 
$\mathcal{L}_i$ is given by 
$\mathcal{L}_i=\frac{1}{2}\mathcal{D}_{Z_i}$.
Consequently, $c_\text{noise}
$ can be evaluated as
$
c_\text{noise}
\le
\frac{1}{2}\cdot 4
=2$ for the XY model.

For local amplitude damping with
$\mathcal{L}_i=\mathcal{D}_{\sigma_i^-}$,
where $\sigma_i^-=\frac{X_i-iY_i}{2}$,
we have
$$
\|[H_k,\sigma_k^-]\|_{\rm op}
\le
\frac{1}{2}
\left(
\|[H_k,X_k]\|_{\rm op}
+
\|[H_k,Y_k]\|_{\rm op}
\right),
$$ and hence $c_\text{noise}\le 6$ for the XY model.

\bibliographystyle{IEEEtran}

\bibliography{external/reference}

% Generated by IEEEtran.bst, version: 1.14 (2015/08/26)
\begin{thebibliography}{100}
\providecommand{\url}[1]{#1}
\csname url@samestyle\endcsname
\providecommand{\newblock}{\relax}
\providecommand{\bibinfo}[2]{#2}
\providecommand{\BIBentrySTDinterwordspacing}{\spaceskip=0pt\relax}
\providecommand{\BIBentryALTinterwordstretchfactor}{4}
\providecommand{\BIBentryALTinterwordspacing}{\spaceskip=\fontdimen2\font plus
\BIBentryALTinterwordstretchfactor\fontdimen3\font minus \fontdimen4\font\relax}
\providecommand{\BIBforeignlanguage}[2]{{%
\expandafter\ifx\csname l@#1\endcsname\relax
\typeout{** WARNING: IEEEtran.bst: No hyphenation pattern has been}%
\typeout{** loaded for the language `#1'. Using the pattern for}%
\typeout{** the default language instead.}%
\else
\language=\csname l@#1\endcsname
\fi
#2}}
\providecommand{\BIBdecl}{\relax}
\BIBdecl

\bibitem{Benioff1980}
\BIBentryALTinterwordspacing
P.~Benioff, ``The computer as a physical system: A microscopic quantum mechanical hamiltonian model of computers as represented by turing machines,'' \emph{Journal of Statistical Physics}, vol.~22, no.~5, pp. 563--591, May 1980. [Online]. Available: \url{https://doi.org/10.1007/BF01011339}
\BIBentrySTDinterwordspacing

\bibitem{Feynman1982}
\BIBentryALTinterwordspacing
R.~P. Feynman, ``Simulating physics with computers,'' \emph{International Journal of Theoretical Physics}, vol.~21, no.~6, pp. 467--488, Jun 1982. [Online]. Available: \url{https://doi.org/10.1007/BF02650179}
\BIBentrySTDinterwordspacing

\bibitem{Deutsch85}
\BIBentryALTinterwordspacing
D.~Deutsch, ``Quantum theory, the church-turing principle and the universal quantum computer,'' \emph{Proceedings of the Royal Society of London. A. Mathematical and Physical Sciences}, vol. 400, no. 1818, pp. 97--117, 1985. [Online]. Available: \url{https://royalsocietypublishing.org/doi/abs/10.1098/rspa.1985.0070}
\BIBentrySTDinterwordspacing

\bibitem{Shor_1997}
\BIBentryALTinterwordspacing
P.~W. Shor, ``Polynomial-time algorithms for prime factorization and discrete logarithms on a quantum computer,'' \emph{SIAM Journal on Computing}, vol.~26, no.~5, pp. 1484--1509, oct 1997. [Online]. Available: \url{http://dx.doi.org/10.1137/S0097539795293172}
\BIBentrySTDinterwordspacing

\bibitem{grover1996fastquantummechanicalalgorithm}
\BIBentryALTinterwordspacing
L.~K. Grover, ``A fast quantum mechanical algorithm for database search,'' 1996. [Online]. Available: \url{https://arxiv.org/abs/quant-ph/9605043}
\BIBentrySTDinterwordspacing

\bibitem{Landauer95}
\BIBentryALTinterwordspacing
R.~Landauer, ``Is quantum mechanics useful?'' \emph{Philosophical Transactions of the Royal Society of London. Series A: Physical and Engineering Sciences}, vol. 353, no. 1703, pp. 367--376, 1995. [Online]. Available: \url{https://royalsocietypublishing.org/doi/abs/10.1098/rsta.1995.0106}
\BIBentrySTDinterwordspacing

\bibitem{Shor95}
\BIBentryALTinterwordspacing
P.~W. Shor, ``Scheme for reducing decoherence in quantum computer memory,'' \emph{Phys. Rev. A}, vol.~52, pp. R2493--R2496, Oct 1995. [Online]. Available: \url{https://link.aps.org/doi/10.1103/PhysRevA.52.R2493}
\BIBentrySTDinterwordspacing

\bibitem{Steane96}
\BIBentryALTinterwordspacing
A.~M. Steane, ``Error correcting codes in quantum theory,'' \emph{Phys. Rev. Lett.}, vol.~77, pp. 793--797, Jul 1996. [Online]. Available: \url{https://link.aps.org/doi/10.1103/PhysRevLett.77.793}
\BIBentrySTDinterwordspacing

\bibitem{Calderbank96}
\BIBentryALTinterwordspacing
A.~R. Calderbank and P.~W. Shor, ``Good quantum error-correcting codes exist,'' \emph{Phys. Rev. A}, vol.~54, pp. 1098--1105, Aug 1996. [Online]. Available: \url{https://link.aps.org/doi/10.1103/PhysRevA.54.1098}
\BIBentrySTDinterwordspacing

\bibitem{Steane96b}
\BIBentryALTinterwordspacing
A.~Steane, ``Multiple-particle interference and quantum error correction,'' \emph{Proceedings of the Royal Society of London. Series A: Mathematical, Physical and Engineering Sciences}, vol. 452, no. 1954, pp. 2551--2577, 1996. [Online]. Available: \url{https://royalsocietypublishing.org/doi/abs/10.1098/rspa.1996.0136}
\BIBentrySTDinterwordspacing

\bibitem{Gottesman96}
\BIBentryALTinterwordspacing
D.~Gottesman, ``Class of quantum error-correcting codes saturating the quantum hamming bound,'' \emph{Phys. Rev. A}, vol.~54, pp. 1862--1868, Sep 1996. [Online]. Available: \url{https://link.aps.org/doi/10.1103/PhysRevA.54.1862}
\BIBentrySTDinterwordspacing

\bibitem{Kitaev97}
\BIBentryALTinterwordspacing
A.~Y. Kitaev, ``Quantum computations: algorithms and error correction,'' \emph{Russian Mathematical Surveys}, vol.~52, no.~6, p. 1191, dec 1997. [Online]. Available: \url{https://dx.doi.org/10.1070/RM1997v052n06ABEH002155}
\BIBentrySTDinterwordspacing

\bibitem{Shor96b}
P.~Shor, ``Fault-tolerant quantum computation,'' in \emph{Proceedings of 37th Conference on Foundations of Computer Science}, 1996, pp. 56--65.

\bibitem{Gottesman98}
\BIBentryALTinterwordspacing
D.~Gottesman, ``Theory of fault-tolerant quantum computation,'' \emph{Phys. Rev. A}, vol.~57, pp. 127--137, Jan 1998. [Online]. Available: \url{https://link.aps.org/doi/10.1103/PhysRevA.57.127}
\BIBentrySTDinterwordspacing

\bibitem{Preskill98}
\BIBentryALTinterwordspacing
J.~Preskill, \emph{Fault-Tolerant Quantum Computation}.\hskip 1em plus 0.5em minus 0.4em\relax World Scientific, 1998, pp. 213--269. [Online]. Available: \url{https://www.worldscientific.com/doi/abs/10.1142/9789812385253_0008}
\BIBentrySTDinterwordspacing

\bibitem{Steane1999}
\BIBentryALTinterwordspacing
A.~M. Steane, ``Efficient fault-tolerant quantum computing,'' \emph{Nature}, vol. 399, no. 6732, pp. 124--126, May 1999. [Online]. Available: \url{https://doi.org/10.1038/20127}
\BIBentrySTDinterwordspacing

\bibitem{Aharanov08}
\BIBentryALTinterwordspacing
D.~Aharonov and M.~Ben-Or, ``Fault-tolerant quantum computation with constant error rate,'' \emph{SIAM Journal on Computing}, vol.~38, no.~4, pp. 1207--1282, 2008. [Online]. Available: \url{https://doi.org/10.1137/S0097539799359385}
\BIBentrySTDinterwordspacing

\bibitem{Nielsen_Chuang_2010}
M.~A. Nielsen and I.~L. Chuang, \emph{Quantum Computation and Quantum Information}.\hskip 1em plus 0.5em minus 0.4em\relax Cambridge: Cambridge University Press, 2000.

\bibitem{Brun_2020}
\BIBentryALTinterwordspacing
D.~Lidar and T.~A. Brun, ``Quantum error correction,'' feb 2013. [Online]. Available: \url{http://dx.doi.org/10.1093/acrefore/9780190871994.013.35}
\BIBentrySTDinterwordspacing

\bibitem{Blais04}
\BIBentryALTinterwordspacing
A.~Blais, R.-S. Huang, A.~Wallraff, S.~M. Girvin, and R.~J. Schoelkopf, ``Cavity quantum electrodynamics for superconducting electrical circuits: An architecture for quantum computation,'' \emph{Phys. Rev. A}, vol.~69, p. 062320, Jun 2004. [Online]. Available: \url{https://link.aps.org/doi/10.1103/PhysRevA.69.062320}
\BIBentrySTDinterwordspacing

\bibitem{Schuster2007}
\BIBentryALTinterwordspacing
D.~I. Schuster, A.~A. Houck, J.~A. Schreier, A.~Wallraff, J.~M. Gambetta, A.~Blais, L.~Frunzio, J.~Majer, B.~Johnson, M.~H. Devoret, S.~M. Girvin, and R.~J. Schoelkopf, ``Resolving photon number states in a superconducting circuit,'' \emph{Nature}, vol. 445, no. 7127, pp. 515--518, Feb 2007. [Online]. Available: \url{https://doi.org/10.1038/nature05461}
\BIBentrySTDinterwordspacing

\bibitem{Blais07}
\BIBentryALTinterwordspacing
A.~Blais, J.~Gambetta, A.~Wallraff, D.~I. Schuster, S.~M. Girvin, M.~H. Devoret, and R.~J. Schoelkopf, ``Quantum-information processing with circuit quantum electrodynamics,'' \emph{Phys. Rev. A}, vol.~75, p. 032329, Mar 2007. [Online]. Available: \url{https://link.aps.org/doi/10.1103/PhysRevA.75.032329}
\BIBentrySTDinterwordspacing

\bibitem{Koch07}
\BIBentryALTinterwordspacing
J.~Koch, T.~M. Yu, J.~Gambetta, A.~A. Houck, D.~I. Schuster, J.~Majer, A.~Blais, M.~H. Devoret, S.~M. Girvin, and R.~J. Schoelkopf, ``Charge-insensitive qubit design derived from the cooper pair box,'' \emph{Phys. Rev. A}, vol.~76, p. 042319, Oct 2007. [Online]. Available: \url{https://link.aps.org/doi/10.1103/PhysRevA.76.042319}
\BIBentrySTDinterwordspacing

\bibitem{Schreier08}
\BIBentryALTinterwordspacing
J.~A. Schreier, A.~A. Houck, J.~Koch, D.~I. Schuster, B.~R. Johnson, J.~M. Chow, J.~M. Gambetta, J.~Majer, L.~Frunzio, M.~H. Devoret, S.~M. Girvin, and R.~J. Schoelkopf, ``Suppressing charge noise decoherence in superconducting charge qubits,'' \emph{Phys. Rev. B}, vol.~77, p. 180502, May 2008. [Online]. Available: \url{https://link.aps.org/doi/10.1103/PhysRevB.77.180502}
\BIBentrySTDinterwordspacing

\bibitem{Lucero2012}
\BIBentryALTinterwordspacing
E.~Lucero, R.~Barends, Y.~Chen, J.~Kelly, M.~Mariantoni, A.~Megrant, P.~O'Malley, D.~Sank, A.~Vainsencher, J.~Wenner, T.~White, Y.~Yin, A.~N. Cleland, and J.~M. Martinis, ``Computing prime factors with a josephson phase qubit quantum processor,'' \emph{Nature Physics}, vol.~8, no.~10, pp. 719--723, Oct 2012. [Online]. Available: \url{https://doi.org/10.1038/nphys2385}
\BIBentrySTDinterwordspacing

\bibitem{Barends2014}
\BIBentryALTinterwordspacing
R.~Barends, J.~Kelly, A.~Megrant, A.~Veitia, D.~Sank, E.~Jeffrey, T.~C. White, J.~Mutus, A.~G. Fowler, B.~Campbell, Y.~Chen, Z.~Chen, B.~Chiaro, A.~Dunsworth, C.~Neill, P.~O'Malley, P.~Roushan, A.~Vainsencher, J.~Wenner, A.~N. Korotkov, A.~N. Cleland, and J.~M. Martinis, ``Superconducting quantum circuits at the surface code threshold for fault tolerance,'' \emph{Nature}, vol. 508, no. 7497, pp. 500--503, Apr 2014. [Online]. Available: \url{https://doi.org/10.1038/nature13171}
\BIBentrySTDinterwordspacing

\bibitem{Kim2023}
\BIBentryALTinterwordspacing
Y.~Kim, A.~Eddins, S.~Anand, K.~X. Wei, E.~van~den Berg, S.~Rosenblatt, H.~Nayfeh, Y.~Wu, M.~Zaletel, K.~Temme, and A.~Kandala, ``Evidence for the utility of quantum computing before fault tolerance,'' \emph{Nature}, vol. 618, no. 7965, pp. 500--505, Jun 2023. [Online]. Available: \url{https://doi.org/10.1038/s41586-023-06096-3}
\BIBentrySTDinterwordspacing

\bibitem{CarreraVazquez2024}
\BIBentryALTinterwordspacing
A.~Carrera~Vazquez, C.~Tornow, D.~Rist{\`{e}}, S.~Woerner, M.~Takita, and D.~J. Egger, ``Combining quantum processors with real-time classical communication,'' \emph{Nature}, vol. 636, no. 8041, pp. 75--79, Dec 2024. [Online]. Available: \url{https://doi.org/10.1038/s41586-024-08178-2}
\BIBentrySTDinterwordspacing

\bibitem{Acharya2024}
\BIBentryALTinterwordspacing
R.~Acharya \emph{et~al.}, ``Quantum error correction below the surface code threshold,'' \emph{Nature}, Dec 2024. [Online]. Available: \url{https://doi.org/10.1038/s41586-024-08449-y}
\BIBentrySTDinterwordspacing

\bibitem{quantumsup}
\BIBentryALTinterwordspacing
F.~Arute \emph{et~al.}, ``Quantum supremacy using a programmable superconducting processor,'' \emph{Nature}, vol. 574, pp. 505--510, 2019. [Online]. Available: \url{https://www.nature.com/articles/s41586-019-1666-5}
\BIBentrySTDinterwordspacing

\bibitem{PhysRevLett.127.180501}
\BIBentryALTinterwordspacing
Y.~Wu, W.-S. Bao, S.~Cao, F.~Chen, M.-C. Chen, X.~Chen, T.-H. Chung, H.~Deng, Y.~Du, D.~Fan, M.~Gong, C.~Guo, C.~Guo, S.~Guo, L.~Han, L.~Hong, H.-L. Huang, Y.-H. Huo, L.~Li, N.~Li, S.~Li, Y.~Li, F.~Liang, C.~Lin, J.~Lin, H.~Qian, D.~Qiao, H.~Rong, H.~Su, L.~Sun, L.~Wang, S.~Wang, D.~Wu, Y.~Xu, K.~Yan, W.~Yang, Y.~Yang, Y.~Ye, J.~Yin, C.~Ying, J.~Yu, C.~Zha, C.~Zhang, H.~Zhang, K.~Zhang, Y.~Zhang, H.~Zhao, Y.~Zhao, L.~Zhou, Q.~Zhu, C.-Y. Lu, C.-Z. Peng, X.~Zhu, and J.-W. Pan, ``Strong quantum computational advantage using a superconducting quantum processor,'' \emph{Phys. Rev. Lett.}, vol. 127, p. 180501, Oct 2021. [Online]. Available: \url{https://link.aps.org/doi/10.1103/PhysRevLett.127.180501}
\BIBentrySTDinterwordspacing

\bibitem{Preskill2018}
\BIBentryALTinterwordspacing
J.~Preskill, ``Quantum {C}omputing in the {NISQ} era and beyond,'' \emph{{Quantum}}, vol.~2, p.~79, aug 2018. [Online]. Available: \url{https://doi.org/10.22331/q-2018-08-06-79}
\BIBentrySTDinterwordspacing

\bibitem{nisq}
M.~AbuGhanem and H.~Eleuch, ``Nisq computers: A path to quantum supremacy,'' \emph{IEEE Access}, vol.~12, pp. 102\,941--102\,961, 2024.

\bibitem{Endo_2021}
\BIBentryALTinterwordspacing
S.~Endo, Z.~Cai, S.~C. Benjamin, and X.~Yuan, ``Hybrid quantum-classical algorithms and quantum error mitigation,'' \emph{Journal of the Physical Society of Japan}, vol.~90, no.~3, p. 032001, mar 2021. [Online]. Available: \url{http://dx.doi.org/10.7566/JPSJ.90.032001}
\BIBentrySTDinterwordspacing

\bibitem{vqa}
\BIBentryALTinterwordspacing
M.~Cerezo, A.~Arrasmith, R.~Babbush, S.~C. Benjamin, S.~Endo, K.~Fujii, J.~R. McClean, K.~Mitarai, X.~Yuan, L.~Cincio, and P.~J. Coles, ``Variational quantum algorithms,'' \emph{Nature Reviews Physics}, vol.~3, no.~9, pp. 625--644, aug 2021. [Online]. Available: \url{http://dx.doi.org/10.1038/s42254-021-00348-9}
\BIBentrySTDinterwordspacing

\bibitem{hegde2024buzzstrategicpathsenabling}
\BIBentryALTinterwordspacing
P.~R. Hegde, O.~Kyriienko, H.~Heimonen, P.~Tolias, G.~Netzer, P.~Barkoutsos, R.~Vinuesa, I.~Peng, and S.~Markidis, ``Beyond the buzz: Strategic paths for enabling useful nisq applications,'' 2024. [Online]. Available: \url{https://arxiv.org/abs/2405.14561}
\BIBentrySTDinterwordspacing

\bibitem{Peruzzo_2014}
\BIBentryALTinterwordspacing
A.~Peruzzo, J.~McClean, P.~Shadbolt, M.-H. Yung, X.-Q. Zhou, P.~J. Love, A.~Aspuru-Guzik, and J.~L. O'Brien, ``A variational eigenvalue solver on a photonic quantum processor,'' \emph{Nature Communications}, vol.~5, no.~1, jul 2014. [Online]. Available: \url{http://dx.doi.org/10.1038/ncomms5213}
\BIBentrySTDinterwordspacing

\bibitem{Tilly_2022}
\BIBentryALTinterwordspacing
J.~Tilly, H.~Chen, S.~Cao, D.~Picozzi, K.~Setia, Y.~Li, E.~Grant, L.~Wossnig, I.~Rungger, G.~H. Booth, and J.~Tennyson, ``The variational quantum eigensolver: A review of methods and best practices,'' \emph{Physics Reports}, vol. 986, pp. 1--128, nov 2022. [Online]. Available: \url{http://dx.doi.org/10.1016/j.physrep.2022.08.003}
\BIBentrySTDinterwordspacing

\bibitem{PhysRevX.7.021050}
\BIBentryALTinterwordspacing
Y.~Li and S.~C. Benjamin, ``Efficient variational quantum simulator incorporating active error minimization,'' \emph{Phys. Rev. X}, vol.~7, p. 021050, Jun 2017. [Online]. Available: \url{https://link.aps.org/doi/10.1103/PhysRevX.7.021050}
\BIBentrySTDinterwordspacing

\bibitem{PhysRevLett.119.180509}
\BIBentryALTinterwordspacing
K.~Temme, S.~Bravyi, and J.~M. Gambetta, ``Error mitigation for short-depth quantum circuits,'' \emph{Phys. Rev. Lett.}, vol. 119, p. 180509, Nov 2017. [Online]. Available: \url{https://link.aps.org/doi/10.1103/PhysRevLett.119.180509}
\BIBentrySTDinterwordspacing

\bibitem{Cai_2023}
\BIBentryALTinterwordspacing
Z.~Cai, R.~Babbush, S.~C. Benjamin, S.~Endo, W.~J. Huggins, Y.~Li, J.~R. McClean, and T.~E. O’Brien, ``Quantum error mitigation,'' \emph{Reviews of Modern Physics}, vol.~95, no.~4, Dec. 2023. [Online]. Available: \url{http://dx.doi.org/10.1103/RevModPhys.95.045005}
\BIBentrySTDinterwordspacing

\bibitem{Czarnik2021}
\BIBentryALTinterwordspacing
P.~Czarnik, A.~Arrasmith, P.~J. Coles, and L.~Cincio, ``Error mitigation with {C}lifford quantum-circuit data,'' \emph{{Quantum}}, vol.~5, p. 592, nov 2021. [Online]. Available: \url{https://doi.org/10.22331/q-2021-11-26-592}
\BIBentrySTDinterwordspacing

\bibitem{Lloyd04}
\BIBentryALTinterwordspacing
S.~Lloyd, A.~J. Landahl, and J.-J.~E. Slotine, ``Universal quantum interfaces,'' \emph{Phys. Rev. A}, vol.~69, p. 012305, Jan 2004. [Online]. Available: \url{https://link.aps.org/doi/10.1103/PhysRevA.69.012305}
\BIBentrySTDinterwordspacing

\bibitem{Burgarth07}
\BIBentryALTinterwordspacing
D.~Burgarth and V.~Giovannetti, ``Full control by locally induced relaxation,'' \emph{Phys. Rev. Lett.}, vol.~99, p. 100501, Sep 2007. [Online]. Available: \url{https://link.aps.org/doi/10.1103/PhysRevLett.99.100501}
\BIBentrySTDinterwordspacing

\bibitem{Burgarth09}
\BIBentryALTinterwordspacing
D.~Burgarth, S.~Bose, C.~Bruder, and V.~Giovannetti, ``Local controllability of quantum networks,'' \emph{Phys. Rev. A}, vol.~79, p. 060305, Jun 2009. [Online]. Available: \url{https://link.aps.org/doi/10.1103/PhysRevA.79.060305}
\BIBentrySTDinterwordspacing

\bibitem{Kay10}
\BIBentryALTinterwordspacing
A.~Kay and P.~J. Pemberton-Ross, ``Computation on spin chains with limited access,'' \emph{Phys. Rev. A}, vol.~81, p. 010301, Jan 2010. [Online]. Available: \url{https://link.aps.org/doi/10.1103/PhysRevA.81.010301}
\BIBentrySTDinterwordspacing

\bibitem{Burgarth10}
\BIBentryALTinterwordspacing
D.~Burgarth, K.~Maruyama, M.~Murphy, S.~Montangero, T.~Calarco, F.~Nori, and M.~B. Plenio, ``Scalable quantum computation via local control of only two qubits,'' \emph{Phys. Rev. A}, vol.~81, p. 040303, Apr 2010. [Online]. Available: \url{https://link.aps.org/doi/10.1103/PhysRevA.81.040303}
\BIBentrySTDinterwordspacing

\bibitem{Heule10}
\BIBentryALTinterwordspacing
R.~Heule, C.~Bruder, D.~Burgarth, and V.~M. Stojanovic, ``Local quantum control of heisenberg spin chains,'' \emph{Phys. Rev. A}, vol.~82, p. 052333, Nov 2010. [Online]. Available: \url{https://link.aps.org/doi/10.1103/PhysRevA.82.052333}
\BIBentrySTDinterwordspacing

\bibitem{Maruyama2016}
\BIBentryALTinterwordspacing
K.~Maruyama and D.~Burgarth, \emph{Gateway Schemes of Quantum Control for Spin Networks}.\hskip 1em plus 0.5em minus 0.4em\relax New York, NY: Springer New York, 2016, pp. 167--192. [Online]. Available: \url{https://doi.org/10.1007/978-1-4939-3658-8_6}
\BIBentrySTDinterwordspacing

\bibitem{Gokler17}
\BIBentryALTinterwordspacing
C.~Gokler, S.~Lloyd, P.~Shor, and K.~Thompson, ``Efficiently controllable graphs,'' \emph{Phys. Rev. Lett.}, vol. 118, p. 260501, Jun 2017. [Online]. Available: \url{https://link.aps.org/doi/10.1103/PhysRevLett.118.260501}
\BIBentrySTDinterwordspacing

\bibitem{KATO2020168046}
\BIBentryALTinterwordspacing
G.~Kato, M.~Owari, and K.~Maruyama, ``Algebra and hilbert space structures induced by quantum probes,'' \emph{Annals of Physics}, vol. 412, p. 168046, 2020. [Online]. Available: \url{https://www.sciencedirect.com/science/article/pii/S000349161930301X}
\BIBentrySTDinterwordspacing

\bibitem{indirect}
\BIBentryALTinterwordspacing
G.~Kato, ``Quantum information processing via indirect quantum control,'' \emph{NTT Technical Review}, vol.~18, no.~11, pp. 32--35, 2020. [Online]. Available: \url{https://cir.nii.ac.jp/crid/1390577763022346752}
\BIBentrySTDinterwordspacing

\bibitem{Burgarth09b}
\BIBentryALTinterwordspacing
D.~Burgarth, K.~Maruyama, and F.~Nori, ``Coupling strength estimation for spin chains despite restricted access,'' \emph{Phys. Rev. A}, vol.~79, p. 020305, Feb 2009. [Online]. Available: \url{https://link.aps.org/doi/10.1103/PhysRevA.79.020305}
\BIBentrySTDinterwordspacing

\bibitem{Franco09}
\BIBentryALTinterwordspacing
C.~Di~Franco, M.~Paternostro, and M.~S. Kim, ``Hamiltonian tomography in an access-limited setting without state initialization,'' \emph{Phys. Rev. Lett.}, vol. 102, p. 187203, May 2009. [Online]. Available: \url{https://link.aps.org/doi/10.1103/PhysRevLett.102.187203}
\BIBentrySTDinterwordspacing

\bibitem{Burgarth09c}
\BIBentryALTinterwordspacing
D.~Burgarth and K.~Maruyama, ``Indirect hamiltonian identification through a small gateway,'' \emph{New Journal of Physics}, vol.~11, no.~10, p. 103019, oct 2009. [Online]. Available: \url{https://dx.doi.org/10.1088/1367-2630/11/10/103019}
\BIBentrySTDinterwordspacing

\bibitem{Kato14}
\BIBentryALTinterwordspacing
Y.~Kato and N.~Yamamoto, ``Structure identification and state initialization of spin networks with limited access,'' \emph{New Journal of Physics}, vol.~16, no.~2, p. 023024, feb 2014. [Online]. Available: \url{https://dx.doi.org/10.1088/1367-2630/16/2/023024}
\BIBentrySTDinterwordspacing

\bibitem{Zhang14}
\BIBentryALTinterwordspacing
J.~Zhang and M.~Sarovar, ``Quantum hamiltonian identification from measurement time traces,'' \emph{Phys. Rev. Lett.}, vol. 113, p. 080401, Aug 2014. [Online]. Available: \url{https://link.aps.org/doi/10.1103/PhysRevLett.113.080401}
\BIBentrySTDinterwordspacing

\bibitem{Sone17}
\BIBentryALTinterwordspacing
A.~Sone and P.~Cappellaro, ``Hamiltonian identifiability assisted by a single-probe measurement,'' \emph{Phys. Rev. A}, vol.~95, p. 022335, Feb 2017. [Online]. Available: \url{https://link.aps.org/doi/10.1103/PhysRevA.95.022335}
\BIBentrySTDinterwordspacing

\bibitem{Burgarth12}
\BIBentryALTinterwordspacing
D.~Burgarth and K.~Yuasa, ``Quantum system identification,'' \emph{Phys. Rev. Lett.}, vol. 108, p. 080502, Feb 2012. [Online]. Available: \url{https://link.aps.org/doi/10.1103/PhysRevLett.108.080502}
\BIBentrySTDinterwordspacing

\bibitem{Owari15}
\BIBentryALTinterwordspacing
M.~Owari, K.~Maruyama, T.~Takui, and G.~Kato, ``Probing an untouchable environment for its identification and control,'' \emph{Phys. Rev. A}, vol.~91, p. 012343, Jan 2015. [Online]. Available: \url{https://link.aps.org/doi/10.1103/PhysRevA.91.012343}
\BIBentrySTDinterwordspacing

\bibitem{anan_indirect_vqe}
T.~Anan, M.~Honda, Y.~Nakata, and M.~Owari, ``Implementation of variational quantum eigensolver with indirect control,'' in \emph{24th Asian Quantum Information Science Conference (AQIS24)}, Sapporo, aug 2024.

\bibitem{Cao2019}
\BIBentryALTinterwordspacing
Y.~Cao, J.~Romero, J.~P. Olson, M.~Degroote, P.~D. Johnson, M.~Kieferov{\'a}, I.~D. Kivlichan, T.~Menke, B.~Peropadre, N.~P.~D. Sawaya, S.~Sim, L.~Veis, and A.~Aspuru-Guzik, ``Quantum chemistry in the age of quantum computing,'' \emph{Chemical Reviews}, vol. 119, no.~19, pp. 10\,856--10\,915, Oct 2019. [Online]. Available: \url{https://doi.org/10.1021/acs.chemrev.8b00803}
\BIBentrySTDinterwordspacing

\bibitem{RevModPhys.92.015003}
\BIBentryALTinterwordspacing
S.~McArdle, S.~Endo, A.~Aspuru-Guzik, S.~C. Benjamin, and X.~Yuan, ``Quantum computational chemistry,'' \emph{Rev. Mod. Phys.}, vol.~92, p. 015003, Mar 2020. [Online]. Available: \url{https://link.aps.org/doi/10.1103/RevModPhys.92.015003}
\BIBentrySTDinterwordspacing

\bibitem{PhysRevA.92.042303}
\BIBentryALTinterwordspacing
D.~Wecker, M.~B. Hastings, and M.~Troyer, ``Progress towards practical quantum variational algorithms,'' \emph{Phys. Rev. A}, vol.~92, p. 042303, Oct 2015. [Online]. Available: \url{https://link.aps.org/doi/10.1103/PhysRevA.92.042303}
\BIBentrySTDinterwordspacing

\bibitem{Kokail2019}
\BIBentryALTinterwordspacing
C.~Kokail, C.~Maier, R.~van Bijnen, T.~Brydges, M.~K. Joshi, P.~Jurcevic, C.~A. Muschik, P.~Silvi, R.~Blatt, C.~F. Roos, and P.~Zoller, ``Self-verifying variational quantum simulation of lattice models,'' \emph{Nature}, vol. 569, no. 7756, pp. 355--360, May 2019. [Online]. Available: \url{https://doi.org/10.1038/s41586-019-1177-4}
\BIBentrySTDinterwordspacing

\bibitem{McClean2018}
\BIBentryALTinterwordspacing
J.~R. McClean, S.~Boixo, V.~N. Smelyanskiy, R.~Babbush, and H.~Neven, ``Barren plateaus in quantum neural network training landscapes,'' \emph{Nature Communications}, vol.~9, no.~1, p. 4812, Nov 2018. [Online]. Available: \url{https://doi.org/10.1038/s41467-018-07090-4}
\BIBentrySTDinterwordspacing

\bibitem{Cerezo2021}
\BIBentryALTinterwordspacing
M.~Cerezo, A.~Sone, T.~Volkoff, L.~Cincio, and P.~J. Coles, ``Cost function dependent barren plateaus in shallow parametrized quantum circuits,'' \emph{Nature Communications}, vol.~12, no.~1, p. 1791, Mar 2021. [Online]. Available: \url{https://doi.org/10.1038/s41467-021-21728-w}
\BIBentrySTDinterwordspacing

\bibitem{Wang2021}
\BIBentryALTinterwordspacing
S.~Wang, E.~Fontana, M.~Cerezo, K.~Sharma, A.~Sone, L.~Cincio, and P.~J. Coles, ``Noise-induced barren plateaus in variational quantum algorithms,'' \emph{Nature Communications}, vol.~12, no.~1, p. 6961, Nov 2021. [Online]. Available: \url{https://doi.org/10.1038/s41467-021-27045-6}
\BIBentrySTDinterwordspacing

\bibitem{PhysRevLett.127.120502}
\BIBentryALTinterwordspacing
L.~Bittel and M.~Kliesch, ``Training variational quantum algorithms is np-hard,'' \emph{Phys. Rev. Lett.}, vol. 127, p. 120502, Sep 2021. [Online]. Available: \url{https://link.aps.org/doi/10.1103/PhysRevLett.127.120502}
\BIBentrySTDinterwordspacing

\bibitem{Cao_2025}
\BIBentryALTinterwordspacing
C.~Cao, Y.~Zhou, S.~Tannu, N.~Shannon, and R.~Joynt, ``Exploiting many-body localization for scalable variational quantum simulation,'' \emph{Quantum}, vol.~9, p. 1942, Dec. 2025. [Online]. Available: \url{http://dx.doi.org/10.22331/q-2025-12-12-1942}
\BIBentrySTDinterwordspacing

\bibitem{Grant_2019}
\BIBentryALTinterwordspacing
E.~Grant, L.~Wossnig, M.~Ostaszewski, and M.~Benedetti, ``An initialization strategy for addressing barren plateaus in parametrized quantum circuits,'' \emph{Quantum}, vol.~3, p. 214, Dec. 2019. [Online]. Available: \url{http://dx.doi.org/10.22331/q-2019-12-09-214}
\BIBentrySTDinterwordspacing

\bibitem{Kurita_2023}
\BIBentryALTinterwordspacing
T.~Kurita, H.~Qassim, M.~Ishii, H.~Oshima, S.~Sato, and J.~Emerson, ``Synergetic quantum error mitigation by randomized compiling and zero-noise extrapolation for the variational quantum eigensolver,'' \emph{Quantum}, vol.~7, p. 1184, nov 2023. [Online]. Available: \url{http://dx.doi.org/10.22331/q-2023-11-20-1184}
\BIBentrySTDinterwordspacing

\bibitem{PhysRevResearch.6.033069}
\BIBentryALTinterwordspacing
T.~Jiang, J.~Rogers, M.~S. Frank, O.~Christiansen, Y.-X. Yao, and N.~Lanat{\`{a}}, ``Error mitigation in variational quantum eigensolvers using tailored probabilistic machine learning,'' \emph{Phys. Rev. Res.}, vol.~6, p. 033069, Jul 2024. [Online]. Available: \url{https://link.aps.org/doi/10.1103/PhysRevResearch.6.033069}
\BIBentrySTDinterwordspacing

\bibitem{PhysRevLett.120.210501}
\BIBentryALTinterwordspacing
E.~F. Dumitrescu, A.~J. McCaskey, G.~Hagen, G.~R. Jansen, T.~D. Morris, T.~Papenbrock, R.~C. Pooser, D.~J. Dean, and P.~Lougovski, ``Cloud quantum computing of an atomic nucleus,'' \emph{Phys. Rev. Lett.}, vol. 120, p. 210501, May 2018. [Online]. Available: \url{https://link.aps.org/doi/10.1103/PhysRevLett.120.210501}
\BIBentrySTDinterwordspacing

\bibitem{digitalZne}
T.~Giurgica-Tiron, Y.~Hindy, R.~LaRose, A.~Mari, and W.~J. Zeng, ``Digital zero noise extrapolation for quantum error mitigation,'' in \emph{2020 IEEE International Conference on Quantum Computing and Engineering (QCE)}, 2020, pp. 306--316.

\bibitem{PhysRevA.102.012426}
\BIBentryALTinterwordspacing
A.~He, B.~Nachman, W.~A. de~Jong, and C.~W. Bauer, ``Zero-noise extrapolation for quantum-gate error mitigation with identity insertions,'' \emph{Phys. Rev. A}, vol. 102, p. 012426, Jul 2020. [Online]. Available: \url{https://link.aps.org/doi/10.1103/PhysRevA.102.012426}
\BIBentrySTDinterwordspacing

\bibitem{doi:10.1098/rsta.1911.0009}
\BIBentryALTinterwordspacing
L.~F. Richardson and R.~T. Glazebrook, ``Ix. the approximate arithmetical solution by finite differences of physical problems involving differential equations, with an application to the stresses in a masonry dam,'' \emph{Philosophical Transactions of the Royal Society of London. Series A, Containing Papers of a Mathematical or Physical Character}, vol. 210, no. 459-470, pp. 307--357, 1911. [Online]. Available: \url{https://royalsocietypublishing.org/doi/abs/10.1098/rsta.1911.0009}
\BIBentrySTDinterwordspacing

\bibitem{Menta_2026}
\BIBentryALTinterwordspacing
R.~Menta, F.~Cioni, R.~Aiudi, F.~Caravelli, M.~Polini, and V.~Giovannetti, ``Building globally controlled quantum processors with zz interactions,'' \emph{Physical Review A}, vol. 113, no.~1, Jan. 2026. [Online]. Available: \url{http://dx.doi.org/10.1103/lz5d-lnz2}
\BIBentrySTDinterwordspacing

\bibitem{PhysRevResearch.7.L012065}
\BIBentryALTinterwordspacing
R.~Menta, F.~Cioni, R.~Aiudi, M.~Polini, and V.~Giovannetti, ``Globally driven superconducting quantum computing architecture,'' \emph{Phys. Rev. Res.}, vol.~7, p. L012065, Mar 2025. [Online]. Available: \url{https://link.aps.org/doi/10.1103/PhysRevResearch.7.L012065}
\BIBentrySTDinterwordspacing

\bibitem{Suzuki_2021}
\BIBentryALTinterwordspacing
Y.~Suzuki, Y.~Kawase, Y.~Masumura, Y.~Hiraga, M.~Nakadai, J.~Chen, K.~M. Nakanishi, K.~Mitarai, R.~Imai, S.~Tamiya, T.~Yamamoto, T.~Yan, T.~Kawakubo, Y.~O. Nakagawa, Y.~Ibe, Y.~Zhang, H.~Yamashita, H.~Yoshimura, A.~Hayashi, and K.~Fujii, ``Qulacs: a fast and versatile quantum circuit simulator for research purpose,'' \emph{Quantum}, vol.~5, p. 559, oct 2021. [Online]. Available: \url{http://dx.doi.org/10.22331/q-2021-10-06-559}
\BIBentrySTDinterwordspacing

\bibitem{kraft1988software}
\BIBentryALTinterwordspacing
D.~Kraft and D.~F. und Versuchsanstalt f{\"u}r Luft-und Raumfahrt~K{\"o}ln, \emph{A software package for sequential quadratic programming}, ser. Forschungsbericht / Deutsche Forschungs- und Versuchsanstalt f{\"u}r Luft- und Raumfahrt.\hskip 1em plus 0.5em minus 0.4em\relax Wissenschaftl. Berichtswesen d. DFVLR, 1988. [Online]. Available: \url{https://books.google.co.jp/books?id=T-YYzwEACAAJ}
\BIBentrySTDinterwordspacing

\bibitem{nocedal2006numerical}
J.~Nocedal and S.~J. Wright, \emph{Numerical Optimization}.\hskip 1em plus 0.5em minus 0.4em\relax New York: Springer, 2006.

\bibitem{lawson1995solving}
\BIBentryALTinterwordspacing
C.~Lawson and R.~Hanson, \emph{Solving Least Squares Problems}, ser. Classics in Applied Mathematics.\hskip 1em plus 0.5em minus 0.4em\relax Society for Industrial and Applied Mathematics, 1995. [Online]. Available: \url{https://books.google.co.jp/books?id=AEwDbHp50FgC}
\BIBentrySTDinterwordspacing

\bibitem{breuer2002theory}
H.-P. Breuer and F.~Petruccione, \emph{The theory of open quantum systems}.\hskip 1em plus 0.5em minus 0.4em\relax OUP Oxford, 2002.

\bibitem{Trotter59}
\BIBentryALTinterwordspacing
H.~F. Trotter, ``On the product of semi-groups of operators,'' \emph{Proceedings of the American Mathematical Society}, vol.~10, no.~4, pp. 545--551, 1959. [Online]. Available: \url{http://www.jstor.org/stable/2033649}
\BIBentrySTDinterwordspacing

\bibitem{lre}
V.~Russo and A.~Mari, ``Quantum error mitigation by layerwise richardson extrapolation,'' \emph{arXiv:2402.04000}, 2024.

\bibitem{van_den_Berg_2023}
\BIBentryALTinterwordspacing
E.~van~den Berg, Z.~K. Minev, A.~Kandala, and K.~Temme, ``Probabilistic error cancellation with sparse pauli–lindblad models on noisy quantum processors,'' \emph{Nature Physics}, vol.~19, no.~8, p. 1116–1121, May 2023. [Online]. Available: \url{http://dx.doi.org/10.1038/s41567-023-02042-2}
\BIBentrySTDinterwordspacing

\bibitem{Murali_2020}
\BIBentryALTinterwordspacing
P.~Murali, D.~C. Mckay, M.~Martonosi, and A.~Javadi-Abhari, ``Software mitigation of crosstalk on noisy intermediate-scale quantum computers,'' in \emph{Proceedings of the Twenty-Fifth International Conference on Architectural Support for Programming Languages and Operating Systems}, ser. ASPLOS ’20.\hskip 1em plus 0.5em minus 0.4em\relax ACM, Mar. 2020, p. 1001–1016. [Online]. Available: \url{http://dx.doi.org/10.1145/3373376.3378477}
\BIBentrySTDinterwordspacing

\bibitem{Brown_2019}
\BIBentryALTinterwordspacing
N.~C. Brown and K.~R. Brown, ``Leakage mitigation for quantum error correction using a mixed qubit scheme,'' \emph{Physical Review A}, vol. 100, no.~3, 2019. [Online]. Available: \url{http://dx.doi.org/10.1103/PhysRevA.100.032325}
\BIBentrySTDinterwordspacing

\bibitem{battistel_mitigating_leakage}
F.~Battistel, ``\BIBforeignlanguage{English}{Mitigating leakage and noise in superconducting quantum computing},'' Dissertation (TU Delft), Delft University of Technology, 2022.

\bibitem{xin2025improvederrorcorrectionleakage}
\BIBentryALTinterwordspacing
Y.~Xin, S.~L.~M. van~der Meer, M.~Serra-Peralta, T.~H.~F. Vroomans, M.~Finkel, H.~M. Veen, M.~W. Beekman, and L.~DiCarlo, ``Improved error correction with leakage reduction units built into qubit measurement in a superconducting quantum processor,'' 2025. [Online]. Available: \url{https://arxiv.org/abs/2511.17460}
\BIBentrySTDinterwordspacing

\bibitem{PhysRevA.91.012330}
\BIBentryALTinterwordspacing
J.~Zhang, D.~Burgarth, R.~Laflamme, and D.~Suter, ``Experimental implementation of quantum gates through actuator qubits,'' \emph{Phys. Rev. A}, vol.~91, p. 012330, Jan 2015. [Online]. Available: \url{https://link.aps.org/doi/10.1103/PhysRevA.91.012330}
\BIBentrySTDinterwordspacing

\bibitem{diamondnv}
\BIBentryALTinterwordspacing
X.~Zhu, S.~Saito, A.~Kemp, K.~Kakuyanagi, S.-i. Karimoto, H.~Nakano, W.~J. Munro, Y.~Tokura, M.~S. Everitt, K.~Nemoto, M.~Kasu, N.~Mizuochi, and K.~Semba, ``Coherent coupling of a superconducting flux qubit to an electron spin ensemble in diamond,'' \emph{Nature}, vol. 478, no. 7368, pp. 221--224, oct 2011. [Online]. Available: \url{http://dx.doi.org/10.1038/nature10462}
\BIBentrySTDinterwordspacing

\bibitem{Qi_2026}
\BIBentryALTinterwordspacing
J.~Qi, C.-H.~H. Yang, P.-Y. Chen, and M.-H. Hsieh, ``Tensorhyper-vqc: a tensor-train-guided hypernetwork for robust and scalable variational quantum computing,'' \emph{npj Quantum Information}, vol.~12, no.~1, Feb. 2026. [Online]. Available: \url{http://dx.doi.org/10.1038/s41534-025-01157-z}
\BIBentrySTDinterwordspacing

\bibitem{Gard_2020}
\BIBentryALTinterwordspacing
B.~T. Gard, L.~Zhu, G.~S. Barron, N.~J. Mayhall, S.~E. Economou, and E.~Barnes, ``Efficient symmetry-preserving state preparation circuits for the variational quantum eigensolver algorithm,'' \emph{npj Quantum Information}, vol.~6, no.~1, Jan. 2020. [Online]. Available: \url{http://dx.doi.org/10.1038/s41534-019-0240-1}
\BIBentrySTDinterwordspacing

\bibitem{PhysRevA.101.052340}
\BIBentryALTinterwordspacing
K.~Seki, T.~Shirakawa, and S.~Yunoki, ``Symmetry-adapted variational quantum eigensolver,'' \emph{Phys. Rev. A}, vol. 101, p. 052340, May 2020. [Online]. Available: \url{https://link.aps.org/doi/10.1103/PhysRevA.101.052340}
\BIBentrySTDinterwordspacing

\bibitem{McClean_2018}
\BIBentryALTinterwordspacing
J.~R. McClean, S.~Boixo, V.~N. Smelyanskiy, R.~Babbush, and H.~Neven, ``Barren plateaus in quantum neural network training landscapes,'' \emph{Nature Communications}, vol.~9, no.~1, Nov. 2018. [Online]. Available: \url{http://dx.doi.org/10.1038/s41467-018-07090-4}
\BIBentrySTDinterwordspacing

\bibitem{Saniee}
K.~Saniee, ``A simple expression for multivariate lagrange interpolation,'' \emph{SIAM Undergraduate Research Online}, vol.~1, 01 2008.

\bibitem{Suzuki_private}
Y.~Suzuki, 2023, private communication.

\bibitem{Wierichs_2022}
\BIBentryALTinterwordspacing
D.~Wierichs, J.~Izaac, C.~Wang, and C.~Y.-Y. Lin, ``General parameter-shift rules for quantum gradients,'' \emph{Quantum}, vol.~6, p. 677, Mar. 2022. [Online]. Available: \url{http://dx.doi.org/10.22331/q-2022-03-30-677}
\BIBentrySTDinterwordspacing

\bibitem{Childs2021}
\BIBentryALTinterwordspacing
A.~M. Childs, Y.~Su, M.~C. Tran, N.~Wiebe, and S.~Zhu, ``Theory of trotter error with commutator scaling,'' \emph{Phys. Rev. X}, vol.~11, p. 011020, Feb 2021. [Online]. Available: \url{https://link.aps.org/doi/10.1103/PhysRevX.11.011020}
\BIBentrySTDinterwordspacing

\bibitem{Wang2026}
\BIBentryALTinterwordspacing
X.~Wang, S.~Zhou, X.~Wang, Y.-C. Zheng, S.~Zhang, and T.~Li, ``Lindbladian simulation with commutator bounds,'' 2026. [Online]. Available: \url{https://arxiv.org/abs/2603.28602}
\BIBentrySTDinterwordspacing

\end{thebibliography}

\begin{IEEEbiographynophoto}{ARIJIT DAS}
was born in India. He received the B.Sc. and M.Sc. degrees in physics from Visva-Bharati University, India, in 2021 and 2023, respectively. He is a 2023 Japanese Government MEXT Scholarship recipient and is currently pursuing a Ph.D. in Information Science and Technology at Shizuoka University, Japan.

His research interests include quantum error mitigation, variational quantum eigensolvers, and indirect-control quantum systems. During 2024 and 2025, he held research internship positions with the Computer and Data Science Laboratories at NTT R\&D, Tokyo, and the Processor Research Team at RIKEN Center for Computational Science (R-CCS
), Kobe. His research interests include quantum error mitigation (QEM), specifically Zero-Noise Extrapolation (ZNE), the Variational Quantum Eigensolver (VQE) algorithm, and quantum error mitigation for indirect-control systems. He has presented his work (poster presentations) at the Asian Quantum Information Science Conference (AQIS, 2024) and the Quantum Information Technology Symposium (QIT52, 2025).
\end{IEEEbiographynophoto}

\begin{IEEEbiographynophoto}{Masaki Owari} (M'18) was born in Japan, in 1979. He received the
B.S. degree from the Faculty of Sciences, The University of Tokyo, Japan,
in 2002, and the M.S. and Ph.D. degrees in science from The University of
Tokyo, in 2004 and 2007, respectively. In 2007, he worked with The University
of Tokyo, as a Project Research Associate. From 2007 to 2009, he worked with
the Institute of Mathematical Science and the Blacket Laboratory, Department
of Physics, Imperial College London, as a Research Associate. From 2009 to
2011, he worked with the Institut für Theoretische Physik, Universität Ulm,
as a JSPS Post-Doctoral Fellow for research abroad. From 2011 to 2015,
he worked with the NTT Communication Science Laboratories, NTT Corporation,
as a Scientist. In 2015, he joined the Department of Computer Science,
Faculty of Informatics, Shizuoka University, as an Associate Professor. His
research interests include quantum information, quantum computation and quantum physics.
\end{IEEEbiographynophoto}

\EOD

\end{document}